\documentclass[pdflatex,sn-mathphys-num]{sn-jnl}
\usepackage{graphicx}%
\usepackage{multirow}%
\usepackage{amsmath,amssymb,amsfonts}%
\usepackage{amsthm}%
\usepackage{mathrsfs}%
\usepackage[title]{appendix}%
\usepackage{xcolor}%
\usepackage{textcomp}%
\usepackage{manyfoot}%
\usepackage{booktabs}%
\usepackage{algorithm}%
\usepackage{algorithmicx}%
\usepackage{algpseudocode}%
\usepackage{listings}%
\usepackage{geometry}

\newcommand{\ii}{\mathrm{i}}

\theoremstyle{thmstyleone}%
\newtheorem{thm}{Theorem}
\newtheorem{prop}{Proposition}%
\newtheorem{corollary}{Corollary}
\newtheorem{lem}{Lemma}

\theoremstyle{thmstyletwo}%
\newtheorem{rem}{Remark}%

\theoremstyle{thmstylethree}%
\newtheorem{definition}{Definition}%

\begin{document}

\title{Nonlinear stability of vector multi-solitons in coupled NLS and modified KdV equations}

\author[1]{\fnm{Liming} \sur{Ling}}\email{linglm@scut.edu.cn}
\equalcont{These authors contributed equally to this work.}

\author*[1]{\fnm{Huajie} \sur{Su}}\email{mashj@mail.scut.edu.cn}
\equalcont{These authors contributed equally to this work.}

\affil*[1]{\orgdiv{School of Mathematics}, \orgname{South China University of Technology}, \orgaddress{\city{Guangzhou}, \postcode{510641}, \country{China}}}

\abstract{We prove that the $N$-solitons, including breathers and multi-hump solitons, of the coupled nonlinear Schr\"odinger (CNLS) equations are nonlinearly stable in the Sobolev space $H^{N}$. Moreover, $(N_{1},N_{2})$-solitons of the coupled modified Korteweg--de Vries (CmKdV) equations are shown to be nonlinearly stable in the Sobolev space $H^{2N_{1}+N_{2}}$. The number of negative eigenvalues of the second variation of the Lyapunov functional is $N$ for $N$-solitons of the CNLS equations, and $N_{1}+\lfloor (N_{2}+1)/2 \rfloor$ for $(N_{1},N_{2})$-solitons of the CmKdV equations, which is obtained by exploiting integrable properties. The stability of solitons for the classical NLS and mKdV equations also follows from the same method. In addition, we show that solutions to the linearized spectral problem of the mixed flow equation can be constructed from solutions of the stationary zero curvature equations in a large class of Lie algebras.}

\keywords{Integrable system, Nonlinear stability, NLS equation, mKdV equation. }
\date{\today}
    
\maketitle
\section{Introduction}
In this work, we investigate the nonlinear stability of $N$-soliton solutions, including breathers, 
multi-hump solitons for the coupled nonlinear Schr\"odinger (CNLS) equations 
\cite{berkhoe_self_1970,manako_theory_1974,roskes_nonlinear_1976} on the real line
\begin{equation}\label{CNLS}
    \begin{split}
 {\rm i}q_{1,t}+q_{1,xx}+2(|q_{1}|^{2}+|q_{2}|^{2})q_{1}=&0,\\
 {\rm i}q_{2,t}+q_{2,xx}+2(|q_{1}|^{2}+|q_{2}|^{2})q_{2}=&0,\\
    \end{split}
\end{equation}
where the potentials $q_{1}(x,t),q_{2}(x,t): \mathbb{R}^{2}\to \mathbb{C}$, and the nonlinear stability 
of $(N_{1},N_{2})$-soliton solutions for
the coupled modified Korteweg-de Vries (CmKdV) equations \cite{yajima1975class,athorne1987generalised}
\begin{equation}\label{CmKdV}
    \begin{split}
 q_{1,t}+q_{1,xxx}+6q_{1}^{2}q_{1,x}+3q_{2}(q_{1}q_{2})_{x}=&0, \\
 q_{2,t}+q_{2,xxx}+6q_{2}^{2}q_{2,x}+3q_{1}(q_{1}q_{2})_{x}=&0,
    \end{split}
\end{equation}
where $q_{1}(x,t),q_{2}(x,t): \mathbb{R}^{2}\to \mathbb{R}$. 
The CNLS equations \eqref{CNLS} have important applications in Bose-Einstein condensates \cite{qin_nondegenerate_2019} 
and birefringent fibers \cite{agrawal_nonlinear_2019}, and the CmKdV equations have numerous physical applications
across various fields, including fluid dynamics \cite{su1969korteweg}, plasma physics \cite{cheemaa2020study}, 
and traffic jam \cite{ge2004stabilization,nagatani1999jamming}. The Cauchy problem for 
the CNLS (CmKdV) equations is globally well-posed in the Sobolev space $H^{k}(\mathbb{R})$ for $k\in\mathbb{N}$,  
see \cite{cazenave_introduction_1989,montenegro1995sistemas,huo2006well,corcho2012global,gomes2021solitary,carvajal2019sharp}.

The CNLS equations and the CmKdV equations are integrable and admit Lax pair \cite{manako_theory_1974}, 
bi-Hamiltonian structure \cite{yang_nonlinear_2010}, 
and an infinite set of conservation laws \cite{faddeev1987hamiltonian,yang_nonlinear_2010,koch2018conserved}. 
The spatial part of the Lax pair for the CNLS equations and the CmKdV equations has the representation
\begin{equation}\label{lax-U}
    \mathbf{\Phi}_{x}(\lambda;x,t)=\mathbf{U}(\lambda,\mathbf{q})\mathbf{\Phi}(\lambda;x,t),
\end{equation}
where 
\begin{equation*}
    \begin{split}
        &\mathbf{U}(\lambda;x,t)={\rm i}\lambda\sigma_{3}+\mathbf{Q},\quad 
    \mathbf{Q}(x,t)=\begin{pmatrix}
        0& \mathbf{r}^{T}\\
        \mathbf{q}& 0
    \end{pmatrix},\,\,\,\, 
    \sigma_{3}=\mathrm{diag}(1,-1,-1)
 ,\\ & \mathbf{q}=(q_{1},q_{2})^{T},\quad
    \mathbf{r}=(r_{1},r_{2})^{T}
    \end{split}
\end{equation*}
with the symmetry $\mathbf{r}=-\mathbf{q}^{*}$ for the CNLS equations and the symmetry $\mathbf{r}=-\mathbf{q}$ for 
the CmKdV equations. 
The evolution part of the Lax pair has the representation
\begin{equation}\label{lax-V}
    \mathbf{\Phi}_{t}(\lambda;x,t)=\mathbf{V}(\lambda,\mathbf{q})\mathbf{\Phi}(\lambda;x,t)
\end{equation}
with a different $\mathbf{V}$ matrix for the CNLS equations and the CmKdV equations, and the zero-curvature condition for 
the Lax pair is given by
\begin{equation}\label{zer-cur}
    \mathbf{U}_{t}-\mathbf{V}_{x}+[\mathbf{U},\mathbf{V}]=0
\end{equation}
where the commutator is defined as $[\mathbf{A},\mathbf{B}]=\mathbf{A}\mathbf{B}-\mathbf{B}\mathbf{A}$.
For the CNLS equations, the $\mathbf{V}$ matrix is given by
\begin{equation*}
    \mathbf{V}_{CNLS}(\lambda;x,t)=2{\rm i}\lambda^{2}\sigma_{3}+2\lambda\mathbf{Q}
 +\ii\sigma_{3}(\mathbf{Q}^{2}-\mathbf{Q}_{x})
\end{equation*}
and for the CmKdV equations,
\begin{equation*}
    \mathbf{V}_{CmKdV}(\lambda;x,t)=4{\rm i}\lambda^{3}\sigma_{3}+4\lambda^{2}\mathbf{Q}
 +2{\rm i}\lambda\sigma_{3}(\mathbf{Q}^{2}-\mathbf{Q}_{x})-\mathbf{Q}\mathbf{Q}_{x}+
        \mathbf{Q}_{x}\mathbf{Q}-(\mathbf{Q}_{xx}-2\mathbf{Q}^{3}).
\end{equation*}
The CNLS equations follow from the zero-curvature condition \eqref{zer-cur} 
with $\mathbf{V}=\mathbf{V}_{CNLS}$, and the CmKdV equations follow from \eqref{zer-cur} with $\mathbf{V}=\mathbf{V}_{CmKdV}$.

As extensions of the nonlinear Schr\"odinger (NLS) equation 
\begin{equation}\label{NLS}
    \ii q_{t}+q_{xx}+2|q|^{2}q=0,
\end{equation}
and the modified Korteweg-de Vries (mKdV) equation
\begin{equation}\label{mKdV}
 q_{t}+q_{xxx}+6q^{2}q_{x}=0,
\end{equation} 
the CNLS equations and the CmKdV equations can be used to study the dynamics of vector solitons \cite{qin_nondegenerate_2019}. 
The two-component extension of the NLS equation \eqref{NLS} is of the form \eqref{CNLS} in many papers \cite{cazenave_introduction_1989,ling_darboux_2015}, 
and there are numerous extensions \cite{yajima1975class,athorne1987generalised,iwao1997soliton} of the 
mKdV equation \eqref{mKdV}. The reason why we consider the CmKdV equations of the form \eqref{CmKdV} is that 
the CmKdV equations share the same spatial part of the Lax pair \eqref{lax-U} as the CNLS equations \eqref{CNLS}.

Various solutions of the CNLS equations and the CmKdV equations have been derived by different methods. 
The $N$-soliton solutions were obtained by the inverse scattering method for the CNLS equations \cite{manako_theory_1974} and
the CmKdV equations \cite{wu2017inverse}. 
The Darboux transformation has been used to derive 
non-degenerate solitons \cite{qin_nondegenerate_2019} and breathers \cite{ling_darboux_2015,ling_darboux_2016} for the CNLS equations. 
The breathers of the CmKdV equations and 
non-degenerate solitons have also 
been obtained by the inverse scattering method \cite{wu2017inverse} and the Hirota bilinear method \cite{stalin2020nondegenerate},
respectively. 
$N$-dark-dark solitons have been derived by the KP-hierarchy reduction method \cite{ohta2011general}. 
Bright and dark solitons have also been obtained by the Hirota bilinear 
method \cite{radhakrishnan1995bright} for the CNLS equations.

In our previous work \cite{ling2024stability}, in collaboration with Pelinovsky,  
we established the spectral stability of non-degenerate solitons and the nonlinear stability of non-degenerate  
solitons and breathers.  
In this paper, we develop a novel strategy, fully derived  
from the integrability framework used in the proof of nonlinear stability in \cite{ling2024stability},  
to prove the nonlinear stability of multi-soliton solutions, including breathers and  
multi-hump solitons, for the CNLS and CmKdV equations.  
The stability of soliton solutions for the NLS and mKdV equations can also be obtained by  
the same method, see Remark~\ref{sta-nls-mkdv}.  
Our stability results are associated with the nonlinear stability aspects covered in prior studies, such as the 
$N$-soliton solutions for the NLS and mKdV equations in \cite{kapitula2007stability,le2021stability}, the breathers for the mKdV equation in \cite{alejo2013nonlinear}. 

\subsection{Review of stability results for integrable equations}
The stability question was initially put forward by Boussinesq in the 1870s.  
A pioneering result was obtained by Benjamin in 1972, in which the $H^{1}$ orbital stability of 
solitary waves for the Korteweg--de Vries (KdV) equation was established \cite{benjamin1972stability}.
Subsequently, the stability of ground states for the NLS equation was established by the concentration-compactness principle \cite{cazenave_orbital_1982},  
and the Lyapunov method was applied to certain dispersive equations \cite{cazenave1983stable,weinstein_lyapunov_1986}.  
Moreover, the Lyapunov method was further extended to a broad class of Hamiltonian equations in \cite{grillakis_stability_1987,grillakis_stability_1990}.  

For single-component integrable equations, such as the NLS equation, the KdV equation, and the mKdV equation, 
numerous works have addressed the stability of solitons and breathers by Lyapunov methods. 
The nonlinear stability of $N$-soliton solutions with distinct speeds has been established for various equations, including 
the KdV equation \cite{maddocks1993stability}, a broad class of integrable systems \cite{kapitula2007stability}, 
the derivative nonlinear Schr\"odinger equation \cite{le2018stability}, 
the mKdV equation \cite{le2021stability}, the Camassa--Holm equation \cite{wang2022stability}, and 
the Hirota equation \cite{xiao2023nonlinear}. 
The nonlinear stability of breathers in the mKdV equation was obtained in \cite{alejo2013nonlinear}, 
while the stability of peakons in the modified Camassa--Holm equation was proved in \cite{qu2013stability}. 
Moreover, stability of $N$-soliton solutions has also been achieved in low-regularity spaces 
\cite{killip_orbital_2022,koch2024multisolitons}. 
Since coupled integrable systems admit a richer variety of solutions \cite{qin_nondegenerate_2019}, 
it is natural to investigate the stability of vector solitons. 

For CNLS equations and their nonintegrable extensions, the stability of degenerate vector solitons with single-humped profiles 
was established in \cite{mesentsev1992stability,li_structural_1998,li_mechanism_2000}. 
The stability theory for more general vector solitons was further developed in \cite{pelinovsky_inertia_2005,pelinovsky_instabilities_2005,yagasaki_bifurcations_2023}. 
In contrast, relatively few results are available for the stability of CmKdV equations. 
The stability of multi-solitons with distinct speeds in the two-component Camassa--Holm system 
was investigated in \cite{wu2025stability}. 

The Lyapunov method is a powerful tool for proving nonlinear stability of solutions to differential equations 
\cite{grillakis_stability_1987,grillakis_stability_1990,maddocks1993stability,kapitula2007stability}. 
A key step of the Lyapunov method is to determine the number of negative eigenvalues of the second variation $\mathcal{L}$ of the Lyapunov functional 
and to identify sufficient nonlinear invariants to characterize the kernel of $\mathcal{L}$.
The spectral analysis of the operator $\mathcal{L}$ is crucial but technically challenging. 
In previous studies, for $N$-soliton solutions, even in single-component integrable equations, 
it is typically required that the soliton speeds are distinct, so that as $t\to\infty$, the $N$-soliton 
decomposes into $N$ individual solitons \cite{maddocks1993stability,kapitula2007stability,le2021stability}. 
When some soliton speeds coincide, the solutions become breathers or multi-hump solitons, 
which requires a more refined analysis of $\mathcal{L}$ 
\cite{pelinovsky_inertia_2005,pelinovsky_instabilities_2005,alejo2013nonlinear}.

Due to the integrability of the equations, the squared eigenfunctions are connected to the operator $\mathcal{L}$ for soliton solutions. 
The squared eigenfunctions lie in the kernel of the operator $\partial_{t}-2\mathcal{J}\mathcal{L}$, 
where $\mathcal{J}$ is an auxiliary skew-adjoint operator \cite{yang_nonlinear_2010,deconinck_orbital_2020}. 
For $N$-soliton solutions, the squared eigenfunctions are steady-state solutions that satisfy the spectral problem of the linearized operator $\mathcal{J}\mathcal{L}$ \cite{kapitula2007stability} 
(recall that $N$-soliton solutions are steady-state solutions to the first variation of the Lyapunov functional).  
The negative Krein signature of the operator $\mathcal{L}$ can be obtained from the squared eigenfunctions 
\cite{kapitula_counting_2004,kapitula2007stability,ling2024stability}, since the completeness of the squared eigenfunctions has been established 
\cite{gerdjikov1981generating,yang_nonlinear_2010}. 
The integrability of the equations therefore provides a natural framework for deriving stability results of soliton solutions.  

By combining the Lyapunov method with the integrability of the equations, the stability of $N$-soliton solutions with distinct
speeds has been established for complex potentials \cite{kapitula2007stability,le2021stability}. 
More recently, in collaboration with Pelinovsky, the nonlinear stability of non-degenerate vector solitons and breathers was proved \cite{ling2024stability}, 
where the argument is fully derived from the integrability of the equations. 
For mKdV-type equations, stability results for breathers have also been obtained \cite{su_stability_2025}. 
To the best of our knowledge, there are no systematic studies establishing the stability of $N$-soliton solutions when some speeds coincide, 
which includes the cases of breathers and multi-hump solitons. 
In this paper, the squared eigenfunction method is developed to address this problem for integrable systems, at least within the two-component AKNS framework.  

The method developed in this paper can be extended to establish the nonlinear stability of multi-solitons 
for other integrable equations in the AKNS hierarchy. 
There are, however, essential differences between NLS-type and mKdV-type equations. 
For instance, in the case of the CNLS equation, the associated linearized operator $\mathcal{L}$ involves the complex 
conjugate of the perturbation function, whereas for the CmKdV equation the corresponding 
operator $\tilde{\mathcal{L}}$ contains no conjugate terms. 
In the CmKdV case, the stability analysis reduces to computing the Krein symbol of $\tilde{\mathcal{L}}$, 
which requires more elaborate calculations.  
\subsection{Main results}
The solutions of the CNLS (resp. CmKdV) equations can be regarded as extensions of the solutions of the 
scalar NLS (resp. mKdV) equation, since the first component reduces to the scalar case when the second 
component is identically zero. 
Moreover, for any $\alpha\in[0,2\pi)$, denote $\mathbf{v}_{\alpha}=(\cos \alpha,\sin \alpha)^{T}$. 
Then $q(x,t)$ is a solution of the NLS or mKdV equation if and only if 
$\mathrm{diag}({\rm e}^{\ii \theta_{1}},{\rm e}^{\ii \theta_{2}})q(x,t)\mathbf{v}_{\alpha}$ 
is a solution of the CNLS equation, or $q(x,t)\mathbf{v}_{\alpha}$ is a solution of the CmKdV equation, 
respectively. 

Recall that the Galilean transformation for the CNLS equations \eqref{CNLS} is given by  
\begin{equation*}
 G(a)\mathbf{q}(x,t)={\rm e}^{-2\ii a(x+2at)}\mathbf{q}(x+4at,t)
\end{equation*}
for $a\in \mathbb{R}$
and that the CNLS equations admit the symmetry  
\begin{equation}\label{symmetric}
 T(x_{0},\alpha,\theta_{1},\theta_{2})\mathbf{q}(x,t)=
    \begin{pmatrix}
        {\rm e}^{\ii \theta_{1}} & 0 \\
        0 & {\rm e}^{\ii \theta_{2}}
    \end{pmatrix}
    \begin{pmatrix}
        \cos \alpha & -\sin \alpha \\
        \sin \alpha & \cos \alpha
    \end{pmatrix}
    \mathbf{q}(x+x_{0},t),
\end{equation}
where $\theta_{1},\theta_{2}\in[0,2\pi)$ and $x_{0}\in\mathbb{R}$.  
Under these symmetries, for $b_{1}>0$, the CNLS equations admit the 1-soliton solution  
\begin{equation}\label{1-soliton-CNLS-1}
    \mathbf{q}^{[1]}(x,t;x_{0},\alpha,\theta_{1},\theta_{2})
    =T(x_{0},\alpha,\theta_{1},\theta_{2})\,G(a_{1})\,
    2b_{1}\,\mathrm{sech}(2b_{1}x){\rm e}^{4\ii b_{1}^{2}t}
    \begin{pmatrix}
       1 \\ 0
    \end{pmatrix}.
\end{equation}
The two components of the 1-soliton solution \eqref{1-soliton-CNLS-1} are proportional, and the 1-soliton solution represents 
a traveling wave with speed $-4a_{1}$ and amplitude $2b_{1}$. The 1-soliton solution \eqref{1-soliton-CNLS-1} is a direct extension 
of the 1-soliton solution for the NLS equation \eqref{NLS}, given by 
\begin{equation}\label{NLS-1-soliton}
 q^{[1]}(x,t;x_{0},\theta)=G(a_{1})\,2b_{1}\ \mathrm{sech}(2b_{1}(x+x_{0}))
 {\rm e}^{4\ii b_{1}^{2}t}{\rm e}^{\ii \theta}.
\end{equation}
Similarly, the CmKdV equations admit the symmetry $T(x_{0},\alpha,0,0)$ and possess the 
$(0,1)$-soliton solutions
\begin{equation}\label{cmkdv-0-1-soliton-1}
    \mathbf{q}^{[0,1]}(x,t;x_{0},\alpha)=
    T(x_{0},\alpha,0,0)\,2b_{1}\ \mathrm{sech}\!\left(
        2b_{1}(x-4b_{1}^{2}t)
 \right)\begin{pmatrix}
        1 \\ 0
    \end{pmatrix},
\end{equation}
which can be regarded as an extension of the $(0,1)$-soliton solution for the mKdV equation,
\begin{equation*}
 q^{[0,1]}(x,t;x_{0})=2b_{1}\ \mathrm{sech}\!\left(
        2b_{1}(x-4b_{1}^{2}t+x_{0})
 \right).
\end{equation*}
The 1-soliton solutions and $(0,1)$-soliton solutions are characterized by the spectral parameters 
$a_{1},b_{1}$ or $b_{1}$, together with the scattering parameters 
$(x_{0},\alpha,\theta_{1},\theta_{2})$ or $(x_{0},\alpha)$, respectively. 
The scattering parameters determine the nonlinear orbit of the soliton family.

Denote $\cdot^{T}$ as the transpose of a matrix.
For CmKdV equations, the $(1,0)$-soliton solution is a breather associated with the spectral 
parameters $a_{1},b_{1}$ and has the form 
\begin{equation}\label{CmKdV-1-0-soliton}
    \mathbf{q}^{[1,0]}(x,t)=8b_{1}\mathrm{Re}\left(\frac{
        2\cosh(\eta_{1}){\rm e}^{-\ii \chi_{1}}+\frac{\ii b_{1}}{a_{1}-\ii b_{1}}(
 {\rm e}^{-\eta_{1}-\ii \chi_{1}}+  
 (\mathbf{v}_{\alpha}^{\tilde{\theta}})^{T}\mathbf{v}_{\alpha}^{\tilde{\theta}}{\rm e}^{\eta_{1}+\ii \chi_{1}}
 )
 }{
        2\cosh^{2}(\eta_{1})
 -\frac{b_{1}^{2}}{a_{1}^{2}+b_{1}^{2}}\left| {\rm e}^{-\eta_{1}-\ii \chi_{1}}+(\mathbf{v}_{\alpha}^{\tilde{\theta}})^{T}\mathbf{v}_{\alpha}^{\tilde{\theta}}{\rm e}^{\eta_{1}+\ii \chi_{1}}\right|^{2}
 } \mathbf{v}_{\alpha}^{-\tilde{\theta}}\right),
\end{equation}
where 
\begin{equation*}
    \eta_{1}(x,t)=2b_{1}(x+4(3a_{1}^{2}-b_{1}^{2})t)+x_{1},\qquad
    \chi_{1}(x,t)=2a_{1}(x+4(a_{1}^{2}-3b_{1}^{2})t)+\theta_{1},
\end{equation*}
and 
\begin{equation*}
    \mathbf{v}_{\alpha}^{\tilde{\theta}}=(\cos \alpha,{\rm e}^{-\ii\tilde{\theta}}\sin \alpha )^{T},\quad 
    \tilde{\theta}=\theta_{1}-\theta_{2}. 
\end{equation*}
The $(1,0)$-soliton solution is a nontrivial extension of soliton solutions of the mKdV equation (i.e., it cannot 
be written in the form $q\mathbf{v}_{\alpha}$ with $q$ being a solution of the mKdV equation) when 
$\theta_{1}\ne \theta_{2}$. If $\theta_{1}=\theta_{2}$, then the $(1,0)$-soliton solution reduces to 
$q^{[1,0]}(x,t)\mathbf{v}_{\alpha}$ with 
\begin{equation*}
 q^{[1,0]}(x,t;x_{1},t_{1})=2\partial_{x}\left(\arctan\left(
        \frac{b_{1}}{a_{1}}\frac{\sin (2a_{1}(x+4(a_{1}^{2}-3b_{1}^{2})t)+\theta_{1})}{\cosh(2b_{1}(x+4(3a_{1}^{2}-b_{1}^{2})t)+x_{1})}
 \right)\right),
\end{equation*}
which recovers the breather solution \cite{alejo2013nonlinear} of the mKdV equation and is nonlinearly 
stable under the orbit $\{q^{[1,0]};x_{1},\theta_{1}\in\mathbb{R}\}$, i.e., the invariance under space and time 
translations ($x\mapsto x-x_{0}$ and $t\mapsto t-t_{0}$ with $x_{0},t_{0}\in\mathbb{R}$).  

For fixed spectral parameters, the scattering parameters determine the nonlinear orbit of soliton solutions. 
The soliton solutions for CNLS equations (respectively, CmKdV equations) are specified by spectral 
and scattering parameters in the Darboux transformation, as formulated in the following definition.  

\begin{definition}
 {\rm (a) ($N$-solitons for CNLS equations).}  
 Let the spectral parameter vector be $\mathbf{\Lambda}=(\lambda_{1},\lambda_{2},\cdots,\lambda_{N})$ 
 and the scattering parameter matrix be $\mathbf{c}=(\mathbf{c}_{1},\mathbf{c}_{2},\cdots,\mathbf{c}_{N})$, 
 where $\lambda_{k}\in \mathbb{C}^{+}=\{z\in\mathbb{C}:\mathrm{Im}\,z>0\}$ for $k=1,2,\cdots,N$ are 
 pairwise distinct, and 
 $\mathbf{c}_{k}=(c_{1k},c_{2k})^{T}\in \mathbb{C}^{2}\backslash \{(0,0)\}$.  
 The $N$-soliton solution of the CNLS equation \eqref{CNLS} is given by  
    \begin{equation}\label{CNLS-Nsoliton}
        \mathbf{q}^{[N]}(x,t;\mathbf{\Lambda},\mathbf{c})=-4
        \frac{\det\begin{pmatrix}
        0 & \mathbf{Y}_{2}^{T} \\
        \mathbf{Y}_{1}^{*} & \mathbf{M}
        \end{pmatrix}}{\det(\mathbf{M})},
    \end{equation}
 where  
    \begin{equation*}
        \begin{split}
            \mathbf{Y}_{1}=&(\mathrm{Im}\lambda_{1}{\rm e}^{\ii \lambda_{1}(x+2\lambda_{1}t)}, \mathrm{Im}\lambda_{2}{\rm e}^{\ii \lambda_{2}(x+2\lambda_{2}t)},\cdots, \mathrm{Im}\lambda_{N}{\rm e}^{\ii \lambda_{N}(x+2\lambda_{N}t)})^{T},\\
            \mathbf{Y}_{2}=&({\rm e}^{-\ii \lambda_{1}(x+2\lambda_{1}t)}\mathbf{c}_{1}, {\rm e}^{-\ii \lambda_{2}(x+2\lambda_{2}t)}\mathbf{c}_{2},\cdots, {\rm e}^{-\ii \lambda_{N}(x+2\lambda_{N}t)}\mathbf{c}_{N})^{T},
        \end{split}
    \end{equation*}
 and  
    \begin{equation*}
        \mathbf{M}=\left(
            \frac{\lambda_{k}-\lambda_{k}^{*}}{\lambda_{l}-\lambda_{k}^{*}}\left(
 {\rm e}^{{\rm i}x(\lambda_{l}-\lambda_{k}^{*})
 +2{\rm i}t(\lambda_{l}^{2}-(\lambda_{k}^{*})^{2})}+
                \mathbf{c}_{k}^{\dagger}\mathbf{c}_{l}{\rm e}^{-{\rm i}x(\lambda_{l}-\lambda_{k}^{*})
 -2{\rm i}t(\lambda_{l}^{2}-(\lambda_{k}^{*})^{2})}
 \right)
 \right)_{1\leq k,l\leq N}.
    \end{equation*}

 {\rm (b) ($(N_{1},N_{2})$-solitons for CmKdV equations).}  
 Let $N=N_{1}+N_{2}$ with nonnegative integers $N_{1},N_{2}$ and $\tilde{N}=N+N_{1}$.  
 The spectral parameter vector is $\mathbf{\Lambda}=(\lambda_{1},\lambda_{2},\cdots,\lambda_{N})$, 
 where $\lambda_{k}\in \mathbb{C}^{++}=\{\lambda\in\mathbb{C}^{+}:\mathrm{Re}\lambda>0\}$ for $k=1,2,\cdots,N_{1}$, 
 and $\lambda_{k}\in \mathbb{C}^{+}\cap \ii\mathbb{R}$ for $k=N_{1}+1,\cdots, N$.  
 The scattering parameter matrix is $\mathbf{c}=(\mathbf{c}_{1},\mathbf{c}_{2},\cdots,\mathbf{c}_{N})$, 
 where $\mathbf{c}_{k}\in \mathbb{C}^{2}\backslash \{(0,0)\}$ for $k=1,2,\cdots,N_{1}$, 
 and $\mathbf{c}_{k}\in \mathbb{R}^{2}\backslash \{(0,0)\}$ for $k=N_{1}+1,\cdots, N$.  

 For $k=1,2,\cdots N_{1}$, set $\lambda_{k+N}=-\lambda_{k}^{*}$ and $\mathbf{c}_{k+N}=\mathbf{c}_{k}^{*}$.  
 The $(N_{1},N_{2})$-soliton solution of the CmKdV equation is given by
    \begin{equation}\label{CmKdV-Nsoliton}
        \mathbf{q}^{[N_{1},N_{2}]}(x,t;\mathbf{\Lambda},\mathbf{c})=-4
        \frac{\det\begin{pmatrix}
        0 & \mathbf{Y}_{2}^{T} \\
        \mathbf{Y}_{1}^{*} & \mathbf{M}
        \end{pmatrix}}{\det(\mathbf{M})},
    \end{equation}
 where  
    \begin{equation*}
        \begin{split}
            \mathbf{Y}_{1}=&(\mathrm{Im}\lambda_{1}{\rm e}^{\ii \lambda_{1}(x+4\lambda_{1}^{2}t)}, \mathrm{Im}\lambda_{2}{\rm e}^{\ii \lambda_{2}(x+4\lambda_{2}^{2}t)},\cdots, \mathrm{Im}\lambda_{\tilde{N}}{\rm e}^{\ii \lambda_{\tilde{N}}(x+4\lambda_{\tilde{N}}^{2}t)})^{T},\\
            \mathbf{Y}_{2}=&({\rm e}^{-\ii \lambda_{1}(x+4\lambda_{1}^{2}t)}\mathbf{c}_{1}, {\rm e}^{-\ii \lambda_{2}(x+4\lambda_{2}^{2}t)}\mathbf{c}_{2},\cdots, {\rm e}^{-\ii \lambda_{\tilde{N}}(x+4\lambda_{\tilde{N}}^{2}t)}\mathbf{c}_{\tilde{N}})^{T},
        \end{split}
    \end{equation*}
 and  
    \begin{equation*}
        \mathbf{M}=\left(
            \frac{\lambda_{k}-\lambda_{k}^{*}}{\lambda_{l}-\lambda_{k}^{*}}\left(
 {\rm e}^{{\rm i}x(\lambda_{l}-\lambda_{k}^{*})
 +4{\rm i}t(\lambda_{l}^{3}-(\lambda_{k}^{*})^{3})}+
                \mathbf{c}_{k}^{\dagger}\mathbf{c}_{l}{\rm e}^{-{\rm i}x(\lambda_{l}-\lambda_{k}^{*})
 -4{\rm i}t(\lambda_{l}^{3}-(\lambda_{k}^{*})^{3})}
 \right)
 \right)_{1\leq k,l\leq \tilde{N}}.
    \end{equation*}
\end{definition}

The complicated formulae \eqref{CNLS-Nsoliton} and \eqref{CmKdV-Nsoliton} represent multi-soliton solutions corresponding to \eqref{1-soliton-CNLS-1}, \eqref{cmkdv-0-1-soliton-1}, and \eqref{CmKdV-1-0-soliton}, respectively. 
If $c_{2k}=0$ for all $k$, then the second component of the $N$-soliton solutions vanishes, and the first component reduces to the $N$-soliton solutions of the NLS and mKdV equations. 
Denote the spectral parameters by $\lambda_{k}=a_{k}+\ii b_{k}$, we have $a_{k}\in\mathbb{R}$ and $b_{k}>0$.  

For the CNLS equations, the 1-soliton solution is given by \eqref{CNLS-Nsoliton}:
\begin{equation}\label{1-soliton-CNLS}
    \mathbf{q}^{[1]}(x,t;\lambda_{1},\mathbf{c}_{1})
 =2b_{1}\ \mathrm{sech}\!\left(2b_{1}(x+4a_{1}t)+\ln |\mathbf{c}_{1}|\right)
 {\rm e}^{-2\ii (a_{1}(x+4a_{1}t)-2(a_{1}^{2}+b_{1}^{2})t)}\tilde{\mathbf{c}}_{1},
\end{equation}
where $\tilde{\mathbf{c}}_{1}=\frac{\mathbf{c}_{1}}{|\mathbf{c}_{1}|}$, $a_{1}\in\mathbb{R}$ and $b_{1}>0$.  
By taking 
$c_{11}={\rm e}^{2b_{1}x_{0}} \cos(\alpha) {\rm e}^{\ii \theta_{1}}$ and 
$c_{21}={\rm e}^{2b_{1}x_{0}} \sin(\alpha) {\rm e}^{\ii \theta_{2}}$, 
the solution \eqref{1-soliton-CNLS} reduces to \eqref{1-soliton-CNLS-1}.  
For $2$-soliton solutions, there exist special cases known as non-degenerate vector soliton solutions, which are traveling waves obtained by setting $a_{1}=a_{2}$ together with $c_{12}=c_{21}=0$ or $c_{11}=c_{22}=0$. 
The profile of one component of such a soliton can be either single-humped or double-humped, while the other component is always double-humped \cite{qin_nondegenerate_2019,ling2024stability}. 
In general, the $N$-soliton solutions can be regarded as the nonlinear superposition of $N$ single-soliton solutions \cite{ling_darboux_2016}.  

For the CmKdV equations, the $(N_{1},N_{2})$-soliton solutions can be regarded as the nonlinear superposition of 
$N_{1}$ breathers and $N_{2}$ single-solitons.  
Analogous to the 1-soliton solution for the CNLS equations, the $(0,1)$-soliton solution for the CmKdV equations is given by
\begin{equation}\label{cmkdv-0-1-soliton}
    \mathbf{q}^{[0,1]}(x,t)=2b_{1}\,\mathrm{sech}\!\left(
        2b_{1}(x-4b_{1}^{2}t)+\ln |\mathbf{c}_{1}|
 \right)\tilde{\mathbf{c}}_{1},
\end{equation}
which corresponds to a traveling wave with velocity $4b_{1}^{2}$ and initial position 
$-\ln |\mathbf{c}_{1}|/(2b_{1})$.  
By taking $c_{11}={\rm e}^{2b_{1}x_{0}} \cos(\alpha)$ and $c_{21}={\rm e}^{2b_{1}x_{0}} \sin(\alpha)$,  
the solution \eqref{cmkdv-0-1-soliton} reduces to \eqref{cmkdv-0-1-soliton-1}.  

The $(0,2)$-soliton solution takes the form
\begin{equation*}
    \mathbf{q}^{[0,2]}(x,t)=4\frac{
 b_{1}\left(B{\rm e}^{-\xi_{2}}+{\rm e}^{\xi_{2}}-B_{2}\tilde{c}{\rm e}^{\xi_{1}}\right)\tilde{\mathbf{c}}_{1}+
 b_{2}\left(-B{\rm e}^{-\xi_{1}}+{\rm e}^{\xi_{1}}-B_{1}\tilde{c}{\rm e}^{\xi_{2}}\right)\tilde{\mathbf{c}}_{2}
 }{B^{2}{\rm e}^{-\xi_{1}-\xi_{2}}+\left(1-B_{1}B_{2}\tilde{c}^{2}\right){\rm e}^{\xi_{1}+\xi_{2}}
 +{\rm e}^{\xi_{1}-\xi_{2}}+{\rm e}^{\xi_{2}-\xi_{1}}-2B_{1}B_{2}\tilde{c}},
\end{equation*}
where $\tilde{\mathbf{c}}_{k}=\frac{\mathbf{c}_{k}}{|\mathbf{c}_{k}|},\tilde{c}=\tilde{\mathbf{c}}_{1}^{T}\tilde{\mathbf{c}}_{2}, B_{k}=\frac{2b_{k}}{b_{1}+b_{2}},B=\frac{b_{1}-b_{2}}{b_{1}+b_{2}}$
and $\xi_{k}(x,t)=2b_{k}(x-4b_{k}^{2}t)+\ln |\mathbf{c}_{k}|$ for $k=1,2$, with $b_{1}\ne b_{2}$ and $b_{1},b_{2}>0$. 
Taking
\begin{equation*}
    \mathbf{c}_{1}=B{\rm e}^{x_{1}}\mathbf{v}_{\alpha_{1}},\quad 
    \mathbf{c}_{2}=-B{\rm e}^{x_{2}}\mathbf{v}_{\alpha_{2}},
\end{equation*}
we obtain
\begin{equation*}
    \mathbf{q}^{[0,2]}(x,t)=4\frac{
 b_{1}\left({\rm e}^{-\xi_{2}}+B\left({\rm e}^{\xi_{2}}+B_{2}k_{1}{\rm e}^{\xi_{1}}\right)\right)\mathbf{v}_{\alpha_{1}}+
 b_{2}\left({\rm e}^{-\xi_{1}}-B\left({\rm e}^{\xi_{1}}+B_{1}k_{1}{\rm e}^{\xi_{2}}\right)\right)\mathbf{v}_{\alpha_{2}}
 }{{\rm e}^{-\xi_{1}-\xi_{2}}+B^{2}\left(1-B_{1}B_{2}k_{1}^{2}\right){\rm e}^{\xi_{1}+\xi_{2}}
 +{\rm e}^{\xi_{1}-\xi_{2}}+{\rm e}^{\xi_{2}-\xi_{1}}+2B_{1}B_{2}k_{1}},
\end{equation*}
where $\xi_{k}=2b_{k}(x-4b_{k}^{2}t)+x_{k}$ and $k_{1}=k_{1}(\alpha_{1},\alpha_{2})=\cos (\alpha_{1}-\alpha_{2})$.  
The $(0,2)$-soliton is nontrivial if $\alpha_{1}\ne \alpha_{2}$.  
If $\alpha_{1}=\alpha_{2}=\alpha$, then it degenerates into $\mathbf{q}^{[0,2]}(x,t)=q^{[0,2]}\mathbf{v}_{\alpha}$ with
\begin{equation*}
 q^{[0,2]}=2\partial_{x}\left(
        \mathrm{arctan}\left(
            \frac{{\rm e}^{\xi_{1}}+{\rm e}^{\xi_{2}}}{1-B^{2}{\rm e}^{\xi_{1}+\xi_{2}}}
 \right)
 \right).
\end{equation*}
The speeds of $(0,N_{2})$-solitons are all positive and mutually distinct, which is in contrast to the $N$-soliton solutions of the CNLS equations.

Denote $\cdot^{\dagger}$ as the conjugate transpose of a matrix (or the adjoint of an operator), and $\cdot^{*}$ as the complex conjugate.
The $(1,0)$-soliton \eqref{CmKdV-1-0-soliton} can be obtained from \eqref{CmKdV-Nsoliton} as  
\begin{equation*}
    \mathbf{q}^{[1,0]}(x,t)=8b_{1}\mathrm{Re}\left(\frac{
        2\cosh(\eta_{1}){\rm e}^{-\ii \chi_{1,0}}+\tfrac{\ii b_{1}}{a_{1}-\ii b_{1}}\big(
 {\rm e}^{-\eta_{1}-\ii\chi_{1,0}}+\tilde{\mathbf{c}}_{1}^{\dagger}\tilde{\mathbf{c}}_{1}^{*}{\rm e}^{\eta_{1}+\ii\chi_{1,0}}
 )
 }{
        2\cosh\!\big(2\eta_{1}\big)+2
 -\tfrac{b_{1}^{2}}{a_{1}^{2}+b_{1}^{2}}\left| {\rm e}^{-\eta_{1}-\ii\chi_{1,0}}+\tilde{\mathbf{c}}_{1}^{\dagger}\tilde{\mathbf{c}}_{1}^{*}{\rm e}^{\eta_{1}+\ii\chi_{1,0}}\right|^{2}
 }\tilde{\mathbf{c}}_{1}\right),
\end{equation*}
where 
\[
    \eta_{1}(x,t)=2b_{1}\big(x+4(3a_{1}^{2}-b_{1}^{2})t\big)+\ln |\mathbf{c}_{1}|,\quad 
    \chi_{1,0}(x,t)=2a_{1}\big(x+4(a_{1}^{2}-3b_{1}^{2})t\big),
\]
by taking
\begin{equation*}
    \mathbf{c}_{1}={\rm e}^{x_{1}}
    \begin{pmatrix}
        \cos (\alpha_{1}) {\rm e}^{-\ii \theta_{1}} \\
        \sin (\alpha_{1}) {\rm e}^{-\ii \theta_{2}}
    \end{pmatrix}. 
\end{equation*}
More generally, the $(N_{1},0)$-soliton can be viewed as the nonlinear superposition of $N_{1}$ breathers.  
The propagation speed of the $k$-th breather is $-4(3a_{k}^{2}-b_{k}^{2})$ for $k=1,2,\dots,N_{1}$,  
which may take negative values.  
Consequently, the $(N_{1},N_{2})$-soliton solution of the CmKdV equations can be regarded as the nonlinear superposition of  
$N_{1}$ breathers, propagating either to the left or to the right, and $N_{2}$ single-solitons with mutually distinct positive speeds,  
all propagating to the right.  
Examples of interactions between breathers and single-solitons can be found in \cite{wu2017inverse}.

The scattering parameters $\mathbf{c}_{k}$ can be regarded as an extension of the symmetry. 
For $N$-soliton solutions of the CNLS equations, since $\mathbf{c}_{k}\in \mathbb{C}^{2}\setminus \{(0,0)\}$, 
one can set 
\begin{equation*}
    \mathbf{c}_{k}={\rm e}^{2 b_{k} x_{k}}
    \begin{pmatrix}
        {\rm e}^{\ii \theta_{1k}} & 0 \\
        0 & {\rm e}^{\ii \theta_{2k}}
    \end{pmatrix}\mathbf{v}_{\alpha_{k}}, 
\end{equation*}
where the parameters $x_{k}, \theta_{1k}, \theta_{2k}, \alpha_{k}$ can be interpreted 
as extensions of the underlying symmetries, corresponding to the $k$-th soliton or breather. 
In particular, $x_{k}$ corresponds to spatial translation, 
$\theta_{1k}$ and $\theta_{2k}$ to the phase shifts of the first and second components, respectively, 
and $\alpha_{k}$ to rotational transformation.
For $(N_{1},N_{2})$-soliton solutions, by the definition \eqref{CmKdV-Nsoliton}, we also set
\begin{equation*}
    \mathbf{c}_{k}=\begin{cases}
        {\rm e}^{x_{k}}
        \begin{pmatrix}
            {\rm e}^{\ii \theta_{1i}} & 0 \\
            0 & {\rm e}^{\ii \theta_{2i}}
        \end{pmatrix}\mathbf{v}_{\alpha_{k}}, & k=1,2,\dots,N_{1}, \\[1ex]
        {\rm e}^{x_{k}}\mathbf{v}_{\alpha_{k}}, & k=N_{1}+1,N_{1}+2,\dots,N,
    \end{cases}
\end{equation*}
so that $(x_{k},\theta_{1k})$ correspond to the spatial and temporal translations of the breathers.

The main result of this paper is the nonlinear stability of soliton solutions:
\begin{thm}\label{thm-stability}
 The $N$-soliton solutions \eqref{CNLS-Nsoliton} for CNLS equations are nonlinearly stable
 in the Sobolev space $H^{N}$, and the $(N_{1},N_{2})$-soliton solutions \eqref{CmKdV-Nsoliton} 
 for CmKdV equations are nonlinearly stable in $H^{2N_{1}+N_{2}}$.
 Denote $\tilde{N}=N$ and $\mathbf{q}_{sol}=\mathbf{q}^{[N]}$ for CNLS equations, and 
 $\tilde{N}=2N_{1}+N_{2}$ and $\mathbf{q}_{sol}=\mathbf{q}^{[N_{1},N_{2}]}$ for CmKdV equations. 
 For any initial condition $\mathbf{u}_{0}(x)$ that evolves along the CNLS (or CmKdV) flow, 
 we denote the global solution by $\mathbf{u}(x,t)$. 
 For any positive constant $\epsilon$, there exists $\delta>0$ such that if
    \begin{equation*}
        \|\mathbf{u}_{0}(\cdot)-\mathbf{q}_{sol}(\cdot,0;\mathbf{\Lambda},\mathbf{c}(0))\|_{H^{\tilde{N}}}<\delta
    \end{equation*}
 for some soliton solution with spectral parameters $\mathbf{\Lambda}$ and
 scattering parameters $\mathbf{c}(0)$ such that every column of $\mathbf{c}(0)$ is nonzero, then there exists 
 a $C^{1}$ function $\mathbf{c}(t)$ such that 
    \begin{equation*}
        \|\mathbf{u}(\cdot,t)-\mathbf{q}_{sol}(\cdot,t;\mathbf{\Lambda},\mathbf{c}(t))\|_{H^{\tilde{N}}}<\epsilon
    \end{equation*}
 for all $t\in\mathbb{R}$. Moreover, the rate of change of the scattering parameters can be controlled by $\epsilon$:
    \begin{equation*}
        \sum_{i,j}|\partial_{t}c_{ij}(t)|\leq C\epsilon
    \end{equation*}
 for some constant $C$. 
\end{thm}

\begin{rem}\label{sta-nls-mkdv}
 The nonlinear stability of soliton solutions to the NLS equation and the mKdV equation can also be 
 obtained by the same method with the same Sobolev index as in Theorem \ref{thm-stability}. 
 Denote $q^{[N]}(x,t)$ and $q^{[N_{1},N_{2}]}(x,t)$ the soliton solutions for the NLS equation \eqref{NLS} and 
 the mKdV equation \eqref{mKdV}, respectively (the first component obtained by taking $c_{2k}=0$ in 
 \eqref{CNLS-Nsoliton} and \eqref{CmKdV-Nsoliton}). Then $q^{[N]}(x,t)$ is nonlinearly stable in $H^{N}$ 
 and $q^{[N_{1},N_{2}]}(x,t)$ is nonlinearly stable in $H^{2N_{1}+N_{2}}$. 
 These stability results are consistent with previous studies 
 \cite{kapitula2007stability,alejo2013nonlinear,le2021stability}. 
\end{rem}

As a corollary of Theorem \ref{thm-stability}, we obtain the orbital stability of single soliton solutions. 
\begin{corollary}
 The 1-soliton solutions and $(0,1)$-soliton solutions are orbitally stable in the Sobolev space $H^{1}$. 
 The orbit of 1-soliton solutions is generated by the symmetry 
 $T(x_{0},\alpha,\theta_{1},\theta_{2})$ in \eqref{symmetric}, 
 where $x_{0}$ corresponds to spatial translation, $\alpha$ to rotation, and
 $\theta_{1},\theta_{2}$ to phase translations. 
 The orbit of $(0,1)$-soliton solutions is generated by $T(x_{0},\alpha,0,0)$. 
\end{corollary}

\subsection{Main steps of the proof}
We outline the main steps in the proof of nonlinear stability. 
The integrability of the CNLS and CmKdV equations plays a central role in the argument. 
The nonlinear stability of soliton solutions is established by means of Lyapunov methods 
with tools from integrable systems.

From the spatial part of the Lax pair \eqref{lax-U}, the $n$-th flow equations 
\cite{sattinger1985hamiltonian,faddeev1987hamiltonian,terng1997soliton} can be obtained from the 
infinitely many conserved quantities $\mathcal{H}_{n}:H^{n}\to \mathbb{R}$ ($n\geq 0$) with the Hamiltonian 
operator $-\ii$. These conserved quantities $\mathcal{H}_{n}$ are mutually in involution. For a 
functional $\mathcal{K}(\mathbf{q})$, the first variation is given by
\begin{equation*}
            \left(\mathbf{v}, \frac{\delta \mathcal{K}}{\delta \mathbf{q}}(\mathbf{q})\right)=
            \lim_{\epsilon\to 0}\frac{\mathcal{K}(\mathbf{q}+\epsilon \mathbf{v})-\mathcal{K}(\mathbf{q})}{\epsilon},
\end{equation*}
and the second variation is given by 
\begin{equation*}
\frac{\delta^{2} \mathcal{K}}{\delta^{2} \mathbf{q}}(\mathbf{q})[\mathbf{v}]=
\lim_{\epsilon\to 0}\frac{\frac{\delta \mathcal{K}}{\delta \mathbf{q}}(\mathbf{q}+\epsilon \mathbf{v})-\frac{\delta \mathcal{K}}{\delta \mathbf{q}}(\mathbf{q})}{\epsilon}, 
\end{equation*}
where the inner product 
 \begin{equation*}
        (\mathbf{f},\mathbf{g})=\mathrm{Re}\int_{\mathbb{R}}\mathbf{f}^{\dagger}\mathbf{g}
        \mathrm{d}x. 
    \end{equation*}
With the reduction $\mathbf{r}=-\mathbf{q}^{*}$ in \eqref{lax-U}, the CNLS equations \eqref{CNLS} 
correspond to the second flow
\begin{equation*}
    \mathbf{q}_{t}=-\ii \frac{\delta \mathcal{H}_{2}}{\delta \mathbf{q}}(\mathbf{q}),
\end{equation*}
and the complex CmKdV equations correspond to the third flow
\begin{equation*}
    \mathbf{q}_{t}=-\ii \frac{\delta \mathcal{H}_{3}}{\delta \mathbf{q}}(\mathbf{q}),
\end{equation*}
which reduce to the real CmKdV equations \eqref{CmKdV} under the constraint that $\mathbf{q}$ is real.
The conserved quantities are introduced in detail in the next section 
(see \eqref{conserved-quantities}). 
The first four conserved quantities are
\begin{align}
    \mathcal{H}_{0}=&\frac{1}{2}\int_{\mathbb{R}}|\mathbf{q}|^{2}\,\mathrm{d}x, \label{H0}\\
    \mathcal{H}_{1}=&\frac{1}{2}\int_{\mathbb{R}}\ii \mathbf{q}^{\dagger}\mathbf{q}_{x}\,\mathrm{d}x,\label{H1}\\
    \mathcal{H}_{2}=&\frac{1}{2}\int_{\mathbb{R}}\Big(|\mathbf{q}_{x}|^{2}-|\mathbf{q}|^{4}\Big)\,\mathrm{d}x,\label{H2}\\
    \mathcal{H}_{3}=&\frac{1}{2}\int_{\mathbb{R}}\ii\Big(\mathbf{q}_{x}^{\dagger}\mathbf{q}_{xx}
        + 3|\mathbf{q}|^{2}\mathbf{q}^{\dagger}_{x}\mathbf{q}\Big)\,\mathrm{d}x. \label{H3}
\end{align}
Note that $\mathcal{H}_{2n+1}$ is real for $n\geq 0$ by integration by parts. 
For the CNLS equations, all conserved quantities $\mathcal{H}_{n}$ are nontrivial. 
In contrast, for the CmKdV equations, all momentum-type conserved quantities vanish due to the real potential condition:
\begin{equation*}
    \mathcal{H}_{2n+1}(\mathbf{q})\equiv 0,\quad \mathbf{q}\ \text{real},\quad n\geq 0.
\end{equation*}
Hence, the nontrivial conserved quantities for the CmKdV equations are 
$\mathcal{H}_{2n}(\mathbf{q})$, $n\geq 0$.

The soliton solutions are steady states of the CNLS and CmKdV equations. 
The Lyapunov functional associated with a soliton is derived from the ordinary differential equation (ODE) satisfied by the soliton itself. 
For the $N$-soliton solutions \eqref{CNLS-Nsoliton} of the CNLS equations, the Lyapunov functional is expressed as a special 
linear combination of higher-order conserved quantities: 
\begin{equation}\label{Lya}
    \mathcal{I}(\mathbf{q})=\sum_{n=0}^{2N}\mu_{n} \mathcal{H}_{n}(\mathbf{q}),
\end{equation}
where the $N$-soliton solutions correspond to critical points of $\mathcal{I}$, i.e. 
$\delta \mathcal{I}/\delta \mathbf{q}(\mathbf{q}^{[N]})=0$, which yields an ODE of order $2N$. 
The real coefficients $\mu_{n}$, given by symmetric polynomials of the spectral parameters, can be obtained 
from the trace formula (see Section \ref{var-Nsoliton}).  

Similarly, the $(N_{1},N_{2})$-soliton solutions of the CmKdV equations are critical points of the corresponding Lyapunov functional 
\begin{equation*}
    \tilde{\mathcal{I}}(\mathbf{q})=\sum_{n=0}^{2 \tilde{N}}\tilde{\mu}_{n}\mathcal{H}_{2n}(\mathbf{q}).
\end{equation*}
Since each $\mathcal{H}_{n}$ is time independent, the Lyapunov functional remains constant in time. 
Using the continuity of $\mathcal{H}_{2n-1}$ and $\mathcal{H}_{2n}$ in Sobolev space $H^{n}$, 
the perturbation of the Lyapunov functional can be controlled by the perturbation of the soliton solution: 
\begin{equation*}
    \mathcal{I}(\mathbf{u}(t))-\mathcal{I}(\mathbf{q}^{[N]}(t))
    =\mathcal{I}(\mathbf{u}(0))-\mathcal{I}(\mathbf{q}^{[N]}(0))
    \leq C \|\mathbf{u}(0)-\mathbf{q}^{[N]}(0)\|_{H^{N}}. 
\end{equation*}
Expanding the Lyapunov functional around a soliton solution, the leading term is characterized by the second variation operator $\mathcal{L}$:
\begin{equation*}
    \mathcal{I}(\mathbf{q}^{[N]}+\mathbf{v})
    =\mathcal{I}(\mathbf{q}^{[N]})
    +\frac{1}{2}(\mathcal{L}\mathbf{v},\mathbf{v})
    +\mathcal{O}(\|\mathbf{v}\|_{H^{N}}^{3}). 
\end{equation*}
It is therefore natural to analyze the spectrum of $\mathcal{L}$ in order to understand the quadratic form $(\mathcal{L}\cdot,\cdot)$.  
The main difficulty in establishing nonlinear stability by the Lyapunov method lies in analyzing the second variation $\mathcal{L}$.  
The spectral parameters determine the number of negative eigenvalues of $\mathcal{L}$, while the scattering parameters determine 
the dimension of its kernel.  
An analogous argument applies to $\tilde{\mathcal{L}}$, the second variation of $\tilde{\mathcal{I}}$.  

Let $\lfloor \cdot \rfloor$ denote the floor function, i.e., $\lfloor x \rfloor$ is the greatest 
integer less than or equal to $x$.
The nonlinear stability is established by the following theorem.
\begin{thm}\label{thm-spectrum-L-tildeL}
 (a) Let $\mathbf{q}^{[N]}$ be the $N$-soliton solution of the CNLS equations given in \eqref{CNLS-Nsoliton}. 
 Then the essential spectrum of the self-adjoint operator $\mathcal{L}(\mathbf{q}^{[N]})$ is 
    \begin{equation*}
        \sigma_{\mathrm{ess}}(\mathcal{L})=\left[2^{2N} \min_{\lambda \in \mathbb{R}}|\mathcal{P}(\lambda)|^{2},+\infty\right),
    \end{equation*}
 where $\mathcal{P}(\lambda)=\prod_{k=1}^{N}(\lambda-\lambda_{k}^{*})$. 
 The point spectrum consists of $N$ negative eigenvalues (counting multiplicities),
    \begin{equation*}
 |\sigma_{\mathrm{point}}(\mathcal{L})\cap \mathbb{R}^{-}|=N,
    \end{equation*}
 and the zero eigenvalue with multiplicity $4N$. 
 Moreover, the point spectrum of $\mathcal{L}(\mathbf{q}^{[N]})$ is finite. 
    
 (b) Let $\mathbf{q}^{[N_{1},N_{2}]}$ be the $(N_{1},N_{2})$-soliton solution of the CmKdV equations given in \eqref{CmKdV-Nsoliton}. 
 Then the essential spectrum of $\tilde{\mathcal{L}}(\mathbf{q}^{[N_{1},N_{2}]})$ is 
    \begin{equation*}
        \sigma_{\mathrm{ess}}(\tilde{\mathcal{L}})=\left[2^{2\tilde{N}} \min_{\lambda \in \mathbb{R}}|\mathcal{P}(\lambda)|^{2},+\infty\right).
    \end{equation*}
 The point spectrum of $\tilde{\mathcal{L}}(\mathbf{q}^{[N_{1},N_{2}]})$ is finite, and $\tilde{\mathcal{L}}$ satisfies
    \begin{equation*}
 |\sigma_{\mathrm{point}}(\tilde{\mathcal{L}})\cap \mathbb{R}^{-}|=N_{1}+\left\lfloor \tfrac{N_{2}+1}{2}\right\rfloor,
    \end{equation*}
 and admits the zero eigenvalue with multiplicity $2\tilde{N}$. 
\end{thm}

\begin{rem}
 For the $N$-soliton solution of the NLS equation and the $(N_{1},N_{2})$-soliton solution of the mKdV equation, 
 the numbers of negative eigenvalues of the corresponding operators $\mathcal{L}$ and $\tilde{\mathcal{L}}$ are still 
 $N$ and $N_{1}+\lfloor (N_{2}+1)/2 \rfloor$, respectively, while the zero eigenvalues have multiplicity $2N$ and 
 $\tilde{N}$, respectively. These results are consistent with previous studies 
 \cite{kapitula2007stability,alejo2013nonlinear,le2021stability}. 
\end{rem}

Considering the spectral parameters of soliton solutions, the operator $\mathcal{L}$ (and $\tilde{\mathcal{L}}$) 
can be reduced to $\mathcal{L}\mathcal{P}$ (and $\tilde{\mathcal{L}}\tilde{\mathcal{P}}$), where 
$\mathcal{P}$ (and $\tilde{\mathcal{P}}$) is the projection onto the subspace determined by the spectral parameters. 
The reduced operators have no negative eigenvalues and are coercive on the orthogonal complement of their kernels. 
Since the kernels are described by the scattering parameters, the nonlinear stability of soliton solutions 
follows from Theorem \ref{thm-stability} via the modulation argument.

The proof of Theorem \ref{thm-spectrum-L-tildeL} relies on squared eigenfunctions and 
squared eigenfunction matrices derived from integrable systems. In Section \ref{sec-gen-lin-U}, 
we discuss squared eigenfunction matrices satisfying the stationary zero-curvature equations in a general 
Lie algebra $\mathcal{U}$ with subalgebra $\mathcal{T}$. We show that the squared eigenfunctions associated 
with the $n$-th flow equation satisfy the corresponding linearized spectral problem. 

Specializing to $\mathcal{U}=\mathrm{gl}(3,\mathbb{C})$, the general linear Lie algebra, and letting 
$\mathcal{T}$ be the fixed-point subalgebra of the conjugation map with respect to $\sigma_{3}$, 
we obtain the squared eigenfunctions required for the CNLS and CmKdV equations, with $(\mathcal{U},\mathcal{T})$ 
forming a symmetric pair. Since the squared eigenfunctions defined by soliton solutions admit separation of 
variables, all eigenfunctions of the auxiliary linearized operator $\mathcal{J}\mathcal{L}$ can be found. 
By completeness of the squared eigenfunctions, the quadratic form $(\mathcal{L}\cdot,\cdot)$ restricted to 
their span can be characterized by $(\mathcal{J}^{-1}\cdot,\cdot)$, obtained from orthogonality relations 
between squared eigenfunctions and adjoint squared eigenfunctions. Consequently, the kernel of 
$\mathcal{L}$ can be characterized in terms of squared eigenfunctions, and the number of negative 
eigenvalues is determined by the negative Krein signature of $\mathcal{L}$ on this set.

\subsection{Outline}
In Section \ref{sec-gen-lin-U}, we show that solutions of the stationary zero curvature equations solve the 
linearized spectral problem of the corresponding mixed flow equations (Theorem \ref{var-L}), 
and establish their relation with steady-state solutions (Theorem \ref{fin-res-G}). 
Section \ref{sec-$N$-soliton-CNLS} presents the Darboux transformation for constructing $N$-soliton solutions 
of the CNLS equations and the associated squared eigenfunctions. 
In Section \ref{sec-stability-CNLS}, we derive the orthogonality relations for squared eigenfunctions and 
squared eigenfunction matrices (Theorem \ref{thm-orthogonal}), which yield the spectral analysis of $\mathcal{L}$ in part (a) 
of Theorem \ref{thm-spectrum-L-tildeL} and establish the nonlinear stability of $N$-soliton solutions. 
Section \ref{sec-stability-CmKdV} constructs squared eigenfunctions for $(N_{1},N_{2})$-soliton solutions of 
the CmKdV equations, analyzes the spectrum of $\tilde{\mathcal{L}}$, and proves their nonlinear stability.

\section{Linearized operator and $\mathbf{L}$ matrix}\label{sec-gen-lin-U}
In this section, the linearized operator associated with evolution equations is considered in the setting of 
a general Lie algebra $\mathcal{U}$ and its subalgebra $\mathcal{T}$ \cite{beals1985inverse,sattinger1985hamiltonian,terng1997soliton}. 
The specific form relevant to this work will be presented in Section~\ref{example-CNLS}, 
where the choice $\mathcal{U} = \mathrm{gl}(3, \mathbb{C})$ is applied.
The starting point is the differential equation
\begin{equation}\label{equ-Lx-a}
    \mathbf{L}_{x}=[\ii \lambda \mathbf{a}+\mathbf{Q},\mathbf{L}]
\end{equation}
where $\mathbf{a} \in \mathcal{T}$ and $\mathbf{Q} \in \mathcal{S}(\mathcal{T}^{\perp})$, 
the space of Schwartz-class smooth functions 
from $\mathbb{R}$ to $\mathcal{T}^{\perp}$. The function $\mathbf{L}$ admits an expansion of the form
\begin{equation}\label{exp-L-b}
    \mathbf{L}=\mathbf{b}+\sum_{n=1}^{\infty}\mathbf{L}_{n}\lambda^{-n}
\end{equation}
as $\lambda\to\infty$ with $\mathbf{b} \in \mathcal{T}$,
and $\mathbf{L}_{n+1}$ corresponds to the $n$-th flow equation in the associated integrable hierarchy \cite{sattinger1985hamiltonian}. 
A key observation is that the projection of a solution $\mathbf{G}$ to the stationary zero curvature equation \eqref{sta-zero-G}, 
when acted upon by $\mathrm{ad}_{\mathbf{b}}$, satisfies the linearized equation associated with the $n$-th flow equation. 
This result is stated in Theorem \ref{var-L}.

Furthermore, if the potential $\mathbf{Q}$ is a steady-state solution, then solutions to the stationary 
zero curvature equations can be expressed as polynomials in $\lambda$, $\mathbf{Q}$, and the derivatives 
of $\mathbf{Q}$; see Theorem~\ref{fin-res-G}. An immediate corollary is that the kernel of the linearized 
operator can be explicitly identified.
\subsection{The variation of $\mathbf{L}$ matrix}
Let $\mathcal{U}$ be a Lie algebra equipped with a nondegenerate $\mathrm{ad}$-invariant bilinear form $(\cdot,\cdot)_{\mathcal{U}}$, and let $\mathcal{T} \subset \mathcal{U}$ be a subalgebra. Denote by $\mathcal{T}^{\perp}$ the orthogonal complement of $\mathcal{T}$ with respect to this bilinear form. Assume that the restriction of the bilinear form to $\mathcal{T}$ is also nondegenerate and that
\begin{equation}\label{sym-T-bot}
[\mathcal{T}^{\perp}, \mathcal{T}^{\perp}] \subset \mathcal{T}.
\end{equation}
Under this assumption, the decomposition $\mathcal{U} = \mathcal{T} \oplus \mathcal{T}^{\perp}$ holds, 
and every element $\mathbf{u} \in \mathcal{U}$ can be written uniquely as $\mathbf{u} = \mathbf{u}^{\pi_{0}} + \mathbf{u}^{\perp}$ 
with $\mathbf{u}^{\pi_{0}} \in \mathcal{T}$ and $\mathbf{u}^{\perp} \in \mathcal{T}^{\perp}$, where 
$\pi_{0}:\mathcal{U}\to\mathcal{T}$ be the projection. It is also noted that
\begin{equation*}
    [\mathcal{T},\mathcal{T}^{\perp}]\subset \mathcal{T}^{\perp}
\end{equation*}
which follows from the $\mathrm{ad}$-invariance of $(\cdot, \cdot)_{\mathcal{U}}$. 

Let $\mathbf{a}, \mathbf{b}$ be two elements in the centralizer of $\mathcal{T}$ such that $\mathbf{a} - \mathbf{b}$ belongs to the centralizer of $\mathcal{U}$, i.e.,
\begin{equation*}
    \begin{split}
        &\mathbf{a}, \mathbf{b} \in C(\mathcal{T}) := \{ \mathbf{u} \in \mathcal{T} : [\mathbf{u}, \mathbf{v}] = 0, \ \forall \mathbf{v} \in \mathcal{T} \}, \quad
\\ &\mathbf{a} - \mathbf{b} \in C(\mathcal{U}) := \{ \mathbf{u} \in \mathcal{U} : [\mathbf{u}, \mathbf{v}] = 0, \ \forall \mathbf{v} \in \mathcal{U} \}.
    \end{split}
\end{equation*}
Assume further that $\mathrm{ad}_{\mathbf{a}}$ is invertible on $\mathcal{T}^{\perp}$. Consider the differential 
equation \eqref{equ-Lx-a}, where $\mathbf{L}$ admits the expansion given in \eqref{exp-L-b}. 
More precisely, the coefficients $\mathbf{L}_{n}$ satisfy the recursive relation
\begin{equation}\label{rec-L-gen}
    \partial_{x}\mathbf{L}_{n}=\ii \mathrm{ad}_{\mathbf{a}}\mathbf{L}_{n+1}^{\perp}+[\mathbf{Q},\mathbf{L}]. 
\end{equation}
Using \eqref{sym-T-bot}, the recursion relations can be decomposed into components in $\mathcal{T}$ and $\mathcal{T}^{\perp}$ as follows:
\begin{align}
    \partial_{x}\mathbf{L}^{\pi_{0}}_{n+1}=&\mathrm{ad}_{\mathbf{Q}} \mathbf{L}^{\perp}_{n+1} ,\label{L-diag-n-gen} \\
    \mathbf{L}^{\perp}_{n+1}=&-\ii\mathrm{ad}_{\mathbf{a}}^{-1}\left(
        \partial_{x}\mathbf{L}^{\perp}_{n}-\mathrm{ad}_{\mathbf{Q}} \mathbf{L}^{\pi_{0}}_{n}
    \right)\label{L-off-n-gen}
\end{align}
since $\mathrm{ad}_{\mathbf{a}}^{-1}: \mathcal{T}^{\perp} \to \mathcal{T}^{\perp}$ exists.
The first few terms in the expansion of $\mathbf{L}$ are given by
\begin{align*}
    \mathbf{L}_{0}=&\mathbf{b}, \\
    \mathbf{L}_{1}=&-\ii\mathbf{Q}, \\
    \mathbf{L}_{2}=&-\partial_{x}^{-1}\mathrm{ad}_{\mathbf{Q}}\mathrm{ad}_{\mathbf{a}}^{-1}\mathbf{Q}_{x}-\mathrm{ad}_{\mathbf{a}}^{-1}\mathbf{Q}_{x},
\end{align*}
where the identity $\mathrm{ad}_{\mathbf{b}} = \mathrm{ad}_{\mathbf{a}}$ is used. We note 
that $\partial_{x}^{-1}=\int_{-\infty}^{x}$ and the integration constant is zero since $\mathbf{Q}\in \mathcal{S}(\mathcal{T}^{\perp})$. 
Denote by $(\sum_{n=-\infty}^{+\infty} A_{n}\lambda^{n})_{+}=\sum_{n=0}^{+\infty} A_{n}\lambda^{n}$ 
the nonnegative part of a formal Laurent series. It is known \cite{beals1985inverse,sattinger1985hamiltonian,terng1997soliton} that if $\mathbf{L}$ also satisfies
\begin{equation*}
    \mathbf{L}_{t}=[\mathbf{V}_{n},\mathbf{L}]
\end{equation*}
where 
\begin{equation}\label{Vn-def}
    \mathbf{V}_{n}=\ii(\lambda^{n} \mathbf{L})_{+},
\end{equation}
then the potential $\mathbf{Q}$ satisfies the $n$-th flow equation
\begin{align*}
    \mathbf{Q}_{t}=(\ii \lambda \mathbf{a}+\mathbf{Q})_{t}=&\mathbf{V}_{n,x}-[\ii \lambda \mathbf{a}+\mathbf{Q},\mathbf{V}_{n}]\\
    =&-\mathrm{ad}_{\mathbf{b}}\left((\lambda^{n+1}\mathbf{L})_{+}-\lambda (\lambda^{n}\mathbf{L})_{+}\right)\\
    =&-\mathrm{ad}_{\mathbf{b}}\mathbf{L}_{n+1}^{\perp}. 
\end{align*}
For the matrix
\begin{equation}\label{V-def}
    \mathbf{V}=\sum_{n=0}^{N}\beta_{n}\mathbf{V}_{n}, 
\end{equation}
then the potential $\mathbf{Q}$ satisfies the mixed flow equation
\begin{equation}\label{mixed-flow}
    \mathbf{Q}_{t}=-\mathrm{ad}_{\mathbf{b}}\sum_{n=0}^{N}\beta_{n}\mathbf{L}_{n+1}^{\perp}. 
\end{equation}
The variation of a function $\mathcal{K}'(\mathbf{Q})$ is an operator and given by 
    \begin{equation}\label{def-var-Q-2}
        \frac{\delta \mathcal{K}'}{\delta \mathbf{Q}}(\mathbf{Q})[\delta \mathbf{Q}]=\frac{d }{d \epsilon}\mathcal{K}'(\mathbf{Q}+\epsilon\delta \mathbf{Q})|_{\epsilon=0}. 
    \end{equation}
Then the variation of $\mathbf{L}_{n}$ with respect to $\mathbf{Q}$ is governed by the following 
recursive relations by \eqref{L-diag-n-gen} and \eqref{L-off-n-gen}:
\begin{align}
    \partial_{x}\frac{\delta\mathbf{L}^{\pi_{0}}_{n+1}}{\delta \mathbf{Q}}=&\mathrm{ad}_{(\cdot)} \mathbf{L}^{\perp}_{n+1}+
    \mathrm{ad}_{\mathbf{Q}} \frac{\delta\mathbf{L}^{\perp}_{n+1}}{\delta \mathbf{Q}} ,\label{delta-diag-n-gen} \\
    \frac{\delta\mathbf{L}^{\perp}_{n+1}}{\delta \mathbf{Q}}=&-\ii\mathrm{ad}_{\mathbf{a}}^{-1}\left(
        \partial_{x}\frac{\delta\mathbf{L}^{\perp}_{n}}{\delta \mathbf{Q}}-\mathrm{ad}_{(\cdot)} \mathbf{L}^{\pi_{0}}_{n}
        -\mathrm{ad}_{\mathbf{Q}} \frac{\delta\mathbf{L}^{\pi_{0}}_{n}}{\delta \mathbf{Q}}
    \right).\label{delta-off-n-gen}
\end{align}
The following theorem concerns the linearized problem associated with the mixed flow equation:

\begin{thm}\label{var-L}
    Let $\mathbf{G} = \mathbf{G}(\lambda; x,t)$ satisfy the stationary zero curvature equations
    \begin{equation}\label{sta-zero-G}
        \mathbf{G}_{x}=[\ii \lambda \mathbf{a}+\mathbf{Q},\mathbf{G}],\quad \mathbf{G}_{t}=[\mathbf{V},\mathbf{G}], 
    \end{equation}
    where $\mathbf{V}$ is given by \eqref{V-def}. Then the function $\mathbf{G}$ satisfies the linearized evolution equation
    \begin{equation}\label{G-t-linear-pro-gen}
        \mathbf{G}_{t}=-\sum_{n=0}^{N}\beta_{n}
        \frac{\delta \mathbf{L}_{n+1}}{\delta \mathbf{Q}}(\mathrm{ad}_{\mathbf{b}}\mathbf{G}). 
    \end{equation}
    In particular, the quantity $\mathrm{ad}_{\mathbf{b}} \mathbf{G}^{\perp}$ solves the linearized problem associated with the mixed flow equation \eqref{mixed-flow}, namely,
    \begin{equation}\label{G-t-linear-pro-gen-1}
        \mathrm{ad}_{\mathbf{b}}\mathbf{G}^{\perp}_{t}=-\mathrm{ad}_{\mathbf{b}}\sum_{n=0}^{N}\beta_{n}
        \frac{\delta \mathbf{L}_{n+1}^{\perp}}{\delta \mathbf{Q}}(\mathrm{ad}_{\mathbf{b}}\mathbf{G}^{\perp}). 
    \end{equation}
\end{thm}
The formula \eqref{G-t-linear-pro-gen-1} follows immediately by applying $\mathrm{ad}_{\mathbf{b}}$ to both sides of \eqref{G-t-linear-pro-gen}.
In view of the time component of the stationary zero curvature equations and \eqref{G-t-linear-pro-gen},
it suffices to verify that the right-hand side of \eqref{G-t-linear-pro-gen} coincides with $[\mathbf{V}, \mathbf{G}]$,
which is a relation determined solely by the spatial part of the stationary zero curvature equations, as shown in the following lemma.
\begin{lem}
    Let $\mathbf{G} = \mathbf{G}(\lambda; x)$ satisfy the first-order differential system
    \begin{equation}\label{G-x-a}
        \mathbf{G}_{x}=[\ii \lambda \mathbf{a}+\mathbf{Q},\mathbf{G}], 
    \end{equation}
    then for all $n \geq 0$, the following relation holds:
    \begin{align}
        [\mathbf{V}_{n},\mathbf{G}]=&-  
        \frac{\delta \mathbf{L}_{n+1}}{\delta \mathbf{Q}}(\mathrm{ad}_{\mathbf{b}}\mathbf{G}),\label{inter-G}
    \end{align}
    where $\mathbf{V}_{n}$ is defined in \eqref{Vn-def}.
\end{lem}
\begin{proof}
    In what follows, we show that
\begin{align}
    [\mathbf{V}_{n},\mathbf{G}]^{\perp}=&-
    \frac{\delta \mathbf{L}_{n+1}^{\perp}}{\delta \mathbf{Q}}(\mathrm{ad}_{\mathbf{b}}\mathbf{G}^{\perp}),\label{inter-G-off} \\
    [\mathbf{V}_{n},\mathbf{G}]^{\pi_{0}}=&-  
    \frac{\delta \mathbf{L}_{n+1}^{\pi_{0}}}{\delta \mathbf{Q}}(\mathrm{ad}_{\mathbf{b}}\mathbf{G}^{\perp})\label{inter-G-diag},
\end{align}
which is equivalent to \eqref{inter-G}, since $\mathrm{ad}_{\mathbf{b}}\mathbf{G}^{\perp}\in \mathcal{T}^{\perp}$. 

    We prove it by induction. For $n=0$, we have 
    \begin{equation}\label{var-L1}
        -\frac{\delta \mathbf{L}_{1}}{\delta \mathbf{Q}}(\mathrm{ad}_{\mathbf{b}}\mathbf{G})=\ii\mathrm{ad}_{\mathbf{b}}\mathbf{G}=[\mathbf{V}_{0},\mathbf{G}].
    \end{equation} 
    If $n=1$, for \eqref{inter-G-off}, we have 
    \begin{align*}
        [\mathbf{V}_{1},\mathbf{G}]^{\perp}=\ii \lambda [\mathbf{a}, \mathbf{G}^{\perp}]+[\mathbf{Q},\mathbf{G}^{\pi_{0}}]
    \end{align*}
    and 
    \begin{align*}
        -\frac{\delta \mathbf{L}_{2}^{\perp}}{\delta \mathbf{Q}}(\mathrm{ad}_{\mathbf{a}}\mathbf{G}^{\perp})=&\partial_{x}\mathbf{G}^{\perp}
        =\ii \lambda [\mathbf{a}, \mathbf{G}^{\perp}]+[\mathbf{Q},\mathbf{G}^{\pi_{0}}].
    \end{align*}
    The diagonal part can be obtained by 
    \begin{equation}\label{var-L2-diag}
        \begin{split}
            \partial_{x}\left([\mathbf{V}_{1},\mathbf{G}]^{\pi_{0}}+\frac{\delta \mathbf{L}_{2}^{\pi_{0}}}{\delta \mathbf{Q}}(\mathrm{ad}_{\mathbf{a}}\mathbf{G}^{\perp})\right)
        =&\partial_{x}[\mathbf{Q},\mathbf{G}^{\perp}]-
        [\mathrm{ad}_{\mathbf{a}}\mathbf{G}^{\perp},\mathrm{ad}_{\mathbf{a}}^{-1}\partial_{x}\mathbf{Q}]-
        \mathrm{ad}_{\mathbf{Q}} \mathrm{ad}_{\mathbf{a}}^{-1}(\mathrm{ad}_{\mathbf{a}}\partial_{x}\mathbf{G}^{\perp})\\
        =&
        [[\mathbf{G}^{\perp},\mathrm{ad}_{\mathbf{a}}^{-1}\mathbf{Q}_{x}],\mathbf{a}] \\
        =&0
        \end{split}
    \end{equation}
    since $\mathbf{a}\in C(\mathcal{T})$ and $\mathbf{Q}\in \mathcal{S}(\mathcal{T}^{\perp})$. 
    Now, assuming that \eqref{inter-G-off} and \eqref{inter-G-diag} hold, differentiating both sides of \eqref{inter-G-off} with respect to $x$, we obtain
    \begin{align*}
        -\partial_{x}\frac{\delta \mathbf{L}_{n+1}^{\perp}}{\delta \mathbf{Q}}(\mathrm{ad}_{\mathbf{a}}\mathbf{G}^{\perp})=&
        [\partial_{x}\mathbf{V}_{n},\mathbf{G}]^{\perp}+[\mathbf{V}_{n},\partial_{x}\mathbf{G}]^{\perp} \\
        =&
        \left(\ii[[\mathbf{a},\mathbf{V}_{n+1}],\mathbf{G}]+\ii[\mathbf{V}_{n+1},[\mathbf{a},\mathbf{G}]]
        +[\mathbf{L}_{n+1},[\mathbf{a},\mathbf{G}]]
        \right.\\&\left.+[[\mathbf{Q},\mathbf{V}_{n}],\mathbf{G}]+
        [[\mathbf{G},\mathbf{Q}],\mathbf{V}_{n}]\right)^{\perp}\\
        =&\left(-\ii[[\mathbf{V}_{n+1},\mathbf{G}],\mathbf{a}]+[\mathbf{L}_{n+1},[\mathbf{a},\mathbf{G}]]-
        [[\mathbf{V}_{n},\mathbf{G}],\mathbf{Q}]\right)^{\perp}
    \end{align*}
    since
    \begin{align*}
        [\partial_{x}\mathbf{V}_{n},\mathbf{G}]=&\ii\left[\sum_{i=0}^{n}\lambda^{i}\partial_{x}\mathbf{L}_{n-i},\mathbf{G}\right] \\
        =&\ii\left[
            \sum_{i=0}^{n}\lambda^{i}\ii [\mathbf{a},\mathbf{L}_{n-i+1}]
            +\lambda^{i}[\mathbf{Q},\mathbf{L}_{n-i}],\mathbf{G}
        \right]\\
        =& \ii\left[
            [\mathbf{a},\mathbf{V}_{n+1}-\ii\lambda^{n+1}\mathbf{a}]
            -\ii[\mathbf{Q},\mathbf{V}_{n}],\mathbf{G}
        \right]\\
        =&\ii[[\mathbf{a},\mathbf{V}_{n+1}],\mathbf{G}]+
        [[\mathbf{Q},\mathbf{V}_{n}],\mathbf{G}]
    \end{align*} 
    and 
    \begin{align*}
        [\mathbf{V}_{n},\partial_{x}\mathbf{G}]=&[\mathbf{V}_{n},\ii[\lambda \mathbf{a},\mathbf{G}]+[\mathbf{Q},\mathbf{G}]] \\
        =&\ii [\lambda \mathbf{V}_{n},[\mathbf{a},\mathbf{G}]]+
        [\mathbf{V}_{n},[\mathbf{Q},\mathbf{G}]]\\
        =&\ii[\mathbf{V}_{n+1},[\mathbf{a},\mathbf{G}]]
        +[\mathbf{L}_{n+1},[\mathbf{a},\mathbf{G}]]+[[\mathbf{G},\mathbf{Q}],\mathbf{V}_{n}].
    \end{align*}
    Hence
    \begin{align*}
        - 
        \frac{\delta \mathbf{L}_{n+2}^{\perp}}{\delta \mathbf{Q}}(\mathrm{ad}_{\mathbf{a}}\mathbf{G}^{\perp})=&
        \ii\mathrm{ad}_{\mathbf{a}}^{-1}\left(\partial_{x}\frac{\delta\mathbf{L}^{\perp}_{n+1}}{\delta \mathbf{Q}}(\mathrm{ad}_{\mathbf{a}}\mathbf{G}^{\perp})-[\mathrm{ad}_{\mathbf{a}}\mathbf{G}^{\perp},\mathbf{L}^{\pi_{0}}_{n+1}]
        -\mathrm{ad}_{\mathbf{Q}} \frac{\delta\mathbf{L}^{\pi_{0}}_{n+1}}{\delta \mathbf{Q}}(\mathrm{ad}_{\mathbf{a}}\mathbf{G}^{\perp})\right)\\
        =&\ii\mathrm{ad}_{\mathbf{a}}^{-1}\left(\ii[[\mathbf{V}_{n+1},\mathbf{G}],\mathbf{a}]-[\mathbf{L}_{n+1},\mathrm{ad}_{\mathbf{a}}\mathbf{G}^{\perp}]+
        [[\mathbf{V}_{n},\mathbf{G}],\mathbf{Q}]\right)^{\perp}\\&-
        \ii\mathrm{ad}_{\mathbf{a}}^{-1}[\mathrm{ad}_{\mathbf{a}}\mathbf{G}^{\perp},\mathbf{L}^{\pi_{0}}_{n+1}]-\ii\mathrm{ad}_{\mathbf{a}}^{-1}\left[\mathbf{Q},\frac{\delta\mathbf{L}^{\pi_{0}}_{n+1}}{\delta \mathbf{Q}}(\mathrm{ad}_{\mathbf{a}}\mathbf{G}^{\perp})\right]\\
        =&[\mathbf{V}_{n+1},\mathbf{G}]^{\perp}-\ii\mathrm{ad}_{\mathbf{a}}^{-1}[\mathbf{Q},[\mathbf{V}_{n},\mathbf{G}]^{\pi_{0}}]-
        \ii\mathrm{ad}_{\mathbf{a}}^{-1}\left[\mathbf{Q},\frac{\delta\mathbf{L}^{\pi_{0}}_{n+1}}{\delta \mathbf{Q}}(\mathrm{ad}_{\mathbf{a}}\mathbf{G}^{\perp})\right]\\
        =&[\mathbf{V}_{n+1},\mathbf{G}]^{\perp}-\ii\mathrm{ad}_{\mathbf{a}}^{-1}\left[\mathbf{Q},[\mathbf{V}_{n},\mathbf{G}]^{\pi_{0}}+\frac{\delta\mathbf{L}^{\pi_{0}}_{n+1}}{\delta \mathbf{Q}}(\mathrm{ad}_{\mathbf{a}}\mathbf{G}^{\perp})\right]\\
        =&[\mathbf{V}_{n+1},\mathbf{G}]^{\perp}. 
    \end{align*}
    by \eqref{delta-off-n-gen} and \eqref{inter-G-diag}. 

    It remains to prove that \eqref{inter-G-diag} holds when $n$ is replaced by $n+1$. This follows from the identity
    \begin{align*}
        \partial_{x}\frac{\delta \mathbf{L}_{n+2}^{\pi_{0}}}{\delta \mathbf{Q}}(\mathrm{ad}_{\mathbf{a}}\mathbf{G}^{\perp})=&
        [\mathrm{ad}_{\mathbf{a}}\mathbf{G}^{\perp},\mathbf{L}_{n+2}^{\perp}]-[\mathbf{Q},[\mathbf{V}_{n+1},\mathbf{G}]^{\perp}]\\
        =&[[\mathbf{a},\mathbf{G}],\mathbf{L}_{n+2}]^{\pi_{0}}-[\mathbf{Q},[\mathbf{V}_{n+1},\mathbf{G}]]^{\pi_{0}}\\
        =&-\partial_{x}[\mathbf{V}_{n+1},\mathbf{G}]^{\pi_{0}}
    \end{align*}
    which is derived using \eqref{delta-diag-n-gen}, \eqref{delta-off-n-gen}, and \eqref{L-off-n-gen}.
\end{proof}
Formula \eqref{G-t-linear-pro-gen} in Theorem \eqref{var-L} results from applying $\sum \beta_n$ to both sides of \eqref{inter-G}. 
\begin{rem}
    The condition $\mathbf{a}-\mathbf{b}\in C(\mathcal{T})$ is required by \eqref{var-L1}, while \eqref{sym-T-bot} is necessary for 
    working within the space $\mathcal{T}^{\perp}$, as seen from \eqref{var-L2-diag}. Unlike the approach in \cite{terng1997soliton}, we 
    do not assume that $\mathbf{a}$ is a regular element, but only require that
    $\mathrm{ad}_{\mathbf{a}}$ is invertible on $\mathcal{T}^{\perp}$. The structural assumption 
    \eqref{sym-T-bot} also allows us to project $[\mathbf{Q},\mathbf{G}]$ onto the subspaces 
    $\mathcal{T}$ and $\mathcal{T}^{\bot}$, yielding $[\mathbf{Q},\mathbf{G}^{\perp}]$ and $[\mathbf{Q},\mathbf{G}^{\pi_{0}}]$, 
    respectively. 
\end{rem}

\subsection{Steady-state solutions to flow equations}
The steady-state solutions, which satisfy an ODE in the spatial variable, form a large class of solutions to the mixed flow equations \eqref{mixed-flow}. For instance, soliton solutions fall into this category \cite{kapitula2007stability}.
This raises the natural question of how to construct more general steady-state solutions for such flows. 
The following theorem addresses this problem by employing the stationary zero curvature equations, 
which also allow for the construction of other types of solutions \cite{novikov1984theory}, such as those expressed in terms of 
elliptic functions (see \cite{ling2023stability}). 
Here we assume that $\mathbf{L}_{n}\in \mathcal{A}$ where $\mathcal{A}=\cup_{n=0}^{\infty}\mathcal{A}_{n}$ with 
$\mathcal{A}_{n}$ denote the algebra of polynomials in $\mathbf{Q}$ and its derivatives up to order $n$.
Then $\mathbf{L}_{n+1}$ is a differential polynomial in $\mathbf{Q}$ and its derivatives with 
respect to $x$. Then one has \cite{sattinger1985hamiltonian}
\begin{equation*}
    \mathbf{L}_{n}^{\perp} \in \mathcal{A}_{n-1},\quad \mathbf{L}_{n}^{\pi_{0}}\in \mathcal{A}_{n-2}.
\end{equation*}
\begin{thm}\label{fin-res-G}
    For $\alpha_{m}\in \mathbb{C}$, let $\mathbf{G}=\mathbf{G}(\lambda;x,t)$ be a function of the form
    \begin{equation}\label{res-G-def}
        \mathbf{G}=\sum_{m=0}^{M}\alpha_{m}\mathbf{V}_{m}    
    \end{equation}
    satisfying the stationary zero curvature equations
    \begin{equation}\label{G-equation-gen}
        \mathbf{G}_{x}=[\ii \lambda \mathbf{a}+\mathbf{Q},\mathbf{G}],\quad \mathbf{G}_{t}=[\mathbf{V},\mathbf{G}], 
    \end{equation}
    where $\mathbf{V}$ is defined in \eqref{V-def}. 
    Then such a function $\mathbf{G}$ exists if and only if the mixed flow equation \eqref{mixed-flow} and 
    the differential equation about $\mathbf{Q}$
    \begin{equation}\label{equ-G-condition-gen}
        \sum_{m=0}^{M}\alpha_{m}\mathbf{L}_{m+1}^{\perp}(\mathbf{Q})=0
    \end{equation}
    hold. 
\end{thm}
Without loss of generality, we set $\alpha_{M}\ne 0$. 
Note that \eqref{equ-G-condition-gen} is a differential equation of order $M$, since $\mathbf{L}_{m+1}^{\perp}\in\mathcal{A}_{m}$. 
By choosing different values of $M$, the equation \eqref{equ-G-condition-gen} can be used to construct steady-state solutions of the mixed flow equation.
Once such a solution is obtained, the corresponding function $\mathbf{G}$ can be reconstructed using the representation \eqref{res-G-def}, which defines 
$\mathbf{G}$ as a polynomial in $\lambda$, $\mathbf{Q}$ and derivatives of $\mathbf{Q}$ in 
view of the definition of $\mathbf{V}_{n}$ in \eqref{Vn-def}. For particular systems, such as 
integrable equations, a Darboux transformation provides a method to construct a 
new solution $\mathbf{G}^{[1]}$ associated with a transformed potential $\mathbf{Q}^{[1]}$. 

The condition $\mathbf{L}_{n}\in \mathcal{A}$ can be obtained in exact integrable equation 
with regular element $\mathbf{a}$ \cite{beals1984scattering,beals1985inverse,sattinger1985hamiltonian}. 
But the theorem we need in this paper is not the case $\mathbf{a}$ is a regular element. Since the 
condition $\mathbf{L}_{n}\in \mathcal{A}$ can be obtained \cite{sattinger1985hamiltonian} by 
the transfer matrix $\mathbf{S}$ and 
the sequence \cite{gel1975asymptotic}
\begin{equation}\label{exa-A}
    \mathcal{A}\xrightarrow{D}\mathcal{A}\xrightarrow{\nabla}\mathcal{A},
\end{equation}
is exact, i.e. $\mathrm{Im}(D)=\mathrm{Ker}(\nabla)$, where $D$ is the operator on $\mathcal{A}$
\begin{equation*}
    D=\frac{\partial}{\partial x}+\sum_{j=0}^{\infty} \mathbf{Q}_{j+1} \frac{\partial}{\partial \mathbf{Q}_{j}},\quad \mathbf{Q}_{j}=\partial_{x}^{j}\mathbf{Q}
\end{equation*}
and $\nabla$ is the Euler-Lagrange derivative 
\begin{equation*}
    \nabla=\sum_{n}(-D)^{n} \frac{\partial}{\partial \mathbf{Q}_{j}},
\end{equation*}
we can also prove the condition $\mathbf{L}_{n}\in \mathcal{A}$ in our case applying theorem 
in \cite{gel1975asymptotic}, see Remark \ref{L-in-A}. 

To prove Theorem \ref{fin-res-G}, it is necessary to analyze the relations between 
$\mathbf{V}_{n}$ and $\mathbf{V}_{m}$, taking into account the specific structure 
of $\mathbf{G}$ given by \eqref{res-G-def}. These relations are described in the following lemma.
\begin{lem}
    The identity 
    \begin{equation}\label{connect-U-Vm}
        [\ii \lambda \mathbf{a}+\mathbf{Q},\lambda^{m}\mathbf{L}]_{+}+\ii[\ii \lambda \mathbf{a}+\mathbf{Q},\mathbf{V}_{m}]=\ii[\mathbf{a},\mathbf{L}_{m+1}]
    \end{equation} 
    holds for all $m\geq 0$. In addition, the following identity is satisfied:
    \begin{equation}\label{connect-Vn-Vm}
        \ii[\mathbf{V}_{n},{\lambda^{m}\mathbf{L}}]_{+}-[\mathbf{V}_{n},\mathbf{V}_{m}]=
        -\sum_{j=1}^{n}\lambda^{n-j}\mathbf{L}_{j,m}
    \end{equation}
    where 
    \begin{equation*}
        \mathbf{L}_{j,m}=\sum_{i=0}^{j-1}[\mathbf{L}_{i},\mathbf{L}_{m-i+j}]. 
    \end{equation*}
    These identities are valid for all $n, m\geq 0$ and $j\geq 1$. 
    The quantities $\mathbf{L}_{j,m}$ satisfy the following recursion relations for their projections onto 
    $\mathcal{T}$ and $\mathcal{T}^{\bot}$, respectively:
    \begin{align}
        \mathbf{L}_{j,m}^{\pi_{0}}=&\ii \partial_{x}^{-1}[\mathrm{ad}_{\mathbf{a}}\mathbf{L}_{j}^{\perp},\mathbf{L}_{m+1}^{\perp}]
        +\partial_{x}^{-1}[\mathbf{Q},\mathbf{L}_{j,m}^{\perp}], \label{rec-Ljnm-diag}\\
        \mathbf{L}_{j,m}^{\perp}=&\ii \mathrm{ad}_{\mathbf{a}}\partial_{x}^{-1}\left(
            \mathbf{L}_{j+1,m}^{\perp}-[\mathbf{L}_{j}^{\pi_{0}},\mathbf{L}_{m+1}^{\perp}]
        \right)+\partial_{x}^{-1}[\mathbf{Q},\mathbf{L}_{j,m}^{\pi_{0}}]\label{rec-Ljnm-off}. 
    \end{align}
\end{lem}
\begin{proof}
    The identity \eqref{connect-U-Vm} follows from a straightforward calculation:
    \begin{align*}
        [\ii \lambda \mathbf{a}+\mathbf{Q},\lambda^{m}\mathbf{L}]_{+}+\ii[\ii \lambda \mathbf{a}+\mathbf{Q},\mathbf{V}_{m}]
        =&[\ii\lambda \mathbf{a}+\mathbf{Q},\lambda^{m}\mathbf{L}]_{+}
        -[\ii\lambda \mathbf{a}+\mathbf{Q},(\lambda^{m}\mathbf{L})_{+}]\\
        =&\ii[\mathbf{a},(\lambda^{m+1}\mathbf{L})_{+}-\lambda (\lambda^{m}\mathbf{L})_{+}]\\
        =&\ii[\mathbf{a},\mathbf{L}_{m+1}]. 
    \end{align*}

    To derive \eqref{connect-Vn-Vm}, we proceed as follows:
    \begin{align*}
        \ii[\mathbf{V}_{n},{\lambda^{m}\mathbf{L}}]_{+}-[\mathbf{V}_{n},\mathbf{V}_{m}]=&
        -\sum_{i=0}^{n}\left[\lambda^{n-i}\mathbf{L}_{i},{\lambda^{m}\mathbf{L}}\right]_{+}+
        \sum_{i=0}^{n}\left[\lambda^{n-i}\mathbf{L}_{i},{(\lambda^{m}\mathbf{L})}_{+}\right]\\
        =&-\sum_{i=0}^{n-1}\left[\mathbf{L}_{i},\sum_{j=m+1}^{n+m-i}\lambda^{n+m-i-j}\mathbf{L}_{j}\right]\\
        =&-\sum_{i=0}^{n-1}\sum_{j=1}^{n-i}\lambda^{n-i-j}\left[\mathbf{L}_{i},\mathbf{L}_{m+j}\right]\\
        =&-\sum_{j=1}^{n}\lambda^{n-j}\sum_{i=0}^{j-1}\left[\mathbf{L}_{i},\mathbf{L}_{m-i+j}\right]
    \end{align*}
    where in the last line we reindex the summation via the substitution $i+j\to j,i\to i$. 
    This proves identity \eqref{connect-Vn-Vm}. 
    We now proceed to prove the recursion relations \eqref{rec-Ljnm-diag} and \eqref{rec-Ljnm-off}. 
    Taking derivative to $\mathbf{L}_{j,m}$
    \begin{equation*}
        \begin{split}
            (\mathbf{L}_{j,m})_{x}=&\sum_{i=0}^{j-1}\left( [(\mathbf{L}_{i})_{x},\mathbf{L}_{m-i+j}]+[\mathbf{L}_{i},(\mathbf{L}_{m-i+j})_{x}] \right)\\
            =&\sum_{i=0}^{j-1}\left( [\ii \mathrm{ad}_{\mathbf{a}}\mathbf{L}_{i+1}^{\perp}+[\mathbf{Q},\mathbf{L}_{i}],\mathbf{L}_{m-i+j}]+[\mathbf{L}_{i},\ii \mathrm{ad}_{\mathbf{a}}\mathbf{L}_{m-i+j+1}^{\perp}+[\mathbf{Q},\mathbf{L}_{m-i+j}]] \right)\\
            =&\sum_{i=0}^{j-1}\left( \ii [\mathrm{ad}_{\mathbf{a}}\mathbf{L}_{i+1}^{\perp},\mathbf{L}_{m-i+j}]+\ii [\mathbf{L}_{i},\mathrm{ad}_{\mathbf{a}}\mathbf{L}_{m-i+j+1}^{\perp}]-[\mathbf{Q},[\mathbf{L}_{m-i+j},\mathbf{L}_{i}]] \right)\\
            =&\mathbf{K}+[\mathbf{Q},\mathbf{L}_{j,m}], \\
        \end{split}
    \end{equation*}
    where 
    \begin{equation*}
        \mathbf{K}=\ii\sum_{i=0}^{j-1}\left( [\mathrm{ad}_{\mathbf{a}}\mathbf{L}_{i+1}^{\perp},\mathbf{L}_{m-i+j}]+[\mathbf{L}_{i},\mathrm{ad}_{\mathbf{a}}\mathbf{L}_{m-i+j+1}^{\perp}]\right).
    \end{equation*}
    The projections of $\mathbf{K}$ onto $\mathcal{T}$ and $\mathcal{T}^{\perp}$ are given by
    \begin{equation*}
        \begin{split}
            \mathbf{K}^{\pi_{0}}=&
            \ii\sum_{i=0}^{j-1}\left([\mathrm{ad}_{\mathbf{a}}\mathbf{L}_{i+1}^{\perp},\mathbf{L}_{m-i+j}^{\perp}]+[\mathbf{L}_{i}^{\perp},\mathrm{ad}_{\mathbf{a}}\mathbf{L}_{m-i+j+1}^{\perp}]\right)\\
            =&\ii\sum_{i=0}^{j-1}\left( [\mathrm{ad}_{\mathbf{a}}\mathbf{L}_{i}^{\perp},\mathbf{L}_{m-i+j+1}^{\perp}]+[\mathbf{L}_{i}^{\perp},\mathrm{ad}_{\mathbf{a}}\mathbf{L}_{m-i+j+1}^{\perp}]\right)
            \\&+\ii[\mathrm{ad}_{\mathbf{a}}\mathbf{L}_{j}^{\perp},\mathbf{L}_{m+1}^{\perp}]-\ii[\mathrm{ad}_{\mathbf{a}}\mathbf{b},\mathbf{L}_{m+j+1}^{\perp}]\\
            =&\ii\sum_{i=0}^{j-1}\mathrm{ad}_{\mathbf{a}}[\mathbf{L}_{i}^{\perp},\mathbf{L}_{m-i+j+1}^{\perp}]+\ii[\mathrm{ad}_{\mathbf{a}}\mathbf{L}_{j}^{\perp},\mathbf{L}_{m+1}^{\perp}]\\
            =&\ii[\mathrm{ad}_{\mathbf{a}}\mathbf{L}_{j}^{\perp},\mathbf{L}_{m+1}^{\perp}]
        \end{split}
    \end{equation*}
    and 
    \begin{equation*}
        \begin{split}
            \mathbf{K}^{\perp}=&
            \ii\sum_{i=0}^{j-1}\left( [\mathrm{ad}_{\mathbf{a}}\mathbf{L}_{i+1}^{\perp},\mathbf{L}_{m-i+j}^{\pi_{0}}]+[\mathbf{L}_{i}^{\pi_{0}},\mathrm{ad}_{\mathbf{a}}\mathbf{L}_{m-i+j+1}^{\perp}]\right)\\
            =&\ii\mathrm{ad}_{\mathbf{a}}\sum_{i=0}^{j-1}\left( [\mathbf{L}_{i+1}^{\perp},\mathbf{L}_{m-i+j}^{\pi_{0}}]+[\mathbf{L}_{i}^{\pi_{0}},\mathbf{L}_{m-i+j+1}^{\perp}]\right)\\
            =&\ii\mathrm{ad}_{\mathbf{a}}\left(\mathbf{L}_{j+1,m}^{\perp}-[\mathbf{L}_{j}^{\pi_{0}},\mathbf{L}_{m+1}^{\perp}]\right). 
        \end{split}
    \end{equation*}
    We complete the proof. 
\end{proof}

Now we can prove Theorem \ref{fin-res-G}:
\begin{proof}[Proof of Theorem \ref{fin-res-G}]
    First, we prove the sufficiency. Since $\mathbf{G}$ satisfies \eqref{G-equation-gen}, the compatibility condition is given by
    \begin{equation*}
        (\ii \lambda \mathbf{a}+\mathbf{Q})_{t}-\mathbf{V}_{x}+[\ii \lambda \mathbf{a}+\mathbf{Q},\mathbf{V}]=0,
    \end{equation*}
    that is, the equation for the potential becomes
    \begin{equation*}
        \mathbf{Q}_{t}=-\sum_{n=0}^{N}\beta_{n} \mathrm{ad}_{\mathbf{a}}\mathbf{L}_{n+1}^{\perp},
    \end{equation*}
    which is equivalent to 
    \begin{equation*}
        \mathbf{L}_{t}=[\mathbf{V},\mathbf{L}]. 
    \end{equation*}
    Using \eqref{connect-U-Vm} and the equation \eqref{equ-Lx-a} satisfied by $\mathbf{L}$, we obtain
    \begin{equation}\label{G-x-part}
        \begin{split}
            \mathbf{G}_{x}-[\ii \lambda \mathbf{a}+\mathbf{Q},\mathbf{G}]=&
            \sum_{m=0}^{M}\alpha_{m} \left(\mathbf{V}_{m,x}-[\ii \lambda \mathbf{a}+\mathbf{Q},\mathbf{V}_{m}]\right) \\
            =\ii& \sum_{m=0}^{M}\alpha_{m} \left(((\lambda^{m}\mathbf{L})_{+})_{x}-[\ii \lambda \mathbf{a}+\mathbf{Q},(\lambda^{m}\mathbf{L})_{+}]\right)\\
            =& \ii\sum_{m=0}^{M}\alpha_{m} \left([\ii \lambda \mathbf{a}+\mathbf{Q},(\lambda^{m}\mathbf{L})]_{+}-[\ii \lambda \mathbf{a}+\mathbf{Q},(\lambda^{m}\mathbf{L})_{+}]\right)\\
            =& -\sum_{m=0}^{M}\alpha_{m}[\mathbf{a},\mathbf{L}_{m+1}]. 
        \end{split}
    \end{equation}

    Now we prove the necessity. If \eqref{equ-G-condition-gen} holds, then $\mathbf{G}_{x}=[\ii \lambda \mathbf{a}+\mathbf{Q},\mathbf{G}]$ by \eqref{G-x-part}. 
    As for the evolution with respect to time, the relation \eqref{connect-Vn-Vm} yields
    \begin{equation*}
        \begin{split}
            \mathbf{G}_{t}-[\mathbf{V},\mathbf{G}]=&
            \sum_{m=0}^{M}\alpha_{m} \left(\mathbf{V}_{m,t}-[\mathbf{V},\mathbf{V}_{m}]\right) \\
            =& \sum_{m=0}^{M}\alpha_{m} \left((\ii(\lambda^{m}\mathbf{L})_{+})_{t}-[\mathbf{V},\mathbf{V}_{m}]\right)\\
            =& \sum_{m=0}^{M}\alpha_{m} \left(\ii[\mathbf{V},\lambda^{m}\mathbf{L}]_{+}-[\mathbf{V},\mathbf{V}_{m}]\right)\\
            =& \sum_{m=0}^{M}\sum_{n=0}^{N}\alpha_{m}\beta_{n}\left(\ii[\mathbf{V}_{n},\lambda^{m}\mathbf{L}]_{+}-[\mathbf{V}_{n},\mathbf{V}_{m}]\right)\\
            =& -\sum_{m=0}^{M}\sum_{n=0}^{N}\alpha_{m}\beta_{n}\sum_{j=1}^{n}\lambda^{n-j}\mathbf{L}_{j,m}\\
            =& -\sum_{n=0}^{N}\beta_{n}\sum_{j=1}^{n}\lambda^{n-j}\sum_{m=0}^{M}\alpha_{m}\mathbf{L}_{j,m}. 
        \end{split}
    \end{equation*}  
    It suffices to prove that
    \begin{equation}\label{sum-Ljnm}
        \sum_{m=0}^{M}\alpha_{m}\mathbf{L}_{j,m}=0
    \end{equation}
    for all $j\geq 1$. 
    We first verify the case $j=1$. Using the definition of $\mathbf{L}_{j,m}$, we obtain
    \begin{equation*}
        \sum_{m=0}^{M}\alpha_{m}\mathbf{L}_{j,m}=\sum_{m=0}^{M}\alpha_{m}[\mathbf{L}_{0},\mathbf{L}_{m+1}]=
        \sum_{m=0}^{M}\alpha_{m}[\mathbf{a},\mathbf{L}_{m+1}]=0. 
    \end{equation*}
    Assuming \eqref{sum-Ljnm} holds for some $j\geq 1$, we next show that it also holds for $j+1$. 
    For the $\mathcal{T}^{\perp}$-component, by \eqref{rec-Ljnm-off},
    \begin{equation*}
        \sum_{m=0}^{M}\alpha_{m}\mathbf{L}_{j+1,m}^{\perp}=
        \sum_{m=0}^{M}\alpha_{m}[\mathbf{L}_{j}^{\pi_{0}},\mathbf{L}_{m+1}^{\perp}]=0.
    \end{equation*}
    For the $\mathcal{T}$-component, we have
    \begin{equation*}
        \sum_{m=0}^{M}\alpha_{m}\mathbf{L}_{j+1,m}^{\pi_{0}}=
        \ii \partial_{x}^{-1}\sum_{m=0}^{M}\alpha_{m}[\mathrm{ad}_{\mathbf{a}}\mathbf{L}_{j+1}^{\perp},\mathbf{L}_{m+1}^{\perp}]=0. 
    \end{equation*}
    Therefore, \eqref{sum-Ljnm} holds for $j+1$, and the proof is complete by induction.
\end{proof}

Combining Theorem~\ref{var-L} and Theorem~\ref{fin-res-G}, we can determine the kernel of the 
linearized operator associated with steady-state solutions.
If the coefficients of $\mathbf{G}$ and $\mathbf{V}$ in \eqref{V-def} and \eqref{res-G-def} coincide (i.e. $\mathbf{G}=\mathbf{V}$), then 
$\mathbf{Q}_{t}=0$ by \eqref{equ-G-condition-gen} and $\mathbf{G}_{t}=[\mathbf{V},\mathbf{G}]=0$. 
We thus obtain the following corollary.
\begin{corollary}\label{ker-lin-pro}
    Suppose that $\mathbf{Q}$ is a steady-state solution of \eqref{equ-G-condition-gen}, and let 
    $\mathbf{V}$ be given by \eqref{V-def}. Then the function $\mathrm{ad}_{\mathbf{b}}\mathbf{V}^{\perp}$
    spans the kernel of the linearized operator appearing on the right-hand side of \eqref{G-t-linear-pro-gen}, that is,
    \begin{equation*}
        \mathrm{Ker}\left(\sum_{n=0}^{N}\beta_{n}
        \frac{\delta \mathbf{L}_{n+1}^{\perp}}{\delta \mathbf{Q}}\right)=
            \left\{\mathrm{ad}_{\mathbf{b}}\mathbf{V}^{\perp}(\mathbf{Q}):
            \sum_{n=0}^{N}\beta_{n}\mathbf{L}_{n+1}^{\perp}(\mathbf{Q})=0\right\}.
    \end{equation*}
\end{corollary}

\begin{proof}
    It suffices to show that every function solves the equation
    \begin{equation*}
        \sum_{n=0}^{N}\beta_{n}
        \frac{\delta \mathbf{L}_{n+1}^{\perp}}{\delta \mathbf{Q}}(\mathbf{f}) = 0
    \end{equation*}
    of the form $\mathrm{ad}_{\mathbf{b}}\mathbf{V}^{\perp}(\mathbf{Q})$. 
    This holds because the mixed flow equations and the corresponding linearized operator are of the same differential order.
    Without loss of generality, we assume $\beta_N \neq 0$.

    Since $\mathbf{L}_{n} \in \mathcal{A}_{n-1}$, the highest derivative appearing in the linearized operator is of order $N$, 
    and the mixed flow equation is also of order $N$. Therefore, the kernel of the linearized operator is of dimension 
    $N \cdot \dim(\mathcal{T}^{\perp})$, which coincides with the dimension of the family of fundamental matrix solutions (FMS) 
    to the mixed flow equations.
\end{proof}

\subsection{Examples for $\mathrm{gl}(3,\mathbb{C})$}\label{example-CNLS}
Now we apply Theorems \ref{var-L} and \ref{fin-res-G} to the case $\mathcal{U} = \mathrm{gl}(3,\mathbb{C})$.
The nondegenerate ad-invariant bilinear form is chosen as the Killing form, which in this setting is given by
\begin{equation*}
    (\mathbf{u}, \mathbf{v})_{\mathcal{U}} = \mathrm{Tr}(\mathbf{u} \mathbf{v}).
\end{equation*}
Define the adjoint action $\mathrm{Ad}_{\mathbf{a}}: \mathcal{U} \to \mathcal{U}$ by
\begin{equation*}
    \mathrm{Ad}_{\mathbf{a}}(\mathbf{u}) = \mathbf{a} \mathbf{u} \mathbf{a}^{-1}.
\end{equation*}
Then the subalgebra $\mathcal{T} \subset \mathcal{U}$ is given by the fixed point set of $\mathrm{Ad}_{\sigma_3}$:
\begin{equation*}
    \mathcal{T} = \mathcal{U}^{\mathrm{Ad}_{\sigma_3}} := \left\{ \mathbf{u} \in \mathcal{U} : \mathrm{Ad}_{\sigma_3}(\mathbf{u}) = \mathbf{u} \right\}.
\end{equation*}

Since $\mathrm{Ad}_{\sigma_3}: \mathcal{U} \to \mathcal{U}$ is an involutive automorphism, 
the pair $(\mathcal{U}, \mathrm{Ad}_{\sigma_3})$ forms a symmetric pair. The orthogonal complement 
of $\mathcal{T}$ with respect to the Killing form is given by 
$\mathcal{T}^{\perp}=\{\mathbf{u}\in\mathcal{U}:\mathrm{Ad}_{\sigma_{3}}(\mathbf{u})=-\mathbf{u}\}$. 
In fact, for any $\mathbf{u} \in \mathrm{gl}(3,\mathbb{C})$, we have the decomposition:
\begin{equation*}
    \mathbf{u}=\frac{\mathbf{u}+\mathrm{Ad}_{\sigma_{3}}(\mathbf{u})}{2}+\frac{\mathbf{u}-\mathrm{Ad}_{\sigma_{3}}(\mathbf{u})}{2},
\end{equation*}
where the first term belongs to $\mathcal{T}$ and the second to $\mathcal{T}^\perp$.
The subspaces $\mathcal{T}$ and $\mathcal{T}^\perp$ are explicitly given by 
\begin{equation*}
    \mathcal{T}=\left\{
        \begin{pmatrix}
            T_{11} & 0 & 0 \\
            0 & T_{22} & T_{23} \\
            0 & T_{32} & T_{33} \\
        \end{pmatrix}\in \mathrm{gl}(3,\mathbb{C})
    \right\},\quad 
    \mathcal{T}^{\perp}=\left\{
        \begin{pmatrix}
            0 & T_{12} & T_{13} \\
            T_{21} & 0 & 0 \\
            T_{31} & 0 & 0 \\
        \end{pmatrix}\in \mathrm{gl}(3,\mathbb{C})
    \right\}. 
\end{equation*}

Take $\mathbf{a} = \mathbf{b} = \sigma_3$ in \eqref{equ-Lx-a} and \eqref{exp-L-b}. Then the differential equation \eqref{equ-Lx-a} becomes
\begin{equation}\label{L-x}
    \mathbf{L}_{x}=[\mathbf{U},\mathbf{L}]. 
\end{equation}
The first few coefficients in the expansion of $\mathbf{L}$ in \eqref{exp-L-b} ($\mathbf{b}=\sigma_{3}$) are given by
\begin{align*}
    \mathbf{L}_{0}=&\sigma_{3}, \\
    \mathbf{L}_{1}=&-\ii\mathbf{Q}, \\
    \mathbf{L}_{2}=&\frac{1}{2}\sigma_{3}\mathbf{Q}^{2}-\frac{1}{2}\sigma_{3}\mathbf{Q}_{x}, \\
    \mathbf{L}_{3}=&\frac{\ii}{4}(\mathbf{Q}\mathbf{Q}_{x}-\mathbf{Q}_{x}\mathbf{Q})+\frac{\ii}{4}(\mathbf{Q}_{xx}-2\mathbf{Q}^{3}), \\
    \mathbf{L}_{4}=&-\frac{1}{8}\sigma_{3}
    (\mathbf{Q}_{xx}\mathbf{Q}+\mathbf{Q}\mathbf{Q}_{xx}-\mathbf{Q}_{x}^{2}-
    3\mathbf{Q}^{4})+\frac{1}{8}\sigma_{3}
    (\mathbf{Q}_{xxx}-3\mathbf{Q}_{x}\mathbf{Q}^{2}-
    3\mathbf{Q}^{2}\mathbf{Q}_{x}). 
\end{align*}
It is straightforward to verify that the adjoint map $\mathrm{ad}_{\sigma_3}$ restricted to $\mathcal{T}^\perp$ is given by 
\begin{equation*}
    \mathrm{ad}_{\sigma_{3}}=2\sigma_{3}. 
\end{equation*}
Hence, $\mathrm{ad}_{\sigma_3}$ is invertible on $\mathcal{T}^\perp$. 
By Theorem~\ref{var-L} and Theorem~\ref{fin-res-G}, we now state the following lemma, which will be used in this paper:
\begin{thm}\label{off-diagonal-variation}
    Let $\mathbf{G} = \mathbf{G}(\lambda; x,t)$ satisfy the stationary zero curvature equations
    \begin{equation}\label{G-equation-gen-U}
        \mathbf{G}_x = [\mathbf{U}, \mathbf{G}], \quad \mathbf{G}_t = [\mathbf{V}, \mathbf{G}],
    \end{equation}
    where $\mathbf{V}$ is given by \eqref{V-def} as a linear combination of $\mathbf{V}_n$ 
    in \eqref{Vn-def} with coefficients $\beta_n$. Then the projection of $\sigma_3 \mathbf{G}$ 
    onto $\mathcal{T}^{\perp}$ evolves according to
    \begin{equation}\label{G-t-linear-pro}
        (\sigma_3 \mathbf{G}^\perp)_t = -2\sigma_3 \sum_{n=0}^N \beta_n \frac{\delta \mathbf{L}_{n+1}^\perp}{\delta \mathbf{Q}}(\sigma_3 \mathbf{G}^\perp).
    \end{equation}
    
    Furthermore, if $\mathbf{G}$ is also a linear combination of $\mathbf{V}_n$ with coefficients $\alpha_m$, as defined in \eqref{res-G-def}, then the potential $\mathbf{Q}$ is a steady-state solution of the evolution equation
    \begin{equation}\label{time-evo-Q}
        \mathbf{Q}_t = -2\sigma_3 \sum_{n=0}^N \beta_n \mathbf{L}_{n+1}^\perp,
    \end{equation}
    subject to the constraint
    \begin{equation}\label{ODE-Q}
        \sum_{m=0}^M \alpha_m \mathbf{L}_{m+1}^\perp = 0.
    \end{equation}
\end{thm}
\begin{rem}\label{L-in-A}
    The condition $\mathbf{L}_{n}\in \mathcal{A}$ can be obtained by the following argument similar 
    to \cite{beals1985inverse,sattinger1985hamiltonian} although 
    $\sigma_{3}$ is not a regular element. Let $\mathbf{Q}\in C_{0}^{\infty}(\mathbb{R})$.
    Introduce the function $\mathbf{L}=l \sigma_{3} l^{-1}$ where 
    $l$ is the solution of equation 
    \begin{equation}\label{l-equ}
        \partial_{x}\phi-\ii \lambda [\sigma_{3},\phi]-\mathbf{Q}\phi=0
    \end{equation}
    satisfying $l=\mathbb{I}_{3}+\mathcal{O}(1/\lambda)$ as $\lambda\to \infty$ and $l\to \mathbb{I}_{3}$ 
    as $x\to -\infty$ and the elements
    $l_{12},l_{13},l_{21},l_{31}$ are bounded in $x$ if $\lambda\notin\mathbb{R}$ by 
    viewing $\sigma_{3}$ a block matrix in \cite{beals1984scattering}. 
    We can also take solution $r$ which is normalized by the condition $r\to \mathbb{I}_{3}$ as $x\to +\infty$. 
    The transfer matrix $\mathbf{S}$ is given by 
    \begin{equation}\label{tra-S-1}
        l(\lambda;x)=r(\lambda;x){\rm e}^{\ii \lambda \sigma_{3} x}\mathbf{S}(\lambda){\rm e}^{-\ii \lambda \sigma_{3} x}. 
    \end{equation}
    Hence the matrix $\mathbf{S}\in \mathcal{T}$ 
    by viewing that $\mathbf{S}$ is a block diagonal matrix \cite{sattinger1985hamiltonian}. Then the 
    relation \eqref{tra-S-1} can be reduced to 
    \begin{equation}\label{tra-S}
        l(\lambda;x)=r(\lambda;x)\mathbf{S}(\lambda). 
    \end{equation}
    Hence 
    $\mathbf{L}_{n} \to 0$ for $n\geq 1$ as $x\to \pm\infty$ since $\mathbf{L}=l\sigma_{3}l^{-1}\to \sigma_{3}$ as 
    $x\to -\infty$ and $\mathbf{L}=r\mathbf{S}\sigma_{3}\mathbf{S}^{-1}r^{-1}=r\sigma_{3}r^{-1}\to \sigma_{3}$ 
    as $x\to +\infty$. Then $\mathbf{L}_{n}\in \mathcal{A}$ can be obtained by induction. It is clear 
    $\mathbf{L}_{0},\mathbf{L}_{1}\in \mathcal{A}$. Now if $\mathbf{L}_{n}\in \mathcal{A}$, then $\mathbf{L}_{n+1}^{\perp}\in 
    \mathcal{A}$ by \eqref{L-off-n-gen}. By applying $\int_{\mathbb{R}}$ on the both sides on \eqref{L-diag-n-gen}, 
    we obtain
    \begin{equation*}
        \int_{\mathbb{R}}[\mathbf{Q},\mathbf{L}_{n+1}^{\perp}]\mathrm{d}x=0
    \end{equation*}
    for all $\mathbf{Q}\in C_{0}^{\infty}(\mathbb{R})$. Then $\nabla[\mathbf{Q},\mathbf{L}_{n+1}^{\perp}]=0$. By 
    the exactness of sequence \eqref{exa-A}, there exist $\mathbf{C}_{n+1}\in\mathcal{A}$ such that 
    \begin{equation*}
        \partial_{x} \mathbf{L}_{n+1}^{\pi_{0}}=[\mathbf{Q},\mathbf{L}_{n+1}^{\perp}]=\partial_{x}\mathbf{C}_{n+1}. 
    \end{equation*}
    Note that $D$ acts on $\mathcal{A}$ formally as $\partial_{x}$ acts on $C^{\infty}$ function. Then 
    $\mathbf{L}_{n+1}^{\pi_{0}}=\mathbf{C}_{n+1}+\mathrm{const}\in \mathcal{A}$. Hence $\mathbf{L}^{n+1}\in \mathcal{A}$. 
    We conclude that all $\mathbf{L}_{n}\in\mathcal{A}$. 
\end{rem}

In the following analysis, we explicitly construct squared eigenfunction matrices $\mathbf{G}$ satisfying \eqref{G-equation-gen-U} for the $N$-soliton solutions $\mathbf{Q}=\mathbf{Q}^{[N]}$.
We then show that the projection of $\sigma_{3}\mathbf{G}$ onto $\mathcal{T}^{\perp}$ yields eigenfunctions 
of the linearized CNLS equation, as described by \eqref{G-t-linear-pro} in 
Theorem \ref{off-diagonal-variation}.

\section{$N$-soliton solutions for CNLS equations}\label{sec-$N$-soliton-CNLS}
In this section, we introduce the Darboux transformation to construct $N$-soliton solutions for the CNLS 
equations. Starting from the fundamental matrix solution (FMS), we construct both the squared eigenfunction 
matrices and the corresponding squared eigenfunctions associated with CNLS equations.

The $N$-fold Darboux transformation maps a FMS $\mathbf{\Phi}^{[0]}$, which satisfies the Lax pair 
associated with the pair $(\mathbf{U}^{[0]},\mathbf{V}^{[0]})$ to a new matrix $\mathbf{\Phi}^{[N]}$ 
satisfying the Lax pair associated with $(\mathbf{U}^{[N]},\mathbf{V}^{[N]})$. The new potential 
$\mathbf{Q}^{[N]}$ can then be obtained from $\mathbf{\Phi}^{[N]}$ and the initial potential $\mathbf{Q}^{[0]}$, 
which is called B\"acklund transformation. 

Applying this transformation to the zero solution yields explicit $N$-soliton solutions of the CNLS equations. 
We present the construction of the N-fold Darboux transformation for CNLS equations in this section. The 
corresponding transformation for the CmKdV equation will be discussed in Section~\ref{DT-mkdv}.
Throughout this section, the matrix $\mathbf{V}$ refers specifically to $\mathbf{V}_{CNLS}$. 

\subsection{Darboux transformation for CNLS equations}
The Darboux transformation for the CNLS equations has the following form \cite{ling_darboux_2015,ling_darboux_2016}:
\begin{prop}\label{DT-N-CNLS}
 For the Lax pair \eqref{lax-U}--\eqref{lax-V} with $(\mathbf{U}^{[0]}(\lambda;x,t), \mathbf{V}^{[0]}(\lambda;x,t))$ 
 and the corresponding FMS $\mathbf{\Phi}^{[0]}(\lambda;x,t)$ associated with the potential 
    $\mathbf{Q}^{[0]}(x,t)$, we choose $N$ eigenfunctions $|\mathbf{y}_k\rangle$ satisfying the Lax pair 
 at distinct eigenvalues $\lambda_k \in \mathbb{C}^+$ for $k = 1, 2, \dots, N$.
 The $N$-fold Darboux matrix for the CNLS equations is given by
    \begin{equation*}
        \mathbf{D}_{r}^{[N]}(\lambda;x,t)=\mathbb{I}_{3}-\sum_{k=1}^{N}\frac{\lambda_{k}-\lambda_{k}^{*}}{\lambda-\lambda_{k}^{*}}
        |\mathbf{x}_{k}\rangle \langle \mathbf{y}_{k}|
    \end{equation*} 
 where the vectors $|\mathbf{x}_k\rangle$ and $|\mathbf{y}_k\rangle$ are related through
        \begin{equation}
 (|\mathbf{y}_1\rangle, |\mathbf{y}_2\rangle,\cdots, |\mathbf{y}_N\rangle)
 =(|\mathbf{x}_1\rangle, |\mathbf{x}_2\rangle,\cdots, |\mathbf{x}_N\rangle)\mathbf{M}, \quad 
        \mathbf{M}=\left(
        \frac{\lambda_{k}-\lambda_{k}^{*}}{\lambda_{l}-\lambda_{k}^{*}}\langle \mathbf{y}_k|\mathbf{y}_l\rangle
 \right)_{1\leq k,l\leq N}.
 \label{matrix-M}
    \end{equation}
 Here, $\langle \mathbf{x}_k|=(|\mathbf{x}_k\rangle)^{\dag}$ and $\langle \mathbf{y}_k|=(|\mathbf{y}_k\rangle)^{\dag}$.
 Applying the $N$-fold Darboux transformation to the FMS $\mathbf{\Phi}^{[0]}(\lambda;x,t)$ yields the new FMS
    \begin{equation*}
        \mathbf{\Phi}_{r}^{[N]}(\lambda;x,t)=\mathbf{D}_{r}^{[N]}(\lambda;x,t)\mathbf{\Phi}^{[0]}(\lambda;x,t)
    \end{equation*}
 which satisfies the Lax pair \eqref{lax-U}--\eqref{lax-V} associated with pair $(\mathbf{U}^{[N]},\mathbf{V}^{[N]})$. 
 The corresponding B\"acklund transformation is given by
    \begin{equation}\label{Ba}
        \mathbf{Q}^{[N]}=\mathbf{Q}^{[0]}+2\ii\sigma_{3}\sum_{k=1}^{N}
        (\lambda_{k}-\lambda_{k}^{*})
        (|\mathbf{x}_k\rangle \langle \mathbf{y}_k|)^{\perp}.
        \end{equation}
\end{prop}

If we take the unbounded vectors $|\mathbf{y}_k\rangle$ as in Proposition \ref{DT-N-CNLS}, then the spectrum of the new Lax pair consists of that of the original Lax pair together with $N$ additional, distinct eigenvalues. In the case of the zero potential $\mathbf{Q}^{[0]}=\mathbf{0}$, the FMS corresponding to the Lax pair is given by
\begin{equation*}
    \mathbf{\Phi}^{[0]}={\rm e}^{{\rm i}\lambda(x+2\lambda t)\sigma_{3}}. 
\end{equation*}
Each vector $|\mathbf{y}_k\rangle$ is a linear combination of the columns of the FMS:
\begin{equation*}
    |\mathbf{y}_k\rangle=\mathbf{\Phi}^{[0]}(\lambda_{k};x,t)c^{[k]}
    ={\rm e}^{{\rm i}\lambda_{k}(x+2\lambda_{k} t)\sigma_{3}}\begin{pmatrix}
        1 \\ \mathbf{c}_{k}
    \end{pmatrix},\quad k=1,2,\cdots, N
\end{equation*}
where 
\begin{equation*}
    \mathbf{c}_{k}=(c_{1k},c_{2k})^{T}\in\mathbb{C}^{2}\backslash \{(0,0)\}.
\end{equation*}
To eliminate the singularities at the point spectrum in the spectral parameter $\lambda$, we consider the Darboux transformation of the form
\begin{equation*}
    \mathbf{D}^{[N]}(\lambda;x,t)=\mathcal{P}(\lambda)\mathbf{D}_{r}^{[N]}(\lambda;x,t)
\end{equation*}
i.e.
\begin{equation}\label{DT-Nsoliton}
    \begin{split}
        &\mathbf{D}^{[N]}(\lambda;x,t)
    \\=&\mathcal{P}(\lambda)-\sum_{s,r=1}^{N}\frac{\mathcal{P}(\lambda)}{\lambda-\lambda_{r}^{*}}
 (\lambda_{r}-\lambda_{r}^{*})
 m_{sr}{\rm e}^{{\rm i}\lambda_{s}(x+2\lambda_{s}t)\sigma_{3}}\begin{pmatrix}
        1 \\ \mathbf{c}_{s}
    \end{pmatrix}
    \begin{pmatrix}
        1 & \mathbf{c}_{r}^{\dagger}
    \end{pmatrix}
 {\rm e}^{-{\rm i}\lambda_{r}^{*}(x+2\lambda_{r}^{*}t)\sigma_{3}}
    \end{split}
\end{equation}
where 
\begin{equation}\label{P-def}
    \mathcal{P}(\lambda)=\prod_{k=1}^{N}(\lambda-\lambda_{k}^{*})
\end{equation}
and the matrix $\mathbf{m}=(m_{sr})$ is the inverse of matrix $\mathbf{M}=(M_{kl})$ defined in \eqref{matrix-M}.

Then the FMS associated with the $N$-soliton solution is given by
\begin{equation}\label{FMS-Nsoliton}
    \mathbf{\Phi}^{[N]}(\lambda;x,t)=\mathbf{D}^{[N]}(\lambda;x,t){\rm e}^{{\rm i}\lambda(x+2\lambda t)\sigma_{3}}
\end{equation}
which is analytic at $\lambda=\lambda_{k}^{*}$. The $N$-soliton solution is then obtained 
by applying the B\"acklund transformation \eqref{Ba} together with the explicit form of the 
FMS \eqref{DT-Nsoliton}:
\begin{equation*}
    \mathbf{q}^{[N]}(x,t;\mathbf{\Lambda},\mathbf{c})=4\sum_{s,r=1}^{N}\mathrm{Im}(\lambda_{r})
 m_{sr}{\rm e}^{-{\rm i}\lambda_{s}(x+2\lambda_{s}t)}
 {\rm e}^{-{\rm i}\lambda_{r}^{*}(x+2\lambda_{r}^{*}t)}\mathbf{c}_{s},
\end{equation*}
which can be rewritten in the matrix form as in \eqref{CNLS-Nsoliton}. 
Recall that we write the spectral parameter as $\lambda_k = a_k + \mathrm{i}b_k$.
Let $\mathbf{\Phi}_i^{[N]}$ denote the $i$-th column of the FMS $\mathbf{\Phi}^{[N]}$ defined in \eqref{FMS-Nsoliton}.
Consider the Lax operator
\begin{equation*}
    \mathcal{L}_{s}=-{\rm i}\sigma_{3}(\partial_{x}-\mathbf{Q}^{[N]}), 
\end{equation*}
for which the FMS $\mathbf{\Phi}^{[N]}$ satisfies the ODE
\begin{equation*}
    \mathcal{L}_{s}\mathbf{\Phi}^{[N]}(\lambda;x,t)=\lambda\mathbf{\Phi}^{[N]}(\lambda;x,t). 
\end{equation*}
Since $\mathcal{L}_s$ is a first-order differential operator, all solutions of the spectral problem associated with $\mathcal{L}_s$ can be obtained from the fundamental matrix solution $\mathbf{\Phi}^{[N]}$.
Denote by $\sigma(\mathcal{A})$ the spectrum of an operator $\mathcal{A}$, and let $\sigma_{\mathrm{point}}(\mathcal{A})$ and $\sigma_{\mathrm{ess}}(\mathcal{A})$ denote its point and essential spectra, respectively.
Then the following lemma concerning the Lax spectrum $\sigma(\mathcal{L}_s)$ holds:
\begin{lem}[Lax spectrum for $N$-solitons]\label{lem-Lax-spectrum}
 Consider the spectral problem
    \begin{equation*}
        \mathcal{L}_{s}\mathbf{\Phi}(\lambda;x,t)=\lambda\mathbf{\Phi}(\lambda;x,t)
    \end{equation*}
 in the space $L^2(\mathbb{R}; \mathbb{C}^3)$, where the spectral parameters $\lambda_k = a_k + \mathrm{i} b_k \in \mathbb{C}^+$ are distinct. Then 
 the essential spectrum of the Lax operator is
    \begin{equation*}
        \sigma_{ess}(\mathcal{L}_{s})=\mathbb{R}
    \end{equation*}
 and the point spectrum is given by
    \begin{equation*}
        \sigma_{point}(\mathcal{L}_{s})=\{\lambda_{k},\lambda_{k}^{*}:k=1,2,\cdots, N\}. 
    \end{equation*}
 Moreover, for each $\lambda \in \mathbb{R}$, the three columns of $\mathbf{\Phi}^{[N]}(\lambda)$ 
satisfy the spectral problem of the Lax operator associated with the essential spectrum 
and form a fundamental system of $L^{\infty}$ solutions.
 For the point spectrum, the eigenspaces at $\lambda = \lambda_k$ and $\lambda = \lambda_k^*$ are one-dimensional, given by
    \begin{align*}
        \mathrm{Ker}(\lambda_k \mathbb{I} - \mathcal{L}_s) &= \mathrm{span} \left\{ \mathbf{\Phi}^{[N]}_1(\lambda_k) \right\}
        \subset \mathcal{S}(\mathbb{R}; \mathbb{C}^3), \\
        \mathrm{Ker}(\lambda_k^* \mathbb{I} - \mathcal{L}_s) &= \mathrm{span} \left\{ \mathbf{\Phi}^{[N]}_1(\lambda_k^*) \right\}
        \subset \mathcal{S}(\mathbb{R}; \mathbb{C}^3).
    \end{align*}
\end{lem}
\begin{proof}
 The essential spectrum of $\mathcal{L}_{s}$ can be determined using Weyl’s essential spectrum theorem: 
    \begin{equation*}
        \sigma_{ess}(\mathcal{L}_{s})=\sigma_{ess}(-{\rm i}\sigma_{3}\partial_{x})=\mathbb{R}. 
    \end{equation*}
 Since the geometric multiplicity of each eigenvalue in the point spectrum $\sigma_{point}(\mathcal{L}_{s})$ is one, it suffices to prove 
 that $\mathbf{\Phi}^{[N]}_{1}(\lambda_{k})$ and $\mathbf{\Phi}^{[N]}_{1}(\lambda_{k}^{*})$ are nonzero in $L^{2}$. 
    The regular FMS matrix satisfies the relation
    \begin{equation*}
        \mathbf{\Phi}^{[N]}(\lambda_{k})(1, c_{1k},c_{2k})^{T}=0
    \end{equation*}
    that is,
    \begin{equation}\label{Rel-Phi-lambda}
        c_{1k}\mathbf{\Phi}^{[N]}_{2}(\lambda_{k})+c_{2k}\mathbf{\Phi}^{[N]}_{3}(\lambda_{k})+\mathbf{\Phi}^{[N]}_{1}(\lambda_{k})=0. 
    \end{equation} 
    Since $\mathbf{\Phi}^{[N]}_{i+1}(\lambda_{k}),i=1,2$ decays exponentially as $x \to -\infty$, and the 
    scattering parameters $\mathbf{c}_{k}=(c_{1k},c_{2k})\ne(0,0)$, it follows from the relation 
    \eqref{Rel-Phi-lambda} that $\mathbf{\Phi}^{[N]}_{1}(\lambda_{k})$ also decays exponentially as $x \to -\infty$.
    As $x\to +\infty$, the function $\mathbf{\Phi}^{[N]}_{1}(\lambda_{k})$ also exhibits exponential decay 
    due to the formula \eqref{FMS-Nsoliton}, since the Darboux matrix $\mathbf{D}^{[N]}(\lambda; x, t)$ remains bounded in $x$.
    Therefore, $\mathbf{\Phi}^{[N]}_{1}(\lambda_{k})$ belongs to the Schwartz class.
    For the eigenfunction corresponding to $\lambda=\lambda_{k}^{*}$, note that
    \begin{equation*}
        \mathrm{Ker}(\mathbf{\Phi}^{[N]}(\lambda_{k}^{*}))=\mathrm{Ker}(\mathbf{\Phi}^{[N]}(\lambda_{k}))^{\perp}
        =\{(1,\mathbf{c}_{k}^{T})^{T}\}^{\perp}. 
    \end{equation*}
 This implies the identity
    \begin{equation}\label{Rel-Phi-lambdac}
        -c_{ik}^{*}\mathbf{\Phi}^{[N]}_{1}(\lambda_{k}^{*})+\mathbf{\Phi}^{[N]}_{i+1}(\lambda_{k}^{*})=0,
        \quad i=1,2.
    \end{equation}
    By an argument analogous to that for $\mathbf{\Phi}^{[N]}_{1}(\lambda_{k})$, we conclude that 
    $\mathbf{\Phi}^{[N]}_{1}(\lambda_{k}^{*})\in \mathcal{S}(\mathbb{R}; \mathbb{C}^3)$. This completes the proof.
\end{proof}
To study the nonlinear stability of $N$-soliton solutions, we construct the squared eigenfunction matrices and the associated squared eigenfunctions in this section. These functions will play a key role in the nonlinear stability analysis, particularly for spectral parameters $\lambda$ on the Lax spectrum.

When we consider the squared eigenfunctions associated with the point spectrum of the Lax operator, linear dependencies arise among them due to Lemma \ref{lem-Lax-spectrum}. The identities \eqref{Rel-Phi-lambda} and \eqref{Rel-Phi-lambdac} are employed to determine a maximal linearly independent subset.

The Lyapunov functional is central to the nonlinear stability analysis of $N$-soliton solutions. The variational characterization, derived from the differential of the Lyapunov functional via the trace formula, is introduced before the squared eigenfunction matrices.
\subsection{The variational characterization for $N$-soliton solutions}\label{var-Nsoliton}
The $N$-soliton solution satisfies a semilinear 
ODE of order $2N$, which arises as the differential of 
the Lyapunov functional via the trace formula. The variation of the Lyapunov functional under perturbations 
of the $N$-soliton can be controlled by the second-order term in its expansion at the $N$-soliton profile, 
as the first-order term vanishes.

The trace formula \cite{faddeev1987hamiltonian} is fundamental in constructing the Lyapunov functional. Since the $N$-soliton is 
parameterized by the spectral parameters $\Lambda$ and the scattering parameters $\mathbf{c}$, 
the conserved quantities can be expressed as polynomials in the spectral parameters and are independent 
of the scattering parameters. Consequently, the variation of the conserved quantities depends polynomially 
on the variation of the spectral parameters. This observation implies that a certain linear combination of 
the variations of the conserved quantities must vanish. For further details, see \cite{ling2024stability}.

The polynomial $\mathcal{P}(\lambda)$ is defined in \eqref{P-def}, and we introduce
\begin{equation}\label{hat-P}
    \hat{\mathcal{{P}}}(\lambda)=\mathcal{P}^{*}(\lambda^{*})=\prod_{k=1}^{N}(\lambda-\lambda_{k}).
\end{equation}
The Lyapunov functional $\mathcal{I}(\mathbf{q})$ for the $N$-soliton solution is given by \eqref{Lya}, where the coefficients 
$\mu_{n}$ are determined by the identity
\begin{equation}\label{P-hatP}
    \mathcal{P}(\lambda)\hat{\mathcal{P}}(\lambda)=\sum_{n=0}^{2N}2^{n-2N}\mu_{n}\lambda^{n}. 
\end{equation}
We note that $\mu_{n}$ are real since $(\mathcal{P}(\lambda)\hat{\mathcal{P}}(\lambda))^{*}=
\mathcal{P}(\lambda^{*})\hat{\mathcal{P}}(\lambda^{*})$. 
The generating function for the conserved quantities is given by
\begin{equation*}
    \ln a(\lambda)=\int_{\mathbb{R}}\mathbf{q}^{\dagger}(x)\omega(\lambda;x,t)\mathrm{d}x
\end{equation*}
where $\omega(\lambda;x,t)$ satisfies the Riccati equation
\begin{equation*}
    \omega_{x}=\mathbf{q}-2\ii \lambda \omega +\omega \mathbf{q}^{\dagger} \omega
\end{equation*}
with the expansion
\begin{equation*}
    \omega(\lambda;x,t)=\sum_{n=1}^{+\infty} \frac{\omega_{n}(x,t)}{(2\ii \lambda)^{n}}. 
\end{equation*}
The conserved quantities are encoded in the generating function
\begin{equation*}
    \ln a(\lambda)=-2\ii\sum_{n=0}^{+\infty}\frac{ \mathcal{H}_{n}}{(2\lambda)^{n+1}},
\end{equation*}
which yields the explicit formula
\begin{equation}\label{conserved-quantities}
    \mathcal{H}_{n}=\frac{(-\ii)^{n}}{2}\int_{\mathbb{R}}\mathbf{q}^{\dagger}\omega_{n+1} \mathrm{d}x. 
\end{equation}
The first few terms in the expansion of $\omega(\lambda;x,t)$ are given by
\begin{align*}
    \omega_{1}&=\mathbf{q}, \\
    \omega_{2}&=-\mathbf{q}_{x}, \\
    \omega_{3}&=\mathbf{q}_{xx}+|\mathbf{q}|^{2}\mathbf{q},\\
    \omega_{4}&=-\mathbf{q}_{xxx}-\mathbf{q}\mathbf{q}^{\dagger}_{x}\mathbf{q}-|\mathbf{q}|^{2}\mathbf{q}_{x}
 -\mathbf{q}\mathbf{q}^{\dagger}\mathbf{q}_{x}. 
\end{align*}
The conserved quantities in \eqref{H0}–\eqref{H3} can be directly derived from the general formula \eqref{conserved-quantities}.
\begin{rem}
 The conserved quantities \eqref{conserved-quantities} can also be derived from the coefficients of 
 the matrix $\mathbf{L}$ defined in \eqref{L-x} in \cite{terng1997soliton}, through the identity
    \begin{equation*}
        \mathcal{H}_{n}=-2^{n-1}\ii \int_{\mathbb{R}}\mathrm{Tr}\left(
            \int_{0}^{1}\mathbf{L}_{n+1}^{\bot}(t\mathbf{Q})\mathbf{Q}\mathrm{d}t\,
 \right)\mathrm{d}x=-
        \frac{2^{n-1}}{n+1}\int_{\mathbb{R}}\mathrm{Tr}(\mathbf{L}_{n+2}\sigma_{3})\mathrm{d}x. 
    \end{equation*}
 Moreover, the variational derivative of $\mathcal{H}_{n}$ with respect to $\mathbf{Q}$ is given by
    \begin{equation*}
       \frac{\delta \mathcal{H}_{n}}{\delta \mathbf{Q}} =-2^{n-1}\ii \mathbf{L}_{n+1}^{\bot}.
    \end{equation*}
 The variation is given by 
    \begin{equation}\label{def-var-Q}
        \frac{d }{d \epsilon}\mathcal{H}_{n}(\mathbf{Q}+\epsilon\delta \mathbf{Q})|_{\epsilon=0} =\int_{\mathbb{R}}\mathrm{Tr}\left(\frac{\delta \mathcal{H}_{n}}{\delta \mathbf{Q}} \delta \mathbf{Q}\right)\mathrm{d}x. 
    \end{equation}
 Further details can be found in \cite{sattinger1985hamiltonian}, but with slightly 
 different calculation since $\sigma_{3}$ is not a regular element. The function $\ln \mathbf{S}$ is 
 given by $\mathbf{P}\mathrm{diag}(\ln \lambda_{1},\ln \lambda_{2}, \ln \lambda_{3})\mathbf{P}^{-1}$ for 
    $\mathbf{S}=\mathbf{P}\mathrm{diag}(\lambda_{1},\lambda_{2},\lambda_{3})\mathbf{P}^{-1}$. Since 
    $\mathbf{S}\in \mathcal{T}$, we can take $\mathbf{P}\in\mathcal{T}$. The 
 generating function is given by
    \begin{equation}\label{gen-Tr-lns}
        \mathcal{H}^{g}(\lambda)=\mathrm{Tr}(\sigma_{3}\ln \mathbf{S}). 
    \end{equation}
 By \eqref{tra-S}, we obtain 
    \begin{equation*}
        \ln \mathbf{S}(\lambda)=\lim_{x\to +\infty}\ln r \mathbf{S}=\lim_{x\to +\infty}\ln l=\int_{\mathbb{R}}
        \partial_{x} (\ln l) \mathrm{d}x. 
    \end{equation*}
 Since $\delta l=r\delta \mathbf{S}+\delta r\mathbf{S}\to \delta \mathbf{S}$ as $x\to +\infty$, we obtain
    \begin{equation*}
        \delta \ln \mathbf{S}=\mathbf{S}^{-1}\delta \mathbf{S}=\lim_{x\to +\infty} l^{-1}\delta l=
        \int_{\mathbb{R}}\partial_{x}(l^{-1}\delta l) \mathrm{d}x. 
    \end{equation*}
 Since $l$ satisfies \eqref{l-equ}, we obtain 
    \begin{equation*}
        \partial_{x}(l^{-1}\delta l)=\ii\lambda [\sigma_{3},l^{-1}\delta l]+l^{-1} \delta \mathbf{Q}l. 
    \end{equation*}
 Then 
    \begin{equation*}
        \delta \mathcal{H}^{g}=\int_{\mathbb{R}}
        \mathrm{Tr}(\sigma_{3}(\ii\lambda [\sigma_{3},l^{-1}\delta l]+l^{-1} \delta \mathbf{Q}l))
        \mathrm{d}x=\int_{\mathbb{R}}\mathrm{Tr}(l\sigma_{3}l^{-1} \delta \mathbf{Q})\mathrm{d}x=
        \int_{\mathbb{R}}\mathrm{Tr}(\mathbf{L} \delta \mathbf{Q})\mathrm{d}x,
    \end{equation*}
 hence 
    \begin{equation*}
        \frac{\delta \mathcal{H}^{g}}{\delta \mathbf{Q}}= \mathbf{L}^{\perp}=\sum_{n=0}^{\infty}\frac{\mathbf{L}_{n+1}^{\perp}}{\lambda^{n+1}}. 
    \end{equation*}
 The connection between generating function $\mathcal{H}^{g}$ and $\ln a(\lambda)$ is 
    \begin{equation*}
        \mathcal{H}^{g}=2\ln a(\lambda)
    \end{equation*}
 since $a(\lambda)$ be the $(1,1)$ element of the transfer matrix and 
    \begin{equation*}
       \mathrm{Tr}\ln \mathbf{S}=\ln \det \mathbf{S}=\ln 1 =0. 
    \end{equation*}
\end{rem}
By the trace formula, 
for the $N$-soliton solutions which can be characterized by the spectral parameters and scattering parameters,
the conserved quantities are given by the spectral parameters
\begin{equation}\label{H-spectral}
    \mathcal{H}_{n}=\frac{2^{n+1}}{n+1}\sum_{k=1}^{N}\mathrm{Im} \lambda_{k}^{n+1}.
\end{equation}
The variation of $H_{n}$ with respect to $\mathbf{q}$ can be translated to the variation of $\lambda_{k},\lambda_{k}^{*}$
\begin{equation}\label{var-Hn-1-spectral}
    \frac{\delta \mathcal{H}_{n}}{\delta \mathbf{q}}=-2^{n}\ii \sum_{k=1}^{N}\left(
        \lambda_{k}^{n}\frac{\delta \lambda_{k}}{\delta \mathbf{q}}-
        (\lambda_{k}^{*})^{n}\frac{\delta \lambda_{k}^{*}}{\delta \mathbf{q}}
    \right). 
\end{equation}
An immediate consequence is that the $N$-soliton solutions satisfy the ODE
\begin{equation*}
    \frac{\delta\mathcal{I}}{\delta \mathbf{q}}(\mathbf{q}^{[N]})=\sum_{n=0}^{2N} \mu_{n} \frac{\delta\mathcal{H}_{n}}{\delta \mathbf{q}}(\mathbf{q}^{[N]})=0. 
\end{equation*}
Now we come back to 
the conserved quantities of the form \eqref{conserved-quantities}. 
The variation of the conserved quantities can be divided into the linear term and the nonlinear term $R_{n}$
\begin{equation*}\label{var-Hn-1}
    \frac{\delta \mathcal{H}_{n}}{\delta\mathbf{q}}(\mathbf{q})=({\ii}\partial_{x})^{n}\mathbf{q}+R_{n}(\mathbf{q},\mathbf{q}^{*},\cdots)
\end{equation*}
where $R_{n}$ is the polynomial with respect to $\mathbf{q},\mathbf{q}^{*}$ and their derivatives and 
the lowest-order term of $R_{n}$ is of degree $3$. 
The following lemma holds:
\begin{lem}\label{lem-critical-point}
 The Lyapunov functional $\mathcal{I}(\mathbf{q})$, which is given in \eqref{Lya}, is independent of time. 
 The $N$-soliton solutions solve the equations
    \begin{equation}\label{ODE-Nsoliton}
        \frac{\delta \mathcal{I}}{\delta\mathbf{q}}(\mathbf{q})=0
    \end{equation}
 which is a semi-linear $2N$-order ODE for vector function $\mathbf{q}=(q_{1},q_{2})^{T}$. 
 Moreover, all solutions are $N$-soliton solutions to the ODE \eqref{ODE-Nsoliton} with boundary condition 
    $\mathbf{q}\to 0$ if $|x|\to \infty$. 
\end{lem}
\begin{proof}
 The equation \eqref{ODE-Nsoliton} can be rewritten as a semilinear ODE:
    \begin{equation}\label{ODE-Nsoliton-1}
        2^{2N}\mathcal{P}\left(\frac{\ii\partial_{x}}{2}\right)
        \hat{\mathcal{P}}\left(\frac{\ii\partial_{x}}{2}\right)(\mathbf{q}) + R(\mathbf{q}) = 0,
    \end{equation}
 where the remainder term $ R(\mathbf{q}) $ consists of nonlinear terms that are at least cubic 
 in $ \mathbf{q} $ and its derivatives.
 Consider first the linearized homogeneous equation:
    \begin{equation*}
        \mathcal{P}\left(\frac{\ii\partial_{x}}{2}\right)
        \hat{\mathcal{P}}\left(\frac{\ii\partial_{x}}{2}\right)(\mathbf{q}) = 0,
    \end{equation*}
 whose fundamental solutions are exponentials of the form
    \begin{equation}
        \mathrm{e}^{-2\ii \lambda_{k} x}\mathbf{e}_i, \quad 
        \mathrm{e}^{-2\ii \lambda_{k}^{*} x}\mathbf{e}_i, \quad k=1,\dots,N,\ i=1,2,
    \end{equation}
 giving a total of $ 4N $ linearly independent solutions. Among these, at spatial infinity 
 ($ x\to\pm\infty $), only $ 2N $ of them decay.
 By standard ODE theory, the space of solutions of the full nonlinear equation \eqref{ODE-Nsoliton-1} 
 decaying at both spatial infinities has dimension at most $ 2N $. The nonlinear term 
    $ R(\mathbf{q}) $ does not affect the asymptotic decay rate at leading order, due to its 
 higher nonlinearity.

 On the other hand, the family of $ N $-soliton solutions forms a smooth manifold of
 dimension $ 2N $, parameterized by 
 scattering parameters $c_{1k},c_{2k},k=1,2,\cdots N$.
 Therefore, all sufficiently smooth, spatially localized solutions of the equation must lie in 
 the $ N $-soliton manifold.
 This completes the proof.
\end{proof}
\begin{rem}
 By Lemma~\ref{lem-critical-point}, the 1-soliton solution satisfies the second-order ODE
    \begin{equation*}
        \mathbf{q}_{xx} + 4a_{1} \ii \mathbf{q}_{x} - 4(a_{1}^{2} + b_{1}^{2}) \mathbf{q} + 2|\mathbf{q}|^{2} \mathbf{q} = 0,
    \end{equation*}
 which can also be obtained directly from the CNLS equation by separating variables $ x + 4a_1 t $ and $ t $.

 In contrast, the ODE satisfied by the 2-soliton solution is significantly more involved:
    \begin{equation*}
        \begin{split}
            &\mathbf{q}_{xxxx} + 4|\mathbf{q}|^{2} \mathbf{q}_{xx} +
            2 \mathbf{q} \mathbf{q}_{xx}^{\dagger} \mathbf{q} +
            4 \mathbf{q} \mathbf{q}^{\dagger} \mathbf{q}_{xx} +
            2 \mathbf{q}_{x} \mathbf{q}_{x}^{\dagger} \mathbf{q} +
            6 \mathbf{q}_{x} \mathbf{q}^{\dagger} \mathbf{q}_{x} +
            2|\mathbf{q}_{x}|^{2} \mathbf{q}+ 6|\mathbf{q}|^{4} \mathbf{q}  \\
& + 4\ii(a_{1} + a_{2}) \left( \mathbf{q}_{xxx} + 3|\mathbf{q}|^{2} \mathbf{q}_{x} + 3 \mathbf{q} \mathbf{q}^{\dagger} \mathbf{q}_{x} \right) 
 - 4(a_{1}^{2} + b_{1}^{2} + a_{2}^{2} + b_{2}^{2} + 4a_{1}a_{2}) \left( \mathbf{q}_{xx} + 2|\mathbf{q}|^{2} \mathbf{q} \right)  \\
 &- 16\left( a_{1}(a_{2}^{2} + b_{2}^{2}) + a_{2}(a_{1}^{2} + b_{1}^{2}) \right) \ii \mathbf{q}_{x}
 + 16(a_{1}^{2} + b_{1}^{2})(a_{2}^{2} + b_{2}^{2}) \mathbf{q} 
 = 0.
        \end{split}
    \end{equation*}
 If $a_{1}=a_{2}=a$, then the above ODE is consistent with \cite{ling2024stability}. 
\end{rem}
By Lemma~\ref{lem-critical-point}, for any $N$-soliton solution $\mathbf{q}(x,t; \mathbf{\Lambda}, \mathbf{c})$, 
there exists a renormalized parameter $\tilde{\mathbf{c}} \in \mathbb{C}^{2 \times N}$ such that
\begin{equation*}
    \mathbf{q}(x,t; \mathbf{\Lambda}, \mathbf{c}) = \mathbf{q}(x,0; \mathbf{\Lambda}, \tilde{\mathbf{c}}).
\end{equation*}
In fact, the renormalized parameters are given explicitly by 
$\tilde{\mathbf{c}}_k = \mathrm{e}^{-4\ii \lambda_k^2 t} \mathbf{c}_k$.
Therefore, it suffices to consider the $N$-soliton solutions at $ t = 0 $ when analyzing the spectrum 
of the second variation of the Lyapunov functional. 

We now proceed to introduce the squared eigenfunction matrices, which form the foundation for the upcoming spectral analysis.

\subsection{The squared eigenfunction matrices for CNLS equations}
In this subsection, we construct the squared eigenfunction matrices and the associated squared eigenfunctions 
for the CNLS equations. The construction for the CmKdV equations will be given separately in 
Section~\ref{squ-eig-cmkdv}. The squared eigenfunction matrices can be obtained using solutions to the Lax pair and its adjoint. 

In the absence of any symmetry (i.e., no relation between $\mathbf{r}$ and $\mathbf{q}$), the squared 
eigenfunction matrices can be constructed directly from the fundamental matrix solution
$\mathbf{\Phi}$ and its inverse. When a symmetry between $\mathbf{r}$ and $\mathbf{q}$ exists, 
the inverse $\mathbf{\Phi}^{-1}$ can be expressed in terms of $\mathbf{\Phi}$ via the symmetry relation\cite{yang_nonlinear_2010,ling2023stability}. 

In this paper, the potential matrix satisfies the symmetry
\begin{equation*}
    \mathbf{Q}^{\dagger} = -\mathbf{Q},
\end{equation*}
which leads to the following symmetry relations for the Lax pair matrices $(\mathbf{U}, \mathbf{V})$:
\begin{equation}\label{sym-U-V}
    \mathbf{U}^{\dagger}(\lambda^{*}) = -\mathbf{U}(\lambda), \quad 
    \mathbf{V}^{\dagger}(\lambda^{*}) = -\mathbf{V}(\lambda).
\end{equation}
Let $\mathbf{\Phi}(\lambda)$ be a fundamental matrix solution of the CNLS Lax pair \eqref{lax-U}--\eqref{lax-V}. Then both $\mathbf{\Phi}^{-1}(\lambda)$ and $\mathbf{\Phi}^{\dagger}(\lambda^{*})$ satisfy the adjoint Lax pair:
\begin{align*}
    \partial_x \mathbf{\Psi}(\lambda;x,t) &= -\mathbf{U}(\lambda;x,t)\, \mathbf{\Psi}(\lambda;x,t), \\
    \partial_t \mathbf{\Psi}(\lambda;x,t) &= -\mathbf{V}(\lambda;x,t)\, \mathbf{\Psi}(\lambda;x,t).
\end{align*}
By the uniqueness of solutions to the ODE system, it follows that
\begin{equation*}
    \mathbf{\Phi}^{-1}(\lambda;x,t)
 = \mathbf{\Phi}^{\dagger}(\lambda^{*};x,t)\,
    \mathbf{\Phi}^{\dagger}(\lambda^{*};0,0)^{-1}\,
    \mathbf{\Phi}^{-1}(\lambda;0,0).
\end{equation*}
Therefore, the adjoint solution $\mathbf{\Phi}^{\dagger}(\lambda^{*})$ can be used to construct the squared eigenfunction matrices.

The squared eigenfunction matrix for the CNLS equations is defined by
\begin{equation}\label{squared-eigenfunction-matrix}
    \begin{split}
 &p_{i}(\mathbf{\Phi})(\lambda) = \big(\mathbf{\Phi}(\lambda)\big)_{1} \cdot \big(\mathbf{\Phi}^{\dagger}(\lambda^{*})\big)^{i+1}, \\
 &p_{-i}(\mathbf{\Phi})(\lambda) = \big(\mathbf{\Phi}(\lambda)\big)_{i+1} \cdot \big(\mathbf{\Phi}^{\dagger}(\lambda^{*})\big)^{1},
    \end{split}
\end{equation}
for $ i = 1, 2 $, where $(\cdot)_{j}$ denotes the $j$-th column and $(\cdot)^{k}$ denotes the $k$-th row of a matrix.  
The associated squared eigenfunctions are given by the off-diagonal entries of these squared eigenfunction matrices
    \begin{equation}\label{squared-eigenfunctions}
        \begin{split}
            &s_{i}(\mathbf{\Phi})=
            \begin{pmatrix}
                \phi_{21}\hat{\phi}_{i+1,1},
                \phi_{31}\hat{\phi}_{i+1,1}, 
                -\phi_{11}\hat{\phi}_{i+1,2},
                -\phi_{11}\hat{\phi}_{i+1,3}
            \end{pmatrix}^{T},\quad \\
            &s_{-i}(\mathbf{\Phi})=\begin{pmatrix}
                \phi_{2,i+1}\hat{\phi}_{11} ,
                \phi_{3,i+1}\hat{\phi}_{11} ,
                -\phi_{1,i+1}\hat{\phi}_{12},
                -\phi_{1,i+1}\hat{\phi}_{13}
            \end{pmatrix}^{T},
        \end{split}
\end{equation}
where $\mathbf{\Phi}=(\phi_{ij})_{1\leq i,j \leq 3}$ and $\mathbf{\Phi}^{\dagger}(\lambda^{*})=(\hat{\phi}_{ij})_{1\leq i,j \leq 3}$. 
The squared eigenfunction matrices satisfy the symmetry
\begin{equation}\label{symmetric-squared-eigenfunction-matrix}
 p_{i}(\mathbf{\Phi})(\lambda) = p_{-i}(\mathbf{\Phi})^{\dagger}(\lambda^{*}),
\end{equation}
which implies that the corresponding squared eigenfunctions obey the symmetry
\begin{equation}\label{symmetric-squared-eigenfunctions}
 s_{i}(\mathbf{\Phi})(\lambda) = -\mathbf{\Sigma} \left(s_{-i}(\mathbf{\Phi})(\lambda^{*})\right)^{*},
\end{equation}
where
\begin{equation*}
    \mathbf{\Sigma} = 
    \begin{pmatrix}
        \mathbf{0}_{2\times 2} & \mathbb{I}_{2} \\
        \mathbb{I}_{2} & \mathbf{0}_{2\times 2}
    \end{pmatrix}.
\end{equation*}

Since the squared eigenfunction matrices are constructed from the column of $\mathbf{\Phi}(\lambda)$, which satisfies the Lax pair, and the row of $\mathbf{\Phi}^{\dagger}(\lambda^{*})$, which satisfies the adjoint Lax pair, they obey the following differential equations:
\begin{align}
    \mathbf{F}_{x}(\lambda) &= [\mathbf{U}(\lambda), \mathbf{F}(\lambda)], \\
    \mathbf{G}_{x}(\lambda) &= -[\mathbf{G}(\lambda), \mathbf{U}(\lambda)],
\end{align}
where the second equation follows from taking the Hermitian transpose and replacing $\lambda$ with $\lambda^{*}$ in the first equation, using the symmetries \eqref{sym-U-V} and \eqref{symmetric-squared-eigenfunction-matrix}.

Differentiating the product $\mathbf{F}(\eta)\mathbf{G}(\lambda)$ yields
\begin{equation*}
 \left(\mathbf{F}(\eta)\mathbf{G}(\lambda)\right)_{x} = \mathbf{U}(\eta)\mathbf{F}(\eta)\mathbf{G}(\lambda)
 - \mathbf{F}(\eta)\mathbf{G}(\lambda)\mathbf{U}(\lambda) + \ii(\lambda - \eta)\mathbf{F}(\eta)\sigma_{3}\mathbf{G}(\lambda).
\end{equation*}
Taking the trace of both sides and applying this identity to the squared eigenfunction matrices gives, for $i,j \in \{\pm 1, \pm 2\}$,
\begin{equation}\label{squ-eig-connect-squ-eig-mat}
    2(\lambda - \eta)\, s_{j}(\mathbf{\Phi})(\eta^{*})^{\dagger} \mathcal{J} \, s_{i}(\mathbf{\Phi})(\lambda) = 
    \partial_{x} \, \mathrm{Tr}\left(p_{j}^{\dagger}(\mathbf{\Phi})(\eta^{*}) p_{i}(\mathbf{\Phi})(\lambda)\right),
\end{equation}
which provides a useful identity for computing the inner product between squared eigenfunctions and 
their adjoint counterparts via the asymptotic behavior of the squared eigenfunction matrices.

In particular, only the squared eigenfunction matrices corresponding to the Lax spectrum described in 
Lemma~\ref{lem-Lax-spectrum} needs to be considered. The squared eigenfunction matrices associated 
with the essential spectrum do not contribute in the negative direction for the second variation of 
the Lyapunov functional (see Theorem~\ref{thm-orthogonal}). 
Therefore, the analysis can be restricted to the squared eigenfunction matrices defined on the point spectrum.
According to Lemma~\ref{lem-Lax-spectrum}, certain squared eigenfunction matrices on the point spectrum 
are linearly independent, since the columns of $\mathbf{\Phi}$ are not linearly dependent.

Now we consider the squared eigenfunction matrices for the $N$-soliton solutions given by the 
fundamental matrix solution $\mathbf{\Phi}^{[N]}$.  
For the squared eigenfunctions evaluated at the point spectrum, by \eqref{Rel-Phi-lambdac}, one has
\begin{align*}
    p_{i}(\mathbf{\Phi}^{[N]})(\lambda_{k})=\mathbf{\Phi}^{[N]}_{1}(\lambda_{k})(\mathbf{\Phi}^{[N]}_{i+1}(\lambda_{k}^{*}))^{\dagger}=c_{ik}\mathbf{\Phi}^{[N]}_{1}(\lambda_{k})(\mathbf{\Phi}^{[N]}_{1}(\lambda_{k}^{*}))^{\dagger},\\
    p_{i}(\mathbf{\Phi}^{[N]})(\lambda_{k}^{*})=\mathbf{\Phi}^{[N]}_{1}(\lambda_{k}^{*})(\mathbf{\Phi}^{[N]}_{i+1}(\lambda_{k}))^{\dagger}=\frac{1}{c_{ik}^{*}}\mathbf{\Phi}^{[N]}_{i+1}(\lambda_{k}^{*})(\mathbf{\Phi}^{[N]}_{i+1}(\lambda_{k}))^{\dagger},
\end{align*}
and similarly,
\begin{align*}
    p_{-i}(\mathbf{\Phi}^{[N]})(\lambda_{k})=\mathbf{\Phi}^{[N]}_{i+1}(\lambda_{k})(\mathbf{\Phi}^{[N]}_{1}(\lambda_{k}^{*}))^{\dagger}=\frac{1}{c_{ik}}\mathbf{\Phi}^{[N]}_{i+1}(\lambda_{k})(\mathbf{\Phi}^{[N]}_{i+1}(\lambda_{k}^{*}))^{\dagger},\\
    p_{-i}(\mathbf{\Phi}^{[N]})(\lambda_{k}^{*})=\mathbf{\Phi}^{[N]}_{i+1}(\lambda_{k}^{*})(\mathbf{\Phi}^{[N]}_{1}(\lambda_{k}))^{\dagger}=c_{ik}^{*}\mathbf{\Phi}^{[N]}_{1}(\lambda_{k}^{*})(\mathbf{\Phi}^{[N]}_{1}(\lambda_{k}))^{\dagger},
\end{align*}
for $i=1,2$ and $k=1,2,\dots, N$. 
From these relations, it follows that for fixed $k$ and $i,j=1,2$,
\begin{equation*}
    p_{i}(\mathbf{\Phi}^{[N]})(\lambda_{k})=\frac{c_{ik}}{c_{jk}}p_{j}(\mathbf{\Phi}^{[N]})(\lambda_{k}). 
\end{equation*}
Moreover, by \eqref{Rel-Phi-lambda}, the following linear relation holds:
\begin{equation*}
    \frac{1}{c_{ik}}p_{i}(\mathbf{\Phi}^{[N]})(\lambda_{k})+\sum_{j=1}^{2}c_{jk}p_{-j}(\mathbf{\Phi}^{[N]})(\lambda_{k})=0.
\end{equation*}
Hence, the span of the squared eigenfunction matrices at $\lambda_{k}$ can be characterized as
\begin{align}\label{inde-lambdai}
    \begin{split}
        &\mathrm{span}\left\{p_{i}(\mathbf{\Phi}^{[N]})(\lambda_{k}),p_{-i}(\mathbf{\Phi}^{[N]})(\lambda_{k}):i=1,2\right\}
    \\=&\mathrm{span}
    \left\{\mathbf{\Phi}^{[N]}_{1}(\lambda_{k})(\mathbf{\Phi}^{[N]}_{1}(\lambda_{k}^{*}))^{\dagger},
    \mathbf{\Phi}^{[N]}_{2}(\lambda_{k})(\mathbf{\Phi}^{[N]}_{1}(\lambda_{k}^{*}))^{\dagger}\right\}
    \end{split}
\end{align}
if $c_{2k}\ne 0$. 
Similarly, at $\lambda_{k}^{*}$,
\begin{align}\label{inde-lambdai-con}
    \begin{split}
        &\mathrm{span}\left\{p_{i}(\mathbf{\Phi}^{[N]})(\lambda_{k}^{*}),p_{-i}(\mathbf{\Phi}^{[N]})(\lambda_{k}^{*}):i=1,2\right\}
    \\=&\mathrm{span}
    \left\{\mathbf{\Phi}^{[N]}_{1}(\lambda_{k}^{*})(\mathbf{\Phi}^{[N]}_{1}(\lambda_{k}))^{\dagger},
    \mathbf{\Phi}^{[N]}_{1}(\lambda_{k}^{*})(\mathbf{\Phi}^{[N]}_{2}(\lambda_{k}))^{\dagger}\right\}
    \end{split}
\end{align}
for $k=1,2,\cdots,N$ if $c_{2k}\ne 0$. 

Now we define the squared eigenfunction matrices and squared eigenfunctions for CNLS equations. 
By Lemma~\ref{lem-critical-point}, it suffices to consider the squared eigenfunction matrices at $t=0$. 
\begin{definition}
 Let $\mathbf{\Phi}^{[N]}(\lambda;x,t)$ be the FMS \eqref{FMS-Nsoliton} associated with the $N$-soliton 
 solutions \eqref{CNLS-Nsoliton} with spectral parameters $\lambda_{k},1\leq k \leq N$.
 We define the squared eigenfunction matrices and squared eigenfunctions as
    \begin{equation}\label{eigenfunction-t0}
        \mathbf{P}_{\pm i}(\lambda;x)=p_{\pm i}(\mathbf{\Phi}^{[N]}|_{t=0}),\quad 
        \mathbf{S}_{\pm i}(\lambda;x)=s_{\pm i}(\mathbf{\Phi}^{[N]}|_{t=0}). 
    \end{equation}
 To distinguish contributions from the essential and point spectra, we introduce the following notation. Let
    \begin{equation}\label{Gamma-set}
        \Gamma=\{k:c_{1k}=0\}, 
    \end{equation}
 we define the set of squared eigenfunctions associated with the essential spectrum by
\begin{align}
    \mathsf{E}_{ess}=&\{ \mathbf{S}_{\pm i}(\lambda;x):i=1,2, \ \lambda\in\sigma_{ess}(\mathcal{L}_{s}) \},\label{E-ess}
\end{align}
and those associated with the point spectrum by
\begin{align}
    \mathsf{E}_{point}=&\{\mathbf{S}_{1}(\lambda_{k};x),\mathbf{S}_{-2}(\lambda_{k};x),\mathbf{S}_{2}(\lambda_{k}^{*};x),
    \mathbf{S}_{-1}(\lambda_{k}^{*};x):k\notin\Gamma\}\cup \label{E-point}
    \\
    &\{\mathbf{S}_{2}(\lambda_{k};x),\mathbf{S}_{-1}(\lambda_{k};x),\mathbf{S}_{1}(\lambda_{k}^{*};x),
    \mathbf{S}_{-2}(\lambda_{k}^{*};x):k\in\Gamma\}, \\
    \hat{\mathsf{E}}_{point}=&\{\mathbf{S}_{1,\lambda}(\lambda_{k};x),\mathbf{S}_{-1,\lambda}(\lambda_{k}^{*};x):k\notin\Gamma\} \cup \{\mathbf{S}_{2,\lambda}(\lambda_{k};x),\mathbf{S}_{-2,\lambda}(\lambda_{k}^{*};x):k\in\Gamma\}. \label{hat-E-point}
\end{align}
\end{definition}
\begin{rem}
 For the coupled NLS equations, the definition of squared eigenfunctions and squared eigenfunction matrices requires more care than in the scalar NLS case.

 In the scalar NLS equation, one can apply an $N$-fold Darboux transformation with spectral parameters $\lambda_k$ and scattering parameters $c_k$ for $k=1,2,\dots,N$. If $c_N = 0$, the $N$-soliton solution degenerates to an $(N-1)$-soliton solution. However, for the coupled NLS equations, differences arise due to the fact that each spectral parameter $\lambda_k$ is associated with two scattering parameters $c_{1k}$ and $c_{2k}$.
    For example, a non-degenerate vector 2-soliton solution \cite{qin_nondegenerate_2019} can arise even when $c_{12} = c_{21} = 0$, provided $c_{11}, c_{22} \in \mathbb{C} \setminus \{0\}$.

 Regarding the squared eigenfunction matrices, if $k \in \Gamma$, then $\mathbf{S}_{1}(\lambda_k) = 0$, and 
 we must instead use $\mathbf{S}_{2}(\lambda_k)$ in the definitions \eqref{E-point} and \eqref{hat-E-point}. 
 Moreover, the matrix $\mathbf{\Phi}_{3}^{[N]}(\lambda_{k})$ is linearly dependent on $\mathbf{\Phi}_{1}^{[N]}(\lambda_{k})$, so it is preferable to use $\mathbf{S}_{-1}(\lambda_k)$ rather than $\mathbf{S}_{-2}(\lambda_k)$. A similar argument applies to $\lambda=\lambda_{k}^{*}$.
 An alternative approach to handle the degenerate case has been proposed in~\cite{ling2024stability}.
\end{rem}
In the next section, we will introduce the asymptotic behavior of \eqref{eigenfunction-t0} in order to evaluate the integrals between the squared eigenfunctions and the adjoint squared eigenfunctions via \eqref{squ-eig-connect-squ-eig-mat}. Furthermore, using \eqref{FMS-Nsoliton}, the squared eigenfunctions in \eqref{eigenfunction-t0} can be expressed in terms of the Darboux matrix as
\begin{equation}\label{squ-eig-t-0}
    \mathbf{P}_{\pm i}={\rm e}^{\pm 2{\rm i}\lambda x}p_{\pm i}(\mathbf{D}^{[N]}|_{t=0}), \quad
    \mathbf{S}_{\pm i}={\rm e}^{\pm 2{\rm i}\lambda x}s_{\pm i}(\mathbf{D}^{[N]}|_{t=0}). 
\end{equation}
Therefore, it suffices to analyze the asymptotic behavior of the Darboux matrix.

\section{Spectral analysis and nonlinear stability of CNLS solitons}\label{sec-stability-CNLS}
In this section, we aim to characterize the number of negative eigenvalues and describe the kernel of the second variation of the Lyapunov functional \eqref{Lya}, given by
\begin{equation}\label{op-L}
    \mathcal{L}=\sum_{n=0}^{2N}\mu_{n}\frac{\delta^{2}\mathcal{H}_{n}}{\delta \mathbf{q}^{2}}
\end{equation}
where $\mathcal{L}$ is a self-adjoint differential operator of order $2N$ and prove the nonlinear stability 
of solitons for the CNLS equations.
The essential spectrum of $\mathcal{L}$ can be determined directly via Weyl’s essential spectrum theorem, so it remains to analyze the point spectrum.
Although determining the full spectrum of $\mathcal{L}$ is difficult due to the complexity of the expression 
\eqref{op-L}, the number of negative eigenvalues and the structure of the kernel can nevertheless 
be characterized in this section.

To overcome this difficulty, we introduce an auxiliary operator $\mathcal{J} = -\ii\mathrm{diag}(\mathbb{I}_2,-\mathbb{I}_2)$ and analyze the spectrum of the operator $\mathcal{J}\mathcal{L}$ instead of $\mathcal{L}$ directly. This approach enables us to determine the number of negative eigenvalues and the dimension of the kernel of $\mathcal{L}$, which is sufficient to establish the nonlinear stability of the $N$-soliton solutions.

The kernel of the operator $\mathcal{L}$ coincides with that of $\mathcal{J}\mathcal{L}$ since the auxiliary 
operator $\mathcal{J}$ is invertible.
The main objective of this section is to determine the number of negative eigenvalues of $\mathcal{L}$, 
which is more involved than computing its kernel. We introduce the negative cone
\begin{equation*}
    \mathcal{N}=\{\mathbf{z}:(\mathcal{L}\mathbf{z},\mathbf{z})<0\},
\end{equation*}
and define $\mathrm{n}(\mathcal{L})$ as the number of negative eigenvalues of $\mathcal{L}$. According 
to \cite{kapitula_counting_2004,kapitula_spectral_2013}, $\mathrm{n}(\mathcal{L})$ equals 
the dimension $\dim(\mathcal{N})$ of the maximal subspace contained in $\mathcal{N}$.
Since the squared eigenfunctions form a complete basis of $L^{2}$ \cite{yang_nonlinear_2010}, 
the value of $\dim(\mathcal{N})$ can be computed using the quadratic form $(\mathcal{L} \cdot, \cdot)$ 
restricted to the set of squared eigenfunctions.

Moreover, since the squared eigenfunctions satisfy the spectral problem associated with $\mathcal{J}\mathcal{L}$, it is sufficient to evaluate the quadratic form $(\mathcal{J}^{-1}\cdot, \cdot) = -(\mathcal{J} \cdot, \cdot)$. Define
\begin{equation}\label{omega}
    \omega(\mathbf{f},\mathbf{g})=\int_{\mathbb{R}}\mathbf{f}^{\dagger}\mathcal{J}\mathbf{g}\mathrm{d}x, 
\end{equation}
so that
\begin{equation*}
    (\mathbf{f},\mathcal{J}\mathbf{g})=\mathrm{Re} \ \omega(\mathbf{f},\mathbf{g}). 
\end{equation*}
The quantity $\omega(\cdot, \cdot)$ evaluated on the span of the squared eigenfunctions can be computed using the asymptotic behavior of the squared eigenfunction matrices via \eqref{squ-eig-connect-squ-eig-mat}. This calculation is equivalent to evaluating the integral between the squared eigenfunctions and their adjoint counterparts, which will be carried out at the beginning of this section.

We consider the operator $\mathcal{L}$ in the real Hilbert space $\mathrm{X}$ defined by
\begin{equation*}
    \mathrm{X}=\{(u_{1},u_{2},u_{1}^{*},u_{2}^{*}):u_{1},u_{2}\in L^{2}(\mathbb{R},\mathbb{C})\}
\end{equation*}
which can be identified with the Hilbert space $L^{2}(\mathbb{R}, \mathbb{C}^{2})$ 
under the inner product
\begin{equation*}
    (\mathbf{f},\mathbf{g})=\mathrm{Re}\int_{\mathbb{R}}\mathbf{f}^{\dagger}\mathbf{g}\mathrm{d}x. 
\end{equation*}
Any operator $\mathcal{A}$ on $L^{2}(\mathbb{R}, \mathbb{C}^{2})$ can be extended to an operator $\mathcal{A}'$ on $\mathrm{X}$ via
\begin{equation*}
    \mathcal{A}'\begin{pmatrix}
        \mathbf{u} \\
        \mathbf{u}^{*}
    \end{pmatrix}= \begin{pmatrix}
        \mathcal{A}\mathbf{u} \\
        (\mathcal{A}\mathbf{u})^{*}
    \end{pmatrix}. 
\end{equation*}
We decompose the operator $\mathcal{L}$ acting on a function $\mathbf{u}$ into 
its $\mathbf{u}$ and $\mathbf{u}^{*}$ components as
\begin{equation}
    \mathcal{L}\mathbf{u}= \mathcal{L}_{1}\mathbf{u}+\mathcal{L}_{2}\mathbf{u}^{*},
\end{equation}
so that the operator $\mathcal{L}$ on $\mathrm{X}$ admits the matrix representation
\begin{equation}\label{matrix-L}
    \mathcal{L}=\begin{pmatrix}
        \mathcal{L}_{1} & \mathcal{L}_{2} \\
        \mathcal{L}_{2}^{*} & \mathcal{L}_{1}^{*}
    \end{pmatrix},
\end{equation}
where we have used the same notation $\mathcal{L}$ by abuse of notation. 

\subsection{The integral between squared eigenfunctions and adjoint squared eigenfunctions}
In this subsection, the integral between the squared eigenfunctions \eqref{squ-eig-t-0} and 
their adjoint eigenfunctions (obtained by multiplying the left side by $\mathcal{J}$) on the Lax spectrum 
is examined.
Define
\begin{equation}\label{JL-eig}
    \hat{\mathsf{E}}=\mathsf{E}_{ess}\cup\mathsf{E}_{point}\cup \hat{\mathsf{E}}_{point},\quad 
    \mathsf{E}=\mathsf{E}_{ess}\cup\mathsf{E}_{point}.
\end{equation}
The main result of this subsection is the evaluation of the integral
\begin{equation}\label{integral-fg}
    \int_{\mathbb{R}}\mathbf{f}^{\dagger}(\lambda;x) \mathbf{g}(\lambda',x)\mathrm{d}x
\end{equation}
with $\mathbf{f}\in \hat{\mathsf{E}}$ and $\mathbf{g}\in\mathcal{J}\mathsf{E}$, and 
$\lambda,\lambda'\in \sigma(\mathcal{L}_{s})$ as in Lemma~\ref{lem-Lax-spectrum}. 
Note that the space $\mathcal{J}\mathsf{E}=\{\mathcal{J}\mathbf{f}:\mathbf{f}\in \mathsf{E}\}$ is used, 
and the case $\mathbf{f}\in \hat{\mathsf{E}}$ and $\mathbf{g}\in\mathcal{J}\hat{\mathsf{E}}$ in 
\eqref{integral-fg} can also be obtained. However, it suffices to consider the case 
$\mathbf{g}\in\mathcal{J}\mathsf{E}$. 

The following theorem holds:
\begin{thm}\label{thm-orthogonal}
    The Hermitian inner product between the squared eigenfunctions and the adjoint squared eigenfunctions 
    can be expressed as the derivative of the trace of the product of the corresponding squared 
    eigenfunction matrices:
    \begin{equation}\label{squ-eig-connect-squ-eig-mat-2}
        \mathbf{S}_{i}^{\dagger}(\eta^{*};x)\mathcal{J}\mathbf{S}_{j}(\lambda;x)=\frac{1}{2(\lambda-\eta)}
        \mathrm{Tr}\left( \mathbf{P}_{i}^{\dagger}(\eta^{*};x)\mathbf{P}_{j}(\lambda;x) \right)_{x}
    \end{equation}
    for $\lambda,\eta \in \mathbb{C}$ and $i,j = \pm 1, \pm 2$. As a consequence, the orthogonality 
    conditions on the set $\mathsf{E}$ in \eqref{JL-eig} follow. 

    For spectral parameters in the essential spectrum of the Lax operator, the squared eigenfunctions 
    and their adjoint counterparts belong to $L^{\infty}$ and satisfy
    \begin{equation}\label{ort-ess}
        \int_{\mathbb{R}}\mathbf{S}_{i}^{\dagger}(\lambda;x)\mathcal{J}\mathbf{S}_{j}(\lambda';x)\mathrm{d}x=-
        \int_{\mathbb{R}}\mathbf{S}_{-j}^{\dagger}(\lambda';x)\mathcal{J}\mathbf{S}_{-i}(\lambda;x)\mathrm{d}x=
        \ii \pi |\mathcal{P}(\lambda)|^{4}\delta(\lambda-\lambda')\delta_{ij}
    \end{equation}
    and 
    \begin{equation*}
        \int_{\mathbb{R}}\mathbf{S}_{i}^{\dagger}(\lambda;x)\mathcal{J}\mathbf{S}_{-j}(\lambda';x)\mathrm{d}x=0
    \end{equation*}
    for $\lambda,\lambda'\in \sigma_{ess}(\mathcal{L}_{s})$ and $i,j\in\{ 1, 2\}$.

    The integral between the squared eigenfunctions on the essential spectrum and the adjoint 
    squared eigenfunctions on the point spectrum vanish:
    \begin{equation}\label{ort-dis-1}
        \int_{\mathbb{R}}\mathbf{S}_{i}^{\dagger}(\lambda';x)\mathcal{J}\mathbf{S}_{j}(\lambda;x)\mathrm{d}x=0
    \end{equation}
    for $\lambda'\in\sigma_{ess}(\mathcal{L}_{s})$ and $\lambda\in\sigma_{point}(\mathcal{L}_{s})$. 

    For the squared eigenfunctions on the point spectrum of the Lax operator, the only nontrivial terms are
    \begin{equation}\label{ort-dis-3}
        \int_{\mathbb{R}}\mathbf{S}_{-i,\lambda}^{\dagger}(\lambda_{k}^{*})\mathcal{J}\mathbf{S}_{j}(\lambda_{k})
        \mathrm{d}x=
        \left(\int_{\mathbb{R}}\mathbf{S}_{j,\lambda}^{\dagger}(\lambda_{k})\mathcal{J}\mathbf{S}_{-i}(\lambda_{k}^{*})
        \mathrm{d}x\right)^{*}=-\frac{1}{2}c_{ik}c_{jk}\hat{\mathcal{P}}_{\lambda}(\lambda_{k})^{2}
        \mathcal{P}(\lambda_{k})^{2}
    \end{equation}
    for $k=1,2,\cdots,N$ and $i,j=1,2$. 
    All other terms vanish, i.e.,
    \begin{equation}\label{ort-dis-4}
        \int_{\mathbb{R}}\mathbf{f}^{\dagger}(\lambda;x) \mathcal{J}\mathbf{g}(\lambda';x)\mathrm{d}x=0
    \end{equation}
    for 
    \begin{align*}
        \mathbf{f}\in&\{\mathbf{S}_{i,\lambda}(\lambda_{k}),\mathbf{S}_{-i,\lambda}(\lambda_{k}^{*}):
        i=1,2,\ k=1,2,\cdots,N\}\\
        \mathbf{g}\in&\{\mathbf{S}_{\pm i}(\lambda):i=1,2, \ \lambda \in \sigma_{point}(\mathcal{L}_{s})\}
    \end{align*}
    except for the pairs $(\mathbf{f},\mathbf{g})\ne (\mathbf{S}_{-i,\lambda}(\lambda_{k}^{*}),\mathbf{S}_{j}(\lambda_{k})),
    (\mathbf{S}_{j,\lambda}(\lambda_{k}),\mathbf{S}_{-i}(\lambda_{k}^{*}))$ with $k=1,2,\cdots, N$ and 
    $i,j=1,2$. 
\end{thm}

The proof begins with formula~\eqref{squ-eig-connect-squ-eig-mat-2}, which 
is derived from~\eqref{squ-eig-connect-squ-eig-mat}.
To establish \eqref{ort-ess}–\eqref{ort-dis-4}, both sides of~\eqref{squ-eig-connect-squ-eig-mat-2} 
are integrated over $\mathbb{R}$. 
Thus, it suffices to analyze the asymptotic behavior of the squared eigenfunction matrices 
as $x \to \pm\infty$ for spectral parameters in the Lax spectrum to determine the integrals on 
the left-hand side of~\eqref{squ-eig-connect-squ-eig-mat-2}.
Because the squared eigenfunctions and their associated matrices have complicated expressions, 
a sequence of matrix functions is introduced to simplify them in the limit $x \to \pm\infty$ 
before proving Theorem~\ref{thm-orthogonal}.

Since the squared eigenfunction matrices are constructed via the FMS, it is sufficient 
to consider the asymptotic behavior of the Darboux matrix~\eqref{DT-Nsoliton} at $t=0$.
For convenience, and with a slight abuse of notation, the Darboux matrix at $t=0$ is denoted by
\begin{equation}\label{DT-Nsoliton-t0}
\mathbf{D}^{[N]}(\lambda;x) = \mathbf{D}^{[N]}(\lambda;x,0).
\end{equation}
For the matrix $\mathbf{M}$ in~\eqref{matrix-M} and its inverse $\mathbf{m}$ at $t=0$, define
\begin{equation*}
    \begin{split}
        F_{+}&=\det\left(
        \frac{\lambda_{i}-\lambda_{i}^{*}}{\lambda_{j}-\lambda_{i}^{*}}\mathbf{c}_{i}^{\dagger}\mathbf{c}_{j}
    \right)_{1\leq i,j \leq N},\quad  F_{+}^{rs}=\det\left(
        \frac{\lambda_{i}-\lambda_{i}^{*}}{\lambda_{j}-\lambda_{i}^{*}}\mathbf{c}_{i}^{\dagger}\mathbf{c}_{j}
    \right)_{1\leq i,j \leq N,i\ne r,j\ne s}, \\
    F_{-}&=\det\left(
        \frac{\lambda_{i}-\lambda_{i}^{*}}{\lambda_{j}-\lambda_{i}^{*}}
    \right)_{1\leq i,j \leq N},\quad  F_{-}^{rs}=\det\left(
        \frac{\lambda_{i}-\lambda_{i}^{*}}{\lambda_{j}-\lambda_{i}^{*}}
    \right)_{1\leq i,j \leq N,i\ne r,j\ne s}
    \end{split}
\end{equation*}
and 
\begin{equation*}
    \begin{split}
        &\mathbf{M}_{+}= \left(
            \frac{\lambda_{s}-\lambda_{s}^{*}}{\lambda_{r}-\lambda_{s}^{*}}
                \mathbf{c}_{s}^{\dagger}\mathbf{c}_{r}{\rm e}^{-{\rm i}x(a_{r}-a_{s})+x(b_{s}+b_{r})}
        \right)_{1\leq s,r\leq N}, \\&
        \mathbf{M}_{-}= \left(
            \frac{\lambda_{s}-\lambda_{s}^{*}}{\lambda_{r}-\lambda_{s}^{*}}
                {\rm e}^{{\rm i}x(a_{r}-a_{s})-x(b_{s}+b_{r})}
        \right)_{1\leq s,r\leq N},\\
        &\mathbf{m}_{\pm}=\frac{1}{F_{\pm}}\left(
            (-1)^{s+r}F_{\pm}^{rs}{\rm e}^{\pm{\rm i}x(a_{s}-a_{r})}{\rm e}^{\mp x(b_{s}+b_{r})}
        \right)_{1\leq s,r \leq N}.
    \end{split}
\end{equation*}
The matrix $\mathbf{M}$ and 
its inverse at $t=0$ are given by
\begin{equation*}
    \mathbf{M}|_{t=0} = \mathbf{M}_{\pm} + \left(\mathcal{O}({\rm e}^{\mp (b_{i}+b_{j}) x})\right)_{1\leq i,j\leq N},\quad 
    \mathbf{m}|_{t=0} =\mathbf{m}_{\pm}+\left(o({\rm e}^{\mp (b_{i}+b_{j})x})\right)_{1\leq i,j\leq N}
\end{equation*}
as $x\to\pm\infty$. 
We note that the determinants $F_{-}$ and $F_{-}^{rs}$ are Cauchy determinants and can be expressed as
\begin{align*}
    F_{-}&=\prod_{i=1}^{N}(2{\rm i}b_{i})
    \frac{\prod_{i=2}^{N}\prod_{j=1}^{i-1}(\lambda_{i}-\lambda_{j})(\lambda_{j}^{*}-\lambda_{i}^{*})}
    {\prod_{i,j=1}^{N}(\lambda_{i}-\lambda_{j}^{*})},\\
    F_{-}^{rs}&= \prod_{i\ne r}^{N}(2{\rm i}b_{i})
    \frac{\prod_{i=2,i\ne s}^{N}\prod_{j=1,j\ne s}^{i-1}(\lambda_{i}-\lambda_{j})
    \prod_{i=2,i\ne r}^{N}\prod_{j=1,j\ne r}^{i-1}(\lambda_{j}^{*}-\lambda_{i}^{*})}
    {\prod_{i,j=1,i\ne s,j\ne r}^{N}(\lambda_{i}-\lambda_{j}^{*})}.
\end{align*}
The asymptotic Darboux matrix is defined by
\begin{equation*}
    \mathbf{D}_{\pm}^{[N]}(\lambda;x)=\mathcal{P}(\lambda)-\sum_{s,r=1}^{N}\frac{\mathcal{P}(\lambda)}{\lambda-\lambda_{r}^{*}}
    (\lambda_{r}-\lambda_{r}^{*})m_{\pm,sr}{\rm e}^{{\rm i}\lambda_{s}x\sigma_{3}}\begin{pmatrix}
        1 \\ \mathbf{c}_{s}
    \end{pmatrix}
    \begin{pmatrix}
        1 & \mathbf{c}_{r}^{\dagger}
    \end{pmatrix}
    {\rm e}^{-{\rm i}\lambda_{r}^{*}x\sigma_{3}}. 
\end{equation*}
The Darboux matrix \eqref{DT-Nsoliton} has the asymptotic expression
\begin{equation*}
    \mathbf{D}^{[N]}(\lambda;x)= \mathbf{D}_{\pm}^{[N]}(\lambda;x)+
    \sum_{s,r=1}^{N}\frac{\mathcal{P}(\lambda)}{\lambda-\lambda_{r}^{*}}
    {\rm e}^{-b_{s}x\sigma_{3}}o({\rm e}^{\mp x(b_{s}+b_{r})})
    {\rm e}^{-b_{r}x\sigma_{3}},\quad x\to \pm\infty. 
\end{equation*}
To simplify the notation, we denote 
\begin{equation*}
    G_{\pm}^{rs}(\lambda)=\frac{\mathcal{P}(\lambda)}{\lambda-\lambda_{r}^{*}}
    \frac{2{\rm i}b_{r}}{F_{\pm}}(-1)^{r+s}F_{\pm}^{rs}, 
\end{equation*}
then the asymptotic Darboux matrix has the representation
\begin{equation}\label{asy-DT-matrix}
    \mathbf{D}_{\pm}^{[N]}(\lambda;x)=(\mathbf{D}_{1,\pm}^{[N]}(\lambda;x),\mathbf{D}_{2,\pm}^{[N]}(\lambda;x),\mathbf{D}_{3,\pm}^{[N]}(\lambda;x))
\end{equation}
with
\begin{align*}
    \mathbf{D}^{[N]}_{1,+}(\lambda;x)&=\mathcal{P}(\lambda)\mathbf{e}_{1}-\sum_{s,r=1}^{N}G_{+}^{rs}(\lambda)\begin{pmatrix}
        {\rm e}^{2{\rm i}(a_{s}-a_{r})x-2(b_{s}+b_{r})x} \\
        \mathbf{c}_{s}{\rm e}^{-2{\rm i}a_{r}x-2b_{r}x}
    \end{pmatrix},\\
    \mathbf{D}^{[N]}_{1,-}(\lambda;x)&=\mathcal{P}(\lambda)\mathbf{e}_{1}-\sum_{s,r=1}^{N}G_{-}^{rs}(\lambda)\begin{pmatrix}
        1 \\
        \mathbf{c}_{s}{\rm e}^{-2{\rm i}a_{s}x+2b_{s}x}
    \end{pmatrix},\\
    \mathbf{D}^{[N]}_{i+1,+}(\lambda;x)&=\mathcal{P}(\lambda)\mathbf{e}_{i+1}-\sum_{s,r=1}^{N}c_{ir}^{*}G_{+}^{rs}(\lambda)\begin{pmatrix}
        {\rm e}^{2{\rm i}a_{s}x-2b_{s}x} \\ \mathbf{c}_{s}
    \end{pmatrix},\\
    \mathbf{D}^{[N]}_{i+1,-}(\lambda;x)&=\mathcal{P}(\lambda)\mathbf{e}_{i+1}-\sum_{s,r=1}^{N}c_{ir}^{*}G_{-}^{rs}(\lambda)\begin{pmatrix}
        {\rm e}^{2{\rm i}a_{r}x+2b_{r}x} \\
        \mathbf{c}_{s}{\rm e}^{-2{\rm i}(a_{s}-a_{r})x+2(b_{s}+b_{r})x}
    \end{pmatrix}
\end{align*}
where $\mathbf{e}_{i}$ denotes the $i$-th standard basis column vector (with the $i$-th component equal to $1$). 
The asymptotic squared eigenfunction matrices can now be considered.
Denote 
\begin{equation}\label{asy-P-hat-1}
    \hat{\mathbf{P}}_{i}^{\pm}(\lambda;x)=\mathbf{D}^{[N]}_{1,\pm}(\lambda;x)
    \mathbf{D}^{[N]}_{i+1,\pm}(\lambda^{*};x)^{\dagger}
\end{equation}
for $i=1,2$ and
\begin{equation}\label{asy-P-hat-2}
    \hat{\mathbf{P}}_{-i}^{\pm}(\lambda;x)=\left(\hat{\mathbf{P}}_{i}^{\pm}(\lambda^{*};x)\right)^{\dagger}.
\end{equation}
Note that we have
\begin{equation*}
    \mathcal{K}^{s}(\lambda)=\sum_{r=1}^{N}G_{-}^{rs}(\lambda)=
    \frac{\prod_{i=1}^{N}(\lambda_{s}-\lambda_{i}^{*})}{\prod_{i\ne s}^{N}(\lambda_{s}-\lambda_{i})}
    \prod_{i\ne s}^{N}(\lambda-\lambda_{i})
\end{equation*}
since  $\mathcal{K}^{s}(\lambda)$ is a polynomial of degree $N-1$ and 
\begin{equation}\label{root-Ks}
    \mathcal{K}^{s}(\lambda_{j})=\mathcal{P}(\lambda_{j})\delta_{sj}
\end{equation}
in view of 
\begin{equation*}
    \sum_{r=1}^{N}(-1)^{r+s}\frac{F_{-}^{rs}}{F_{-}}
    \frac{\lambda_{r}-\lambda_{r}^{*}}{\lambda_{j}-\lambda_{r}^{*}}
    =\delta_{sj}, 
\end{equation*}
which is given by $\mathbf{m}_{\pm}\mathbf{M}_{\pm}=\mathbb{I}_{N}$. 
Then the squared eigenfunction matrices admit the following asymptotic expansion
\begin{equation}\label{asy-square-eig-1}
    \mathbf{P}_{i}(\lambda;x)\sim{\rm e}^{2{\rm i}\lambda x}\hat{\mathbf{P}}_{i}^{\pm}(\lambda;x),\quad 
    \mathbf{P}_{-i}(\lambda;x)\sim{\rm e}^{-2{\rm i}\lambda x}\hat{\mathbf{P}}_{-i}^{\pm}(\lambda;x)
\end{equation}
for $i=1,2$ as $x\to\pm\infty$. 

Before proving Theorem~\ref{thm-orthogonal}, we first establish a lemma needed for the 
proof of \eqref{ort-ess}.
\begin{lem}\label{essential-part-1}
    For $i,j \in \{1,2\}$ and $\lambda,\lambda' \in \mathbb{R}$, the following identities hold:
    \begin{equation*}
        \mathrm{Tr}\left(
            (\hat{\mathbf{P}}_{i}^{+}(\lambda;x))^{\dagger}\hat{\mathbf{P}}_{j}^{+}(\lambda';x)
        \right)=\mathcal{P}(\lambda)\hat{\mathcal{P}}(\lambda)\mathcal{P}(\lambda')\hat{\mathcal{P}}(\lambda')(\delta_{ij}-
        \mathcal{G}_{ij}(\lambda,\lambda'))
    \end{equation*}
    where $\mathcal{G}_{ij}(\lambda,\lambda')$ is given by the ratio of two polynomials in 
    $\lambda$ and $\lambda'$ and satisfies 
    $\mathcal{G}_{ij}(\lambda,\lambda) = 0$.
\end{lem}

\begin{proof}
    It suffices to show that $\mathcal{G}_{ij}(\lambda,\lambda)=0$. 
    From~\eqref{asy-DT-matrix}, the functions
    \begin{equation*}
        \begin{split}
            &\mathcal{G}_{ij}(\lambda,\lambda)\\=&
        \sum_{r=1}^{N}\left(\frac{2\ii b_{r}}{\lambda-\lambda_{r}^{*}}G_{ij}^{[r]}+
        \frac{-2\ii b_{r}}{\lambda-\lambda_{r}}(G_{ji}^{[r]})^{*}\right)-
        \sum_{r=1}^{N}\frac{-2\ii b_{r}}{\lambda-\lambda_{r}}
        \begin{pmatrix}
            (G_{j1}^{[r]})^{*} & (G_{j2}^{[r]})^{*}
        \end{pmatrix}
        \sum_{r=1}^{N}\frac{2\ii b_{r}}{\lambda-\lambda_{r}^{*}}
        \begin{pmatrix}
            G_{i1}^{[r]} \\
            G_{i2}^{[r]}
        \end{pmatrix}
        \end{split}
    \end{equation*}
    where 
    \begin{equation*}
        G_{ij}^{[r]}=\sum_{s=1}^{N}c_{ir}^{*} \frac{1}{F_{+}}(-1)^{r+s}F_{+}^{rs}c_{js}. 
    \end{equation*}
    Since $\mathcal{G}_{ij}(\lambda,\lambda)$ is a meromorphic function with simple poles at 
    $\lambda_{1},\lambda_{2},\cdots,\lambda_{N}$
    and $\lambda_{1}^{*},\lambda_{2}^{*},\cdots,\lambda_{N}^{*}$, the condition $\mathcal{G}_{ij}(\lambda,\lambda)=0$
    is equivalent to requiring that all residues vanish, i.e.,
    \begin{equation*}
        G_{ij}^{[k]}=-\sum_{r=1}^{N}\frac{2\ii b_{r}}{\lambda_{k}^{*}-\lambda_{r}}
        \begin{pmatrix}
            (G_{j1}^{[r]})^{*} & (G_{j2}^{[r]})^{*}
        \end{pmatrix}
        \begin{pmatrix}
            G_{i1}^{[k]} \\
            G_{i2}^{[k]}
        \end{pmatrix}
    \end{equation*}
    for $k=1,2,\cdots N$. 
    Denote 
    \begin{equation*}
        \mathbf{G}_{k}=\begin{pmatrix}
            G_{11}^{[k]} & G_{12}^{[k]} \\
            G_{21}^{[k]} & G_{22}^{[k]}
        \end{pmatrix},
    \end{equation*}
    the condition $\mathcal{G}_{ij}(\lambda,\lambda)=0$
    is equivalent to
    \begin{equation}\label{DT+}
        \left(\mathbb{I}-\sum_{r=1}^{N}\frac{2\ii b_{r}}{\lambda_{k}-\lambda_{r}^{*}}\mathbf{G}_{r}\right)\mathbf{G}_{k}^{\dagger}
        =0. 
    \end{equation}
    By the definition of $G_{ij}^{[k]}$, the matrix $\mathbf{G}_{k}$ can be written as
    \begin{equation*}
        \mathbf{G}_{k}=\mathbf{c}_{k}^{*} \mathbf{d}_{k}^{T}, 
    \end{equation*}
    where $\mathbf{d}_{k}$ is determined from
    \begin{equation*}
        (\mathbf{c}_{1}, \mathbf{c}_{2}, \cdots \mathbf{c}_{N})=
        (\mathbf{d}_{1}, \mathbf{d}_{2}, \cdots \mathbf{d}_{N})\hat{\mathbf{M}}_{+}
    \end{equation*}
    with
    \begin{equation*}
        \hat{\mathbf{M}}_{+}=\left(
            \frac{\lambda_{k}-\lambda_{k}^{*}}{\lambda_{l}-\lambda_{k}^{*}}\mathbf{c}_{k}^{\dagger}\mathbf{c}_{l}
        \right)_{1\leq k,l \leq N}. 
    \end{equation*}
    Then \eqref{DT+} follows, since it is equivalent to
    \begin{equation*}
        \mathbf{d}_{k}^{\dagger}\left(\mathbb{I}-\sum_{r=1}^{N}\frac{2\ii b_{r}}{\lambda_{k}-\lambda_{r}^{*}}\mathbf{d}_{r} \mathbf{c}_{r}^{\dagger}\right)
        =0.
    \end{equation*} 
    This completes the proof.
\end{proof}

We now proceed to the proof of Theorem \ref{thm-orthogonal}.
\begin{proof}[Proof of Theorem \ref{thm-orthogonal}]
    We now prove \eqref{ort-ess}, the orthogonality condition for the squared eigenfunctions in the essential 
    spectrum. Let $\lambda, \lambda' \in \mathbb{R}$ in \eqref{squ-eig-connect-squ-eig-mat-2}, we obtain
    \begin{equation*}
        \int_{-x}^{x}\mathbf{S}_{i}^{\dagger}(\lambda;y)\mathcal{J}\mathbf{S}_{j}(\lambda';y)\mathrm{d}y=\frac{1}{2(\lambda'-\lambda)}
            \mathrm{Tr}\left( \mathbf{P}_{i}^{\dagger}(\lambda;x)\mathbf{P}_{j}(\lambda';x)-
            \mathbf{P}_{i}^{\dagger}(\lambda;-x)\mathbf{P}_{j}(\lambda';-x) \right). 
    \end{equation*}
    Since 
    \begin{equation*}
        \mathrm{Tr}\left( \mathbf{P}_{i}^{\dagger}(\lambda;-x)\mathbf{P}_{j}(\lambda';-x) \right)
        \sim \mathcal{P}(\lambda)^{2}\hat{\mathcal{P}}(\lambda')^{2}{\rm e}^{-2{\rm i}(\lambda'-\lambda) x}\delta_{ij}
    \end{equation*}
    and 
    \begin{equation*}
        \mathrm{Tr}\left( \mathbf{P}_{i}^{\dagger}(\lambda;x)\mathbf{P}_{j}(\lambda';x) \right)\sim 
        \mathcal{P}(\lambda)\hat{\mathcal{P}}(\lambda)\mathcal{P}(\lambda')\hat{\mathcal{P}}(\lambda'){\rm e}^{2{\rm i}(\lambda'-\lambda) x}\delta_{ij}+
        2(\lambda'-\lambda)\mathcal{P}_{ij}(\lambda,\lambda'){\rm e}^{2{\rm i}(\lambda'-\lambda) x}
    \end{equation*}
    for some polynomial $\mathcal{P}_{ij}(\lambda,\lambda')$ in $\lambda$ and $\lambda'$ 
    by Lemma \ref{essential-part-1}, 
    we define, for a Schwartz function $\mathbf{f}(\lambda)$,
    \begin{equation*}
        A=\lim_{x\to \infty}\left(
            \int_{-x}^{x}\mathbf{S}_{i}^{\dagger}(\lambda;y)\mathcal{J}\mathbf{S}_{j}(\lambda';y)\mathrm{d}y,
            \mathbf{f}(\lambda)
        \right). 
    \end{equation*} 
    Then we obtain
    \begin{align*}
        A=&\lim_{x\to \infty}\left(
            |\mathcal{P}(\lambda)\hat{\mathcal{P}}(\lambda')|^{2}{\rm e}^{-\ii\theta_{1}}\frac{
                {\rm e}^{2{\rm i}(\lambda'-\lambda) x+\ii \theta_{1}}-{\rm e}^{-2{\rm i}(\lambda'-\lambda) x-\ii \theta_{1}}
            }{2(\lambda'-\lambda)},\mathbf{f}(\lambda)
            \right)\delta_{ij}
            \\&+\lim_{x\to \infty}\mathcal{F}_{\lambda}(\mathcal{P}_{1}(\lambda,\lambda'){\rm e}^{2\ii\lambda'}{\mathbf{f}}(\lambda))(2x)\\
            =&\lim_{x\to \infty}\left(\ii|\mathcal{P}(\lambda)\hat{\mathcal{P}}(\lambda')|^{2}{\rm e}^{-\ii\theta_{1}}
                \frac{\mathrm{sin}(2(\lambda'-\lambda)x+\theta_{1})}{\lambda'-\lambda},\mathbf{f}(\lambda)
            \right)\delta_{ij}\\
            =&\left(
                \ii\pi|\mathcal{P}(\lambda)\hat{\mathcal{P}}(\lambda')|^{2}{\rm e}^{-\ii\theta_{1}}\delta(\lambda'-\lambda)\cos(\theta_{1}),\mathbf{f}(\lambda)
            \right)\delta_{ij}
    \end{align*}
    where 
    \begin{equation*}
        \theta_{1}=\mathrm{arg}(\hat{\mathcal{P}}(\lambda)\mathcal{P}(\lambda')). 
    \end{equation*}
    Hence 
    \begin{equation*}
        \int_{\mathbb{R}}\mathbf{S}_{i}^{\dagger}(\lambda;y)\mathcal{J}\mathbf{S}_{j}(\lambda';y)\mathrm{d}y=\ii \pi 
        |\mathcal{P}(\lambda)|^{4}\delta(\lambda'-\lambda)\delta_{ij}. 
    \end{equation*}
    Other terms in \eqref{ort-ess} can be obtained in a similar way.

    Next, we consider \eqref{ort-dis-3} and \eqref{ort-dis-4}, taking $i=j=1$ without loss of generality. 
Since the squared eigenfunctions in $\mathsf{E}$ and their derivatives with respect 
to $\lambda$ in $\hat{\mathsf{E}}$ are of Schwartz class, the integral in \eqref{squ-eig-connect-squ-eig-mat-2} 
over $\mathbb{R}$ exists.
Differentiating \eqref{squ-eig-connect-squ-eig-mat-2} with respect to $\eta$ and integrating both sides, we obtain:
\begin{equation}\label{cal-1}
    \begin{split}
        \omega(\mathbf{S}_{-1,\eta}(\eta),\mathbf{S}_{-1}(\lambda))=\frac{1}{2(\lambda - \eta^{*})} 
    \mathrm{Tr}\left(\mathbf{P}_{1,\eta}(\eta^{*}) \mathbf{P}_{-1}(\lambda) +
     \frac{1}{\lambda - \eta^{*}} \mathbf{P}_{1}(\eta^{*}) \mathbf{P}_{-1}(\lambda)\right) \bigg|_{-\infty}^{+\infty}
    \end{split}
\end{equation}
for $\lambda,\eta\in\{\lambda_{k}^{*}:k=1,2,\cdots,N\}$,
\begin{equation}\label{cal-2}
    \omega(\mathbf{S}_{1,\eta}(\eta),\mathbf{S}_{1}(\lambda))=\frac{1}{2(\lambda - \eta^{*})} 
     \mathrm{Tr}\left(\mathbf{P}_{-1,\eta}(\eta^{*}) \mathbf{P}_{1}(\lambda) +
      \frac{1}{\lambda - \eta^{*}} \mathbf{P}_{-1}(\eta^{*}) \mathbf{P}_{1}(\lambda)\right) \bigg|_{-\infty}^{+\infty}
\end{equation}
for $\lambda,\eta\in\{\lambda_{k}:k=1,2,\cdots,N\}$ and 
\begin{equation}\label{cal-3}
    \begin{split}
        \omega(\mathbf{S}_{-1,\eta}(\eta),\mathbf{S}_{1}(\lambda))=\frac{1}{2(\lambda - \eta^{*})} 
    \mathrm{Tr}\left(\mathbf{P}_{1,\eta}(\eta^{*}) \mathbf{P}_{1}(\lambda) +
     \frac{1}{\lambda - \eta^{*}} \mathbf{P}_{1}(\eta^{*}) \mathbf{P}_{1}(\lambda)\right) \bigg|_{-\infty}^{+\infty}
    \end{split}
\end{equation}
for $\lambda\ne\eta^{*},\lambda\in\{\lambda_{k}:k=1,2,\cdots,N\},\eta\in\{\lambda_{k}^{*}:k=1,2,\cdots,N\}$. 
By the symmetry property \eqref{symmetric-squared-eigenfunction-matrix},
    formula \eqref{squ-eig-connect-squ-eig-mat-2} can be rewritten as
\begin{equation}\label{discrete-2}
    \mathbf{S}_{i}^{\dagger}(\eta^{*};x)\mathcal{J}\mathbf{S}_{j}(\lambda;x)=\frac{1}{2(\lambda-\eta)}
        \partial_{x}\mathrm{Tr}\left( \mathbf{P}_{-i}(\eta;x)\mathbf{P}_{j}(\lambda;x) \right)
\end{equation}
for $i,j\in\{\pm 1, \pm 2\}$. For the value of the left-hand side of \eqref{discrete-2} when $\lambda = \eta$, we 
rewrite \eqref{discrete-2} as
\begin{equation*}
    2(\lambda-\eta)\mathbf{S}_{i}^{\dagger}(\eta^{*};x)\mathcal{J}\mathbf{S}_{j}(\lambda;x)=
        \mathrm{Tr}\left( \mathbf{P}_{-i}(\eta;x)\mathbf{P}_{j}(\lambda;x) \right)_{x},
\end{equation*}
then 
\begin{equation*}
    -4\mathbf{S}_{i,\eta}^{\dagger}(\eta^{*};x)\mathcal{J}\mathbf{S}_{j}(\lambda;x)
    +2(\lambda-\eta)\mathbf{S}_{i,\eta\eta}^{\dagger}(\eta^{*};x)\mathcal{J}\mathbf{S}_{j}(\lambda;x)=
    \mathrm{Tr}\left( \mathbf{P}_{-i,\eta\eta}(\eta;x)\mathbf{P}_{j}(\lambda;x) \right)_{x}. 
\end{equation*}
Taking $\lambda = \eta = \lambda_{k}$ and $i = -1,\ j = 1$, we obtain
\begin{equation}\label{cal-nontrivial}
    \begin{split}
        \omega(\mathbf{S}_{-1,\eta}(\lambda_{k}^{*}),\mathbf{S}_{1}(\lambda_{k}))=-\frac{1}{4} 
    \mathrm{Tr}\left(\mathbf{P}_{1,\eta\eta}(\lambda_{k}) \mathbf{P}_{1}(\lambda_{k})\right) \bigg|_{-\infty}^{+\infty}. 
    \end{split}
\end{equation}
To obtain \eqref{ort-dis-3} and \eqref{ort-dis-4} for the case $i = j = 1$, it remains to 
verify that the right-hand sides of \eqref{cal-1}–\eqref{cal-3} are zero, and that the 
right-hand side of \eqref{cal-nontrivial} is
\begin{equation*}
    -\frac{1}{4} 
    \mathrm{Tr}\left(\mathbf{P}_{1,\eta\eta}(\lambda_{k}) \mathbf{P}_{1}(\lambda_{k})\right) \bigg|_{-\infty}^{+\infty}=-\frac{1}{2}c_{1k}^{2}\hat{\mathcal{P}}_{\lambda}(\lambda_{k})^{2}
    \mathcal{P}(\lambda_{k})^{2}.
\end{equation*}
First, we consider the nontrivial term \eqref{cal-nontrivial}. By \eqref{asy-square-eig-1}, we have
\begin{equation}\label{asy-P1-der0}
    \mathbf{P}_{1}(\lambda_{k})\sim {\rm e}^{2{\rm i}a_{k}x}{\rm e}^{-2b_{k}x}\hat{\mathbf{P}}_{1}^{\pm}(\lambda_{k})
\end{equation}
and 
\begin{equation}\label{asy-P1-der2}
    \mathbf{P}_{1,\lambda \lambda}(\lambda_{k})\sim {\rm e}^{2{\rm i}a_{k}x}{\rm e}^{-2b_{k}x}
    \left(\hat{\mathbf{P}}_{1,\lambda \lambda}^{\pm}(\lambda_{k})+4{\rm i}x\hat{\mathbf{P}}_{1,\lambda}^{\pm}(\lambda_{k})
    -4x^{2}\hat{\mathbf{P}}_{1}^{\pm}(\lambda_{k})\right)
\end{equation}
as $x\to \pm \infty$. 
As $x\to + \infty$, we obtain
\begin{equation*}
    \mathrm{Tr}\left(\mathbf{P}_{1,\eta\eta}(\lambda_{k}) \mathbf{P}_{1}(\lambda_{k})\right)\to 0
\end{equation*}
since each term contains the factor ${\rm e}^{-2b_{k}x}$ which decays exponentially. For $x\to -\infty$, 
we obtain 
\begin{align}
    \hat{\mathbf{P}}_{1}^{-}(\lambda_{k})=&c_{1k}{\rm e}^{-2{\rm i}a_{k}x+2b_{k}x}
    \left[
        \mathcal{P}(\lambda_{k})\mathbf{e}_{1}-\sum_{s,r=1}^{N}G_{-}^{rs}(\lambda_{k})\begin{pmatrix}
            1 \\
            \mathbf{c}_{s}{\rm e}^{-2{\rm i}a_{s}x+2b_{s}x}
        \end{pmatrix}
    \right]\\&\left[
        -\sum_{s=1}^{N}(G_{-}^{ks}(\lambda_{k}^{*}))^{*}\begin{pmatrix}
            1 \\
            \mathbf{c}_{s}^{*}{\rm e}^{2{\rm i}a_{s}x+2b_{s}x}
        \end{pmatrix}^{T}
    \right]\\ 
    =&c_{1k}{\rm e}^{-2{\rm i}a_{k}x+2b_{k}x}
    \left[
        \mathcal{P}(\lambda_{k})\mathbf{e}_{1}-\sum_{s=1}^{N}\mathcal{P}(\lambda_{k})\delta_{sk} \begin{pmatrix}
            1 \\
            \mathbf{c}_{s}{\rm e}^{-2{\rm i}a_{s}x+2b_{s}x}
        \end{pmatrix}
    \right]\\&\left[
        -\sum_{s=1}^{N}(G_{-}^{ks}(\lambda_{k}^{*}))^{*}\begin{pmatrix}
            1 \\
            \mathbf{c}_{s}^{*}{\rm e}^{2{\rm i}a_{s}x+2b_{s}x}
        \end{pmatrix}^{T}
    \right]\\
    =&c_{1k}\mathcal{P}(\lambda_{k}){\rm e}^{-4{\rm i}a_{k}x+4b_{k}x}
    \left[
        \begin{pmatrix}
            0 \\
            \mathbf{c}_{k}
        \end{pmatrix}
    \right]\left[
        \sum_{s=1}^{N}(G_{-}^{ks}(\lambda_{k}^{*}))^{*}\begin{pmatrix}
            1 \\
            \mathbf{c}_{s}^{*}{\rm e}^{2{\rm i}a_{s}x+2b_{s}x}
        \end{pmatrix}^{T}
    \right]\label{cal-P-der0}
\end{align}
since
\begin{equation*}
    \sum_{r=1}^{N}G_{-}^{rs}(\lambda_{k})=\mathcal{K}^{s}(\lambda_{k})=\mathcal{P}(\lambda_{k})\delta_{sk}. 
\end{equation*}
For the derivative of $\hat{\mathbf{P}}_{1}^{-}$, we obtain
\begin{align}
    \hat{\mathbf{P}}_{1,\lambda}^{-}(\lambda_{k})
    =&\mathcal{P}(\lambda_{k}){\rm e}^{-2{\rm i}a_{k}x+2b_{k}x}
    \left[
        -\begin{pmatrix}
            0 \\
            \mathbf{c}_{k}
        \end{pmatrix}
    \right]\\&\left[
        \mathcal{P}^{*}_{\lambda}(\lambda_{k}^{*})\mathbf{e}_{2}^{T}
        -\sum_{s,r=1}^{N}(G_{-,\lambda}^{rs}(\lambda_{k}^{*}))^{*}\begin{pmatrix}
            c_{1r}{\rm e}^{-2{\rm i}a_{r}x+2b_{r}x} \\
            c_{1r}\mathbf{c}_{s}^{*}{\rm e}^{2{\rm i}(a_{s}-a_{r})x+2(b_{s}+b_{r})x}
        \end{pmatrix}^{T}
    \right]\\
    &+c_{1k}{\rm e}^{-2{\rm i}a_{k}x+2b_{k}x}\left[
        \mathcal{P}_{\lambda}(\lambda_{k})\mathbf{e}_{1}-\sum_{s=1}^{N}\mathcal{K}_{\lambda}^{s}(\lambda_{k})\begin{pmatrix}
            1 \\
            \mathbf{c}_{s}{\rm e}^{-2{\rm i}a_{s}x+2b_{s}x}
        \end{pmatrix}
    \right]\\&\left[
        -\sum_{s=1}^{N}(G_{-}^{ks}(\lambda_{k}^{*}))^{*}\begin{pmatrix}
            1 \\
            \mathbf{c}_{s}^{*}{\rm e}^{2{\rm i}a_{s}x+2b_{s}x}
        \end{pmatrix}^{T}
    \right]\label{cal-P-der1}
\end{align}
and 
\begin{align*}
    \hat{\mathbf{P}}_{1,\lambda \lambda}^{-}(\lambda_{k})
    =&\mathcal{P}(\lambda_{k}){\rm e}^{-2{\rm i}a_{k}x+2b_{k}x}
    \left[
        -\begin{pmatrix}
            0 \\
            \mathbf{c}_{k}
        \end{pmatrix}
    \right]\\&\left[
        \mathcal{P}^{*}_{\lambda \lambda}(\lambda_{k}^{*})\mathbf{e}_{2}^{T}
        -\sum_{s,r=1}^{N}(G_{-,\lambda \lambda}^{rs}(\lambda_{k}^{*}))^{*}\begin{pmatrix}
            c_{1r}{\rm e}^{-2{\rm i}a_{r}x+2b_{r}x} \\
            c_{1r}\mathbf{c}_{s}^{*}{\rm e}^{2{\rm i}(a_{s}-a_{r})x+2(b_{s}+b_{r})x}
        \end{pmatrix}^{T}
    \right]\\
    &+c_{1k}{\rm e}^{-2{\rm i}a_{k}x+2b_{k}x}\left[
        \mathcal{P}_{\lambda \lambda}(\lambda_{k})\mathbf{e}_{1}-\sum_{s=1}^{N}\mathcal{K}_{\lambda \lambda}^{s}(\lambda_{k})\begin{pmatrix}
            1 \\
            \mathbf{c}_{s}{\rm e}^{-2{\rm i}a_{s}x+2b_{s}x}
        \end{pmatrix}
    \right]\\&\left[
        -\sum_{s=1}^{N}(G_{-}^{ks}(\lambda_{k}^{*}))^{*}\begin{pmatrix}
            1 \\
            \mathbf{c}_{s}^{*}{\rm e}^{2{\rm i}a_{s}x+2b_{s}x}
        \end{pmatrix}^{T}
    \right]\\
    &+2\left[
        \mathcal{P}_{\lambda}(\lambda_{k})\mathbf{e}_{1}-\sum_{s=1}^{N}\mathcal{K}_{\lambda}^{s}(\lambda_{k})\begin{pmatrix}
            1 \\
            \mathbf{c}_{s}{\rm e}^{-2{\rm i}a_{s}x+2b_{s}x}
        \end{pmatrix}
    \right] \\&\left[
        \mathcal{P}^{*}_{\lambda}(\lambda_{k}^{*})\mathbf{e}_{2}^{T}
        -\sum_{s,r=1}^{N}(G_{-,\lambda}^{rs}(\lambda_{k}^{*}))^{*}\begin{pmatrix}
            c_{1r}{\rm e}^{-2{\rm i}a_{r}x+2b_{r}x} \\
            c_{1r}\mathbf{c}_{s}^{*}{\rm e}^{2{\rm i}(a_{s}-a_{r})x+2(b_{s}+b_{r})x}
        \end{pmatrix}^{T}
    \right].
\end{align*}
As $x\to -\infty$, we collect the constant terms and obtain
\begin{align}
    \mathrm{Tr}\left(\mathbf{P}_{1,\eta\eta}(\lambda_{k}) \mathbf{P}_{1}(\lambda_{k})\right)=&
    \mathrm{Tr}\left(
        2\left(\mathcal{P}_{\lambda}(\lambda_{k})-\sum_{s=1}^{N}\mathcal{K}_{\lambda}^{s}(\lambda_{k})\right)\mathbf{e}_{1}
        \mathcal{P}^{*}_{\lambda}(\lambda_{k}^{*})\mathbf{e}_{2}^{T}
    \right.\\&\left.
        c_{1k}\mathcal{P}(\lambda_{k})\begin{pmatrix}
            0 \\
            \mathbf{c}_{k}
        \end{pmatrix}
        \sum_{s=1}^{N}(G_{-}^{ks}(\lambda_{k}^{*}))^{*}\mathbf{e}_{1}^{T}
    \right)+o(1)\\
    =& 
    2c_{1k}^{2}\left(\mathcal{P}_{\lambda}(\lambda_{k})-\sum_{s=1}^{N}\mathcal{K}_{\lambda}^{s}(\lambda_{i})\right)
        \mathcal{P}^{*}_{\lambda}(\lambda_{k}^{*})
        \mathcal{P}(\lambda_{k})
        \sum_{s=1}^{N}(G_{-}^{ks}(\lambda_{k}^{*}))^{*}+o(1). \label{P1P1-1}
\end{align}
Now we analyze the right-hand side of \eqref{P1P1-1} to simplify the expression. Using \eqref{root-Ks}, we have
\begin{equation}\label{root-sum-Ks}
    \sum_{s=1}^{N}\mathcal{K}^{s}(\lambda_{k})=\mathcal{P}(\lambda_{k})
\end{equation}
for $k=1,2,\cdots,N$. 
Since both the polynomial $\mathcal{P}(\lambda)-\hat{\mathcal{P}}(\lambda)$ 
and $\sum_{s=1}^{N}\mathcal{K}^{s}(\lambda)$ are of degree $N-1$, and 
\begin{equation}\label{root-P-Phat}
    \mathcal{P}(\lambda_{k})-\hat{\mathcal{P}}(\lambda_{k})=\mathcal{P}(\lambda_{k}),
\end{equation}
it follows from \eqref{root-sum-Ks} and \eqref{root-P-Phat} that
\begin{equation*}
    \sum_{s=1}^{N}\mathcal{K}^{s}(\lambda)=\mathcal{P}(\lambda)-\hat{\mathcal{P}}(\lambda). 
\end{equation*}
Hence 
\begin{equation*}
    \mathcal{P}_{\lambda}(\lambda_{k})-\sum_{s=1}^{N}\mathcal{K}_{\lambda}^{s}(\lambda_{k})=\hat{\mathcal{P}}_{\lambda}(\lambda_{k}). 
\end{equation*}
It remains to consider the term
\begin{align}
    \sum_{s=1}^{N}G_{-}^{ks}(\lambda_{k}^{*})=&2\ii b_{k}\mathcal{P}_{\lambda}(\lambda_{k}^{*} )\sum_{s=1}^{N}
    \frac{(-1)^{k+s}F_{-}^{ks}}{F_{-}}\\
    =&2\ii b_{k}\mathcal{P}_{\lambda}(\lambda_{k}^{*})\frac{F_{-}^{(k)}}{F_{-}} \\
    =&-\mathcal{P}(\lambda_{k})^{*}, \label{sum-Gis-1}
\end{align}
where $F_{-}^{(k)}$ is the determinant derived from $F_{-}$ by replacing its $k$-th row with all entries equal to 1, i.e.,
\begin{equation*}
    F_{-}^{(k)}=\begin{vmatrix}
        \frac{\lambda_{1}-\lambda_{1}^{*}}{\lambda_{1}-\lambda_{1}^{*}} & \frac{\lambda_{1}-\lambda_{1}^{*}}{\lambda_{2}-\lambda_{1}^{*}}  & \cdots & \frac{\lambda_{1}-\lambda_{1}^{*}}{\lambda_{N}-\lambda_{1}^{*}} \\
        \frac{\lambda_{2}-\lambda_{2}^{*}}{\lambda_{1}-\lambda_{2}^{*}} & \frac{\lambda_{2}-\lambda_{2}^{*}}{\lambda_{2}-\lambda_{2}^{*}} & \cdots & \frac{\lambda_{2}-\lambda_{2}^{*}}{\lambda_{N}-\lambda_{2}^{*}} \\
        \vdots & \vdots & \ddots & \vdots  \\
        1 & 1  & \cdots & 1 \\
        \vdots & \vdots & \ddots & \vdots \\
        \frac{\lambda_{N}-\lambda_{N}^{*}}{\lambda_{1}-\lambda_{N}^{*}} & \frac{\lambda_{N}-\lambda_{N}^{*}}{\lambda_{2}-\lambda_{N}^{*}}  &\cdots & \frac{\lambda_{N}-\lambda_{N}^{*}}{\lambda_{N}-\lambda_{N}^{*}}
        \end{vmatrix}. 
\end{equation*}
Now we prove that 
\begin{equation}\label{Fr-F}
    F^{(k)}_{-}=-\frac{\mathcal{P}(\lambda_{k})^{*}}{\mathcal{P}_{\lambda}(\lambda_{k}^{*})2\ii b_{k}}F_{-}
\end{equation}
for $k=1,2,\cdots, N$, which implies that \eqref{sum-Gis-1} holds. To analyze $F^{(k)}_{-}$, 
subtract column 1 from each of the columns 2 through $n$, we obtain
\begin{equation*}
    \frac{1}{\lambda_{j}-\lambda_{i}^{*}}-\frac{1}{\lambda_{1}-\lambda_{i}^{*}}=
    \frac{\lambda_{1}-\lambda_{j}}{(\lambda_{j}-\lambda_{i}^{*})(\lambda_{1}-\lambda_{i}^{*})},
\end{equation*}
hence
\begin{align*}
    F^{(k)}_{-}=&(-1)^{k+1}\prod_{i\ne k}(\lambda_{i}-\lambda_{i}^{*})\prod_{j\ne 1}(\lambda_{1}-\lambda_{j})
    \prod_{i\ne k}\frac{1}{\lambda_{1}-\lambda_{i}^{*}}\mathrm{det}\left(
        \frac{1}{\lambda_{j}-\lambda_{i}^{*}}
    \right)_{i\ne k,j\ne 1}\\ =&(-1)^{k+1}\prod_{j\ne 1}(\lambda_{1}-\lambda_{j})
    \prod_{i\ne k}\frac{1}{\lambda_{1}-\lambda_{i}^{*}}F^{k1}_{-}. 
\end{align*}
In addition, for $F_{-}$, subtracting column 1 multiplied by $(\lambda_{1}-\lambda_{k}^{*})/(\lambda_{j}-\lambda_{k}^{*})$ 
from each of columns $j$ varying from 2 to $n$, we obtain 
\begin{equation*}
    \frac{1}{\lambda_{j}-\lambda_{i}^{*}}-\frac{1}{\lambda_{1}-\lambda_{i}^{*}}\frac{\lambda_{1}-\lambda_{k}^{*}}{\lambda_{j}-\lambda_{k}^{*}}=
    \frac{(\lambda_{1}-\lambda_{j})(\lambda_{i}^{*}-\lambda_{k}^{*})}{(\lambda_{j}-\lambda_{i}^{*})(\lambda_{1}-\lambda_{i}^{*})(\lambda_{j}-\lambda_{k}^{*})},
\end{equation*}
hence 
\begin{align*}
    F_{-}=&(-1)^{k+1}\prod_{j}(\lambda_{j}-\lambda_{j}^{*})\prod_{j\ne 1}(\lambda_{1}-\lambda_{j})
    \prod_{i\ne k}(\lambda_{i}^{*}-\lambda_{k}^{*})
    \\&\frac{1}{\prod_{i \ne k}(\lambda_{1}-\lambda_{i}^{*})\prod_{j}(\lambda_{j}-\lambda_{k}^{*})}
    \mathrm{det}\left(
        \frac{1}{\lambda_{j}-\lambda_{i}^{*}}
    \right)_{i\ne k,j\ne 1}\\
    =&(-1)^{k+1}2\ii b_{k}\prod_{j\ne 1}(\lambda_{1}-\lambda_{j})
    \prod_{i\ne k}(\lambda_{i}^{*}-\lambda_{k}^{*})\frac{1}{\prod_{i\ne k}(\lambda_{1}-\lambda_{i}^{*})\prod_{j}(\lambda_{j}-\lambda_{k}^{*})}F^{k1}_{-}.
\end{align*}
Then the equality \eqref{Fr-F} holds. Hence we obtain \eqref{ort-dis-3} for the case $i=j=1$. Other 
cases can be considered by similar method. 

By \eqref{cal-3}, \eqref{cal-P-der0}, and \eqref{cal-P-der1}, it follows that
\begin{equation}
        \omega(\mathbf{S}_{-1,\eta}(\eta),\mathbf{S}_{1}(\lambda))=0
\end{equation}
for $\lambda\ne\eta,\lambda\in\{\lambda_{k}:i=1,2,\cdots,N\},\eta\in\{\lambda_{k}^{*}:i=1,2,\cdots,N\}$ since 
\begin{equation}\label{asy-P1-der1}
    \mathbf{P}_{1,\lambda}(\lambda_{k})\sim{\rm e}^{2{\rm i}a_{k}x}{\rm e}^{-2b_{k}x}
    \left(\hat{\mathbf{P}}_{1,\lambda}^{\pm}(\lambda_{k})+2{\rm i}x\hat{\mathbf{P}}_{1}^{\pm}(\lambda_{k})\right). 
\end{equation}
To derive \eqref{cal-1} and \eqref{cal-2}, the asymptotic behavior of
\begin{equation*}
    \mathbf{P}_{-1}(\lambda_{k}^{*}),\quad \mathbf{P}_{-1,\lambda}(\lambda_{k}^{*})
\end{equation*}
as $x\to\infty$ is required.
Since 
\begin{equation*}
    \mathbf{P}_{-1}(\lambda)\sim{\rm e}^{-2{\rm i}\lambda x}\hat{\mathbf{P}}_{-1}^{\pm}(\lambda),\quad x\to\pm\infty,
\end{equation*}
it follows that
\begin{align*}
    \mathbf{P}_{-1}(\lambda_{k}^{*})&\sim {\rm e}^{-2{\rm i}a_{k} x}{\rm e}^{-2 b_{k} x}\left(
        \hat{\mathbf{P}}_{-1}^{\pm}(\lambda_{k}^{*})
    \right),\\
    \mathbf{P}_{-1,\lambda}(\lambda_{k}^{*})&\sim {\rm e}^{-2{\rm i}a_{k} x}{\rm e}^{-2 b_{k} x}\left(
        \hat{\mathbf{P}}_{-1,\lambda}^{\pm}(\lambda_{k}^{*})-2{\rm i}x\hat{\mathbf{P}}_{-1}^{\pm}(\lambda_{k}^{*})
    \right). 
\end{align*}
Since $\mathbf{P}_{-1}(\lambda_{k}^{*})\to 0$ as $x\to +\infty$ due to the exponential decay factor 
${\rm e}^{-2 b_{k} x}$, it remains to consider the limit $x\to -\infty$. Using
\begin{align*}
    \hat{\mathbf{P}}_{-1}^{-}(\lambda^{*})=\left(\hat{\mathbf{P}}_{1}^{-}(\lambda)\right)^{\dagger}, 
\end{align*}
and applying \eqref{cal-P-der0} and \eqref{cal-P-der1}, it can be concluded that the 
expressions in \eqref{cal-1}, \eqref{cal-2}, and \eqref{cal-3} vanish.

Now we turn to \eqref{ort-dis-1}. By \eqref{squ-eig-connect-squ-eig-mat-2},
\begin{equation*}
    \int_{\mathbb{R}}\mathbf{S}_{i}^{\dagger}(\lambda';x)\mathcal{J}\mathbf{S}_{j}(\lambda;x)\mathrm{d}x=
    \frac{1}{2(\lambda-\lambda')}\mathrm{Tr}\left.\left(
        \mathbf{P}_{i}^{\dagger}(\lambda';x)\mathbf{P}_{j}(\lambda;x)
    \right)\right|_{-\infty}^{+\infty}
\end{equation*}
for $\lambda'\in\mathbb{R}$. Consider the case $\lambda=\lambda_{k}$ and $j=1$. The right-hand side vanishes by \eqref{asy-P1-der0} and \eqref{cal-P-der0}. Other cases follow by similar arguments.
This completes the proof. 
\end{proof}
Now the spectrum of the operator $\mathcal{L}$ is considered with the aid of the operator $\mathcal{J}\mathcal{L}$. 
\subsection{The spectrum analysis of $\mathcal{L}$}
After establishing the orthogonality conditions for the squared eigenfunctions,
the spectrum of $\mathcal{L}$ can be analyzed using their completeness \cite{yang_nonlinear_2010,ling2024stability}.
Denote
\begin{align*}
    \mathsf{E}_{ess}^{+}=&\{ \mathbf{S}_{i}(\lambda;x):i=1,2, \ \lambda\in\sigma_{ess}(\mathcal{L}_{s}) \},\\
    \mathsf{E}_{point}^{+}=&\{\mathbf{S}_{1}(\lambda_{k};x),\mathbf{S}_{2}(\lambda_{k}^{*};x):k\notin\Gamma\}\cup\{\mathbf{S}_{2}(\lambda_{k};x),\mathbf{S}_{1}(\lambda_{k}^{*};x):k\in\Gamma\}, \\
    \hat{\mathsf{E}}_{point}^{+}=&\{\mathbf{S}_{1,\lambda}(\lambda_{k};x):k\notin \Gamma\}\cup\{\mathbf{S}_{2,\lambda}(\lambda_{k};x):k\in \Gamma\}.
\end{align*}
Consider the transformation
\begin{equation*}
    \mathcal{C}: L^{2}(\mathbb{R},\mathbb{C}^{4})\to \mathrm{X}:\quad 
    \mathcal{C}\mathbf{P}=\mathbf{P}+(\mathbf{\Sigma} \mathbf{P})^{*}
\end{equation*}
which commutes with the operator $\mathcal{L}$. 
Let $\mathcal{C}\mathsf{E}=\left\{\mathcal{C}\mathbf{f}:\mathbf{f}\in \mathsf{E}\right\}$ for a given set 
$\mathsf{E}$. 
The set consisting of the squared eigenfunctions forms a basis in $L^{2}$, 
and the following lemma holds.
\begin{lem}\label{L2-X-basic}
    The space $L^{2}(\mathbb{R},\mathbb{C}^{4})$ admits a decomposition:
    \begin{equation}\label{complete-eigen-1}
        L^{2}(\mathbb{R},\mathbb{C}^{4})=\mathbb{E}_{ess}+\mathbb{E}_{point}
    \end{equation}
    where the essential spectrum subspace is defined by \eqref{E-ess} as
    \begin{equation*}
        \mathbb{E}_{ess}=\mathrm{span}\left\{
            \int_{\mathbb{R}}\omega(\lambda)\mathbf{S}(\lambda;x)\mathrm{d}\lambda:\mathbf{S}(\lambda;x)\in
            \mathsf{E}_{ess},\omega(\lambda)\in L^{2}(\mathbb{R},\mathbb{C})
        \right\}
    \end{equation*}
    and the point spectrum subspace is given by \eqref{E-point} and \eqref{hat-E-point} as
    \begin{equation*}
        \mathbb{E}_{point}=\mathrm{span}\left\{
            \mathbf{S}:\mathbf{S}\in \mathsf{E}_{point}\cup \hat{\mathsf{E}}_{point}
        \right\}. 
    \end{equation*}
    Moreover, the space $\mathrm{X}$ admits the decomposition
    \begin{equation}\label{complete-eigen-2}
        \mathrm{X}=\mathbb{E}_{ess}^{\mathrm{X}} +\mathbb{E}_{point}^{\mathrm{X}}.
    \end{equation}
    The subspaces $\mathbb{E}_{ess}^{\mathrm{X}}$ and $\mathbb{E}_{point}^{\mathrm{X}}$ can be represented as
    \begin{align*}
        \mathbb{E}_{ess}^{\mathrm{X}}=&\mathrm{span}\left\{
            \int_{\mathbb{R}}\omega(\lambda)\mathbf{S}(\lambda;x)\mathrm{d}\lambda:\mathbf{S}(\lambda;x)\in
            \mathcal{C}\mathsf{E}_{ess}^{+}\cup\mathcal{C}\ii\mathsf{E}_{ess}^{+}, \omega(\lambda)\in L^{2}(\mathbb{R},\mathbb{R})
        \right\},\\
        \mathbb{E}_{point}^{\mathrm{X}}=&\mathrm{span}\left\{
            \mathbf{S}:\mathbf{S}\in \mathcal{C}\mathsf{E}_{point}^{+}\cup 
            \mathcal{C}\ii\mathsf{E}_{point}^{+}\cup \mathcal{C}\hat{\mathsf{E}}_{point}^{+}
            \cup \mathcal{C}\ii\hat{\mathsf{E}}_{point}^{+}
        \right\}. 
    \end{align*}
\end{lem}
\begin{proof}
    The proof of \eqref{complete-eigen-1} can be found in \cite{yang_nonlinear_2010,ling2024stability}. Here, the focus is on proving \eqref{complete-eigen-2}. 
    Let $\mathbf{f}(x)\in \mathrm{X}$. 
    Since $\mathbf{f}\in  \mathrm{X}\subset L^{2}(\mathbb{R},\mathbb{C}^{4})$, there exist functions 
    $\mathbf{S}_{i}(\lambda;x)\in \mathsf{E}_{ess}$ with coefficients $\omega_{i}(\lambda)$ 
    and elements $\hat{\mathbf{S}}_{j}(x)\in \mathsf{E}_{point}\cup \hat{\mathsf{E}}_{point}$ 
    with coefficients $\hat{\omega}_{j}\in\mathbb{C}$ such that
    \begin{equation*}
        \mathbf{f}=\sum_{i}\int \omega_{i}(\lambda)\mathbf{S}_{i}(\lambda;x)\mathrm{d}\lambda+
        \sum_{j}\hat{\omega}_{j}\hat{\mathbf{S}}_{j}(x). 
    \end{equation*}
    Since $\mathcal{C}\mathbf{f}=2\mathbf{f}$, we obtain 
    \begin{equation*}
        \begin{split}
            \mathbf{f}=&\frac{1}{2}\left(
            \sum_{i}\int \mathrm{Re}(\omega_{i}(\lambda))\mathcal{C}\mathbf{S}_{i}(\lambda;x)\mathrm{d}\lambda+
            \int \mathrm{Im}(\omega_{i}(\lambda))\mathcal{C}\ii\mathbf{S}_{i}(\lambda;x)\mathrm{d}\lambda
            \right.\\&\left.+\sum_{j}\mathrm{Re}(\hat{\omega}_{j})\mathcal{C}\hat{\mathbf{S}}_{j}(x)
        +\sum_{j}\mathrm{Im}(\hat{\omega}_{j})\mathcal{C}\ii\hat{\mathbf{S}}_{j}(x)
        \right). 
        \end{split}
    \end{equation*}
    In view of the symmetry \eqref{symmetric-squared-eigenfunctions}, we have 
    \begin{equation}\label{sym-C-squared}
        \begin{split}
            &\mathcal{C}\mathbf{S}_{i}(\lambda)=-\mathcal{C}\mathbf{S}_{-i}(\lambda^{*}),\
        \mathcal{C}\ii\mathbf{S}_{i}(\lambda)=\mathcal{C}\ii\mathbf{S}_{-i}(\lambda^{*}), \\
        &\mathcal{C}\mathbf{S}_{i,\lambda}(\lambda)=-\mathcal{C}\mathbf{S}_{-i,\lambda}(\lambda^{*}),\
        \mathcal{C}\ii\mathbf{S}_{i,\lambda}(\lambda)=\mathcal{C}\ii\mathbf{S}_{-i,\lambda}(\lambda^{*}). 
        \end{split}
    \end{equation}
    Hence, it suffices to consider $\mathbf{S}_{i}(\lambda;x)\in \mathcal{C}\mathsf{E}_{ess}^{+}\cup\mathcal{C}\ii\mathsf{E}_{ess}^{+}$ and 
    $\hat{\mathbf{S}}_{j}(x)\in \mathcal{C}(\mathsf{E}_{point}^{+}\cup \hat{\mathsf{E}}_{point}^{+})
    \cup \mathcal{C}\ii(\mathsf{E}_{point}^{+}\cup \hat{\mathsf{E}}_{point}^{+})$. This completes the proof.
\end{proof}

The squared eigenfunctions constitute a complete set in the $L^{2}$ space. To analyze the spectrum of the operator
$\mathcal{L}$, it suffices to evaluate the quadratic form $(\mathcal{L}\cdot,\cdot)$ on the set of 
squared eigenfunctions.

Denote 
\begin{equation*}
    Q_{\pm}(\lambda;\Omega)=\Omega \mp 2^{2N}{\rm i}\mathcal{P}(\lambda)\hat{\mathcal{P}}(\lambda).
\end{equation*}
All solutions to the spectral problem
\begin{equation}\label{JL-spec-pro}
\mathcal{J}\mathcal{L} \mathbf{f} = \Omega \mathbf{f}, \quad \Omega \in \mathbb{C},
\end{equation}
can be expressed in terms of the squared eigenfunctions as follows:
\begin{lem}\label{JL-eigenvalue}
    The squared eigenfunctions $\mathbf{S}_{\pm i}(\lambda)$ satisfy the spectral problem associated with 
    $\mathcal{J}\mathcal{L}$ 
    \begin{equation}\label{JL-spectral-problem}
        \mathcal{J}\mathcal{L}\mathbf{S}_{\pm i}(\lambda)=\pm 2^{2N}{\rm i}\mathcal{P}(\lambda)\hat{\mathcal{P}}(\lambda)\mathbf{S}_{\pm i}(\lambda),\quad i=1,2.
    \end{equation}
    If the polynomial $Q_{+}(\lambda;\Omega)Q_{-}(-\lambda;\Omega)$ in $\lambda$ has no multiple roots, and 
    $\lambda_{k},\lambda_{k}^{*}$ for $k=1,2,\cdots, N$
    are not roots of $Q_{+}(\lambda;\Omega)$ and $Q_{-}(\lambda;\Omega)$(i.e. $\Omega\ne 0$), 
    then the solutions to the spectral problem \eqref{JL-spec-pro} can be obtained by
    \begin{equation*}
        \mathrm{span}\left\{\{\mathbf{S}_{i}(\lambda):Q_{+}(\lambda;\Omega)=0,i=1,2\}\cup 
        \{\mathbf{S}_{-i}(\lambda):Q_{-}(\lambda;\Omega)=0,i=1,2\}\right\}.
    \end{equation*}
    For other cases, the solutions to \eqref{JL-spec-pro} can be obtained as limits of the FMS.
    The essential spectrum of $\mathcal{J}\mathcal{L}$ is given by
    \begin{equation*}
        \sigma_{ess}(\mathcal{J}\mathcal{L})=(-\ii \infty, -2^{2N}\ii \min_{\lambda \in \mathbb{R}}|\mathcal{P}(\lambda)|^{2}]\cup 
        [2^{2N}\ii \min_{\lambda \in \mathbb{R}}|\mathcal{P}(\lambda)|^{2},\ii \infty)
    \end{equation*}
    with the $L^{\infty}$ solution basis given by $\mathsf{E}_{ess}$. The point spectrum is
    \begin{equation*}
        \sigma_{point}(\mathcal{J}\mathcal{L})=\{0\},
    \end{equation*}
    with the $L^{2}$ eigenfunctions given by $\mathsf{E}_{point}$ 
    and the $L^{2}$ generalized eigenfunctions given by $\hat{\mathsf{E}}_{point}$. 
\end{lem}
\begin{proof}
    The proof relies on properties of the linearized operator. We consider the associated Lax pair
    \begin{align}
        \mathbf{\Phi}_{N,x}(\lambda;x,t_{N})=&\mathbf{U}(\lambda,\mathbf{Q})\mathbf{\Phi}_{N}(\lambda;x,t_{N})\label{x-part-N}\\
        \mathbf{\Phi}_{N,t_{N}}(\lambda;x,t_{N})=&\sum_{n=0}^{2N}2^{n}\mu_{n}\mathbf{V}_{n}(\lambda,\mathbf{Q})\mathbf{\Phi}_{N}(\lambda;x,t_{N})\label{t-part-N}
    \end{align}
    with time variable $t_{N}$. Let $\mathbf{\Phi}_{N}$ be a fundamental matrix solution of the Lax pair 
    \eqref{x-part-N} and \eqref{t-part-N}. Then, the squared eigenfunction matrices defined in 
    \eqref{squared-eigenfunction-matrix} via $\mathbf{\Phi}_{N}$ satisfy 
    \begin{equation*}
        \mathbf{B}_{N,x}=[\mathbf{U},\mathbf{B}_{N}],\quad 
        \mathbf{B}_{N,t_{N}}=\sum_{n=0}^{2N}2^{n}\mu_{n}[\mathbf{V}_{n},\mathbf{B}_{N}].
    \end{equation*}
    By Theorem \ref{off-diagonal-variation}, we obtain
    \begin{equation}
        \mathbf{B}_{N,t_{N}}^{\perp}=-2\sum_{n=0}^{2N}2^{n}\mu_{n}
        \frac{\delta \mathbf{L}_{n+1}^{\perp}}{\delta \mathbf{Q}}(\sigma_{3}\mathbf{B}_{N}^{\perp}). 
    \end{equation}
    From the symmetry $\mathbf{Q}=-\mathbf{Q}^{\dagger}$ and \eqref{def-var-Q}, we have
    \begin{equation*}
        \begin{pmatrix}
            0 & -\frac{\delta \mathcal{H}_{n}}{\delta \mathbf{q}}^{*}(-\mathbf{q}^{*},\mathbf{q})^{T} \\
            \frac{\delta \mathcal{H}_{n}}{\delta \mathbf{q}}(\mathbf{q},-\mathbf{q}^{*}) & 0
        \end{pmatrix}=-2\frac{\delta \mathcal{H}_{n}}{\delta \mathbf{Q}}(\mathbf{Q})=2^{n}\ii \mathbf{L}_{n+1}^{\perp}(\mathbf{Q})
    \end{equation*} 
    and hence
    \begin{equation*}
        \begin{pmatrix}
            0 & -\frac{\delta^{2} \mathcal{H}_{n}}{\delta \mathbf{q}^{2}}^{*}(-\mathbf{h},\mathbf{g})^{T} \\
            \frac{\delta^{2} \mathcal{H}_{n}}{\delta \mathbf{q}^{2}}(\mathbf{g},-\mathbf{h}) & \mathbf{0}
        \end{pmatrix}=2^{n} \ii \frac{\delta \mathbf{L}_{n+1}^{\perp}}{\delta \mathbf{Q}}\begin{pmatrix}
            0 & -\mathbf{h}^{T} \\
            \mathbf{g} & \mathbf{0}
        \end{pmatrix}.
    \end{equation*}
    Let 
    \begin{equation*}
        \mathbf{B}_{N}=\begin{pmatrix}
            f_{N} & \mathbf{h}_{N}^{T} \\
            \mathbf{g}_{N} & -\mathbf{f}_{1,N}
        \end{pmatrix},
    \end{equation*} 
    it leads to
    \begin{equation*}
        -\sigma_{3}\mathbf{B}_{N,t_{N}}^{\perp}=2\ii\sigma_{3}\sum_{n=0}^{2N}\mu_{n}
        \begin{pmatrix}
            0 & \frac{\delta^{2} \mathcal{H}_{n}}{\delta \mathbf{q}^{2}}^{*}(\mathbf{h}_{N},-\mathbf{g}_{N})^{T} \\
            \frac{\delta^{2} \mathcal{H}_{n}}{\delta \mathbf{q}^{2}}(\mathbf{g}_{N},-\mathbf{h}_{N}) & \mathbf{0}
        \end{pmatrix}.
    \end{equation*}
    Hence, the squared eigenfunctions given in \eqref{squared-eigenfunction-matrix} by the FMS $\mathbf{\Phi}_{N}$ satisfy
    \begin{equation}\label{JL-spectral-problem-1}
        \begin{pmatrix}
            \mathbf{g}_{N} \\
            -\mathbf{h}_{N}
        \end{pmatrix}_{t_{N}}=2\mathcal{J}\mathcal{L}\begin{pmatrix}
            \mathbf{g}_{N} \\
            -\mathbf{h}_{N}
        \end{pmatrix}. 
    \end{equation}
    Now, applying the $N$-fold Darboux transformation to the solution
    \begin{equation*}
        \mathbf{\Phi}_{N}^{[0]}={\rm e}^{\ii(\lambda x+2^{2N}\mathcal{P}(\lambda)\hat{\mathcal{P}}(\lambda)t_{N})\sigma_{3}}
    \end{equation*}
    of the Lax pair \eqref{x-part-N} and \eqref{t-part-N} associated with the zero potential, 
    with spectral parameters $\lambda_{1},\lambda_{2},\cdots, \lambda_{N}$ and scattering parameters 
    $\mathbf{c}_{1}, \mathbf{c}_{2},\cdots, \mathbf{c}_{N}$, we obtain that the new FMS takes the form
    \begin{equation*}
        \mathbf{\Phi}_{N}^{[N]}(\lambda;x,t_{N})=\mathbf{\Phi}^{[N]}(\lambda;x,0)
        {\rm e}^{\ii 2^{2N}\mathcal{P}(\lambda)\hat{\mathcal{P}}(\lambda)t_{N}\sigma_{3}}. 
    \end{equation*}
    Hence, \eqref{JL-spectral-problem} follows directly from \eqref{JL-spectral-problem-1}. 

    For fixed $\Omega\ne 0$, we denote the roots of $Q_{+}(\lambda;\Omega)=0$ by $\hat{\lambda}_{k,+}$ and 
    the roots of $Q_{-}(\lambda;\Omega)=0$ by $\hat{\lambda}_{k,-}$ for $k=1,2,\cdots,2N$. If 
    $Q_{+}(\lambda;\Omega)Q_{-}(-\lambda;\Omega)$ has no multiple root, then $\{\hat{\lambda}_{k,+},
    -\hat{\lambda}_{k,-}:k=1,2,\cdots,2N\}$ are distinct. Since the asymptotic 
    behaviors of the squared eigenfunctions for different $\lambda$ are independent by \eqref{asy-square-eig-1}, 
    we find the following: $\mathbf{S}_{i}(\hat{\lambda}_{k,+})$ is linearly independent of $\mathbf{S}_{-j}(\hat{\lambda}_{k,-})$
    for $i,j=1,2$ and $k=1,2,\cdots,N$; $\mathbf{S}_{i}(\hat{\lambda}_{k,+})$ ($\mathbf{S}_{-i}(\hat{\lambda}_{k,-})$) is linearly independent of
    $\mathbf{S}_{j}(\hat{\lambda}_{l,+})$ ($\mathbf{S}_{-j}(\hat{\lambda}_{l,-})$) for $i=1,2$ and $k\ne l$. 
    
    Moreover, for two squared eigenfunctions corresponding to the same frequency (the case $k=l$), if
    \begin{equation*}
        d_{1}\mathbf{S}_{1}(\hat{\lambda}_{k,+})+d_{2}\mathbf{S}_{2}(\hat{\lambda}_{k,+})=0,
    \end{equation*}
    it follows that
    \begin{equation*}
        (d_{1}\mathbf{P}_{1}(\hat{\lambda}_{k,+})+d_{2}\mathbf{P}_{2}(\hat{\lambda}_{k,+}))^{\perp}=0
    \end{equation*}
    which in turn implies
    \begin{equation*}
        d_{1}\mathbf{\Phi}_{2}(\hat{\lambda}_{k,+})+d_{2}\mathbf{\Phi}_{3}(\hat{\lambda}_{k,+})=0,
    \end{equation*}
    by the definition of squared eigenfunctions. 
    Then we obtain $d_{1}=d_{2}=0$ since $\mathbf{\Phi}_{2}$ and $\mathbf{\Phi}_{3}$ are linearly independent. 
    Therefore, $\mathbf{S}_{1}(\hat{\lambda}_{k,+})$ is linearly independent of $\mathbf{S}_{2}(\hat{\lambda}_{k,+})$. 
    Similarly, $\mathbf{S}_{-1}(\hat{\lambda}_{k,-})$ is linearly independent of $\mathbf{S}_{-2}(\hat{\lambda}_{k,-})$. 
    Consequently, for the case $\Omega\ne0$ and $Q_{+}(\lambda;\Omega)Q_{-}(-\lambda;\Omega)$ having no multiple roots, we obtain 
    $8N$ linearly independent eigenfunctions of the operator $\mathcal{J}\mathcal{L}$ corresponding to the eigenvalue $\Omega$. 
    Moreover, if $\Omega\ne0$, then the set of roots of $Q_{+}(\lambda;\Omega)$ has no intersection with 
    the set of roots of $Q_{-}(\lambda;\Omega)$. Let $\mathcal{I}$ denote the identity map.
    Since $\mathcal{J}\mathcal{L}$ is of order $2N$, these eigenfunctions span the eigenspace 
    $\mathrm{Ker}(\Omega\mathcal{I}-\mathcal{J}\mathcal{L})$. 

    If $\Omega \ne 0$ and $Q_{+}(\lambda;\Omega)Q_{-}(-\lambda;\Omega)$ has multiple roots, we rewrite the spectral problem \eqref{JL-spec-pro} as a first-order ODE
    \begin{equation}\label{JL-spec-pro-1}
        \partial_{x}\begin{pmatrix}
            \mathbf{f} \\
            \mathbf{f}_{x} \\
            \mathbf{f}_{xx} \\
            \vdots \\
            \mathbf{f}_{(2N-1)}
        \end{pmatrix}=
        \begin{pmatrix}
            0 & 1 & 0 & \cdots & 0 \\
            0 & 0 & 1 & \cdots & 0 \\
            0 & 0 & 0 & \cdots & 0 \\
            \vdots &\vdots & \vdots & \ddots & \vdots \\
            \mathbf{A}_{2N,1}(\Omega;\mathbf{q}) & \mathbf{A}_{2N,2}(\Omega;\mathbf{q}) & \mathbf{A}_{2N,3}(\Omega;\mathbf{q})
             & \cdots & \mathbf{A}_{2N,2N}(\Omega;\mathbf{q})
        \end{pmatrix}\begin{pmatrix}
            \mathbf{f} \\
            \mathbf{f}_{x} \\
            \mathbf{f}_{xx} \\
            \vdots \\
            \mathbf{f}_{(2N-1)}
        \end{pmatrix}
    \end{equation}
    where the coefficients are determined by a differential operator of order $2N-1$
    \begin{equation*}
        \sum_{i=1}^{2N}\mathbf{A}_{2N,i}\partial_{x}^{i-1}=\mathcal{J}(\mathcal{L}-\partial_{x}^{2N}). 
    \end{equation*}
    Especially, 
    \begin{equation*}
        \mathrm{Tr}(\mathbf{A}_{2N,2N})=\mu_{2N-1}\ii (-1)^{N}\mathrm{Tr}(\mathrm{diag}(1,1,-1,-1))=0. 
    \end{equation*}
    If $\Omega\ne 0$ and $Q_{+}(\lambda;\Omega)Q_{-}(-\lambda;\Omega)$ has no multiple root, then the FMS of 
    \eqref{JL-spec-pro-1} can be obtained from \eqref{JL-spec-pro}, which we denote 
    by $\mathbf{F}_{N}(\Omega;x)$. 
    If $Q_{+}(\lambda;\Omega_{0})Q_{-}(-\lambda;\Omega_{0})$ has multiple roots, then the FMS can be given by
    \begin{equation*}
        \lim_{\Omega\to \Omega_{0}}\mathbf{F}_{N}(\Omega;x)\mathbf{F}_{N}(\Omega;0)^{-1},
    \end{equation*} 
    since, by Abel’s theorem,
    \begin{equation*}
        \mathrm{det}(\mathbf{F}_{N}(\Omega;x)\mathbf{F}_{N}(\Omega;0)^{-1})=
        \mathrm{det}(\mathbf{F}_{N}(\Omega;0)\mathbf{F}_{N}(\Omega;0)^{-1})=1. 
    \end{equation*}

    If $\Omega=0$, then $Q_{+}(\lambda;0)=Q_{-}(\lambda;0)=\mathcal{P}(\lambda)\hat{\mathcal{P}}(\lambda)$ 
    and the roots are $\lambda_{1},\lambda_{2},\cdots,\lambda_{N}$ and $\lambda_{1}^{*},\lambda_{2}^{*},\cdots,\lambda_{N}^{*}$. 
    In this case, we obtain $8N$ eigenfunctions
    \begin{equation*}
        \mathbf{S}_{\pm 1}(\lambda_{k}),\quad \mathbf{S}_{\pm 2}(\lambda_{k}),\quad 
        \mathbf{S}_{\pm 1}(\lambda_{k}^{*}),\quad \mathbf{S}_{\pm 2}(\lambda_{k}^{*}), 
    \end{equation*}
    for $k=1,2,\cdots N$. 
    Hence $\mathbf{S}_{\pm i}(\lambda_{k})$ for $i=1,2$ is linear independent of 
    $\mathbf{S}_{\pm j}(\lambda_{l}),\mathbf{S}_{\pm j}(\lambda_{r}^{*})$ for $j=1,2$, $k\ne l$ and $k,r=1,2,\cdots,N$
    by \eqref{asy-square-eig-1}. 
    By \eqref{inde-lambdai}, if $k\notin\Gamma$, there are only two independent eigenfunctions at $\lambda=\lambda_{k}$:
    \begin{equation*}
        \mathrm{span}\left\{\mathbf{S}_{1}(\lambda_{k}),\mathbf{S}_{-2}(\lambda_{k})\right\}=
        \mathrm{span}\left\{\mathbf{S}_{\pm 1}(\lambda_{k}),\mathbf{S}_{\pm 2}(\lambda_{k})\right\}. 
    \end{equation*}
    If  $k\in\Gamma$, then we take $\mathbf{S}_{2}(\lambda_{k}),\mathbf{S}_{-1}(\lambda_{k})$. 
    In this way, we obtain $4N$ linearly independent eigenfunctions $\mathsf{E}_{point}\subset \mathcal{S}(\mathbb{R},\mathbb{C}^{4})$
    for $\mathcal{J}\mathcal{L}$ at the eigenvalue $\Omega=0$. 
    The remaining $4N$ eigenfunctions can be constructed via
    \begin{equation*}
        \lim_{\Omega\to 0}\mathbf{F}_{N}(\Omega;x)\mathbf{F}_{N}(\Omega;0)^{-1}. 
    \end{equation*} 
    By an argument similar to that in Lemma \ref{lem-critical-point}, we obtain that 
    $\mathrm{Ker}(\mathcal{J}\mathcal{L})$ is spanned by $\mathsf{E}_{point}$. 
    Taking the $\lambda$-derivative of both sides of \eqref{JL-spectral-problem} yields the 
    generalized eigenfunctions 
    \begin{equation*}
        \mathbf{S}_{i,\lambda}(\lambda_{k}),\mathbf{S}_{i,\lambda}(\lambda_{k}^{*}),
        \mathbf{S}_{-i,\lambda}(\lambda_{k}),\mathbf{S}_{-i,\lambda}(\lambda_{k}^{*})
    \end{equation*}
    for $i=1,2$. 
    For $x\to \infty$, we have
\begin{equation*}
    \mathbf{P}_{i,\lambda}(\lambda_{k}^{*};x)=
   \left( \mathcal{P}_{\lambda}(\lambda_{k}^{*})\mathbf{e}_{1}\left(
        \mathcal{P}(\lambda_{k})\mathbf{e}_{i+1}-\sum_{s,r=1}^{N}c_{ir}^{*}G_{+}^{rs}(\lambda_{k})\begin{pmatrix}
    0\\ \mathbf{c}_{s}
\end{pmatrix}
    \right)^{T}+o(1)\right) {\rm e}^{2b_{k}x}. 
\end{equation*}
Since
\begin{equation*}
    \mathcal{P}(\lambda_{k})\mathbf{e}_{i}-\sum_{s,r=1}^{N}c_{ir}^{*}G_{+}^{rs}(\lambda_{k})\mathbf{c}_{s}=
    \mathcal{P}(\lambda_{k})\left(
        \mathbf{e}_{i}-\sum_{r=1}^{N}\frac{2\ii b_{r}}{\lambda_{k}-\lambda_{r}^{*}}\begin{pmatrix}
            G_{i1}^{[r]} \\ G_{i2}^{[r]}
        \end{pmatrix}
    \right),
\end{equation*}
we introduce the matrix function
\begin{equation*}
    \mathbf{G}^{+}(\lambda)=\mathbb{I}_{2}-\sum_{r=1}^{N}\frac{2\ii b_{r}}{\lambda-\lambda_{r}^{*}}\mathbf{G}_{r}^{T}=
    \mathbb{I}_{2}-\sum_{r=1}^{N}\frac{2\ii b_{r}}{\lambda-\lambda_{r}^{*}}\mathbf{d}_{r}\mathbf{c}_{r}^{\dagger},
\end{equation*}
which admits one eigenfunction $\mathbf{c}_{k}$ at $\lambda=\lambda_{k}$. Then
\begin{equation*}
    \mathbf{P}_{i,\lambda}(\lambda_{k}^{*};x)=
   \left( \mathcal{P}_{\lambda}(\lambda_{k}^{*})\mathcal{P}(\lambda_{k})\mathbf{e}_{1}\begin{pmatrix}
    0 \\ \mathbf{G}_{i}^{+}(\lambda_{k})
   \end{pmatrix}^{\dagger}
   +o(1)\right) {\rm e}^{2 b_{k}x}
\end{equation*}
for $\mathbf{G}^{+}(\lambda)=(\mathbf{G}^{+}_{1}(\lambda),\mathbf{G}^{+}_{2}(\lambda))$. Since
\begin{equation*}
    \mathbf{G}^{+}(\lambda_{k})\mathbf{c}_{k}=\mathbf{G}^{+}_{1}(\lambda_{k})c_{1k}+\mathbf{G}^{+}_{2}(\lambda_{k})c_{2k}=0
\end{equation*}
and $\mathbf{G}^{+}(\lambda_{k})$ has rank $1$, the vectors $\mathbf{G}^{+}_{1}(\lambda_{k})$ and $\mathbf{G}^{+}_{2}(\lambda_{k})$
are not equal to $(0,0)^{T}$ when $c_{1k}$ and $c_{2k}$ are nonzero.  
Hence the squared eigenfunctions $\mathbf{S}_{i,\lambda}(\lambda_{k}^{*})$ exhibit exponential growth as $x\to+\infty$.
Since $\mathbf{S}_{i,\lambda}(\lambda_{i}^{*})$ and $\mathbf{S}_{-i,\lambda}(\lambda_{i})$ also grow exponentially at infinity, there are only $2N$ eigenfunctions in $\hat{\mathsf{E}}_{point}$.  
If $c_{2k}=0$, then $\mathbf{G}^{+}_{1}(\lambda_{k})=0$. 
Since $\mathbf{G}^{+}(\lambda_{k})$ has rank $1$, it follows that $\mathbf{G}^{+}_{2}(\lambda_{k})\ne0$. 
Hence, $\mathbf{S}_{2,\lambda}(\lambda_{k}^{*})$ still has exponential growth at $x=+\infty$.
Thus there are still $2N$ eigenfunctions in $\hat{\mathsf{E}}_{point}$. If $c_{1k}=0$, a similar argument applies to $\mathbf{S}_{2,\lambda}(\lambda_{k}^{*})$.  
The second derivatives of the squared eigenfunctions on the point spectrum are not in $L^{2}$ by \eqref{asy-DT-matrix}, so no further generalized eigenfunctions exist.  
This completes the proof.
\end{proof}
By \eqref{JL-spectral-problem}, we obtain
\begin{align*}
    \mathcal{J}\mathcal{L}\mathbf{S}_{i,\lambda}(\lambda_{k})
    =& 2^{2N}{\rm i}\mathcal{P}(\lambda_{k})\hat{\mathcal{P}}_{\lambda}(\lambda_{k})\mathbf{S}_{i}(\lambda_{k}), \\
    \mathcal{J}\mathcal{L}\mathbf{S}_{-i,\lambda}(\lambda_{k}^{*})
=&-2^{2N}{\rm i}\mathcal{P}_{\lambda}(\lambda_{k}^{*})\hat{\mathcal{P}}(\lambda_{k}^{*})\mathbf{S}_{-i}(\lambda_{k}^{*}). 
\end{align*} 
By Lemma \ref{JL-eigenvalue}, we conclude that
\begin{align*}
    \int_{\mathbb{R}}\mathbf{S}_{-i,\lambda}(\lambda_{k}^{*})^{\dagger}\mathcal{L}\mathbf{S}_{j,\lambda}(\lambda_{k})\mathrm{d}x
    =\left(\int_{\mathbb{R}}\mathbf{S}_{j,\lambda}(\lambda_{k})^{\dagger}\mathcal{L}\mathbf{S}_{-i,\lambda}(\lambda_{k}^{*})\mathrm{d}x\right)^{*}
    =c_{ik}c_{jk}A_{k}
\end{align*}
where
\begin{equation*}
    A_{k}=2^{2N}\ii \hat{\mathcal{P}}_{\lambda}(\lambda_{k})^{3}\mathcal{P}(\lambda_{k})^{3}.
\end{equation*}
We can now conclude the proof of part (a) in Theorem \ref{thm-spectrum-L-tildeL}.
\begin{proof}[Proof of part (a) in Theorem \ref{thm-spectrum-L-tildeL}]
    The essential spectrum of $\mathcal{L}$ follows directly from Weyl’s essential spectrum theorem.
The kernel of $\mathcal{L}$ is spanned by the squared eigenfunctions in $\mathsf{E}_{\mathrm{point}}^{+}$, since $\mathcal{J}$ is invertible, or equivalently, by the derivatives of the scattering parameters:
    \begin{equation*}
        \mathrm{Ker}(\mathcal{L})=\mathrm{span}\left\{
            \mathcal{C}\mathbf{q}^{[N]}_{ij},
            \mathcal{C}\ii \mathbf{q}^{[N]}_{ij}: i=1,2 ,\ j=1,2,\cdots,N
        \right\}
    \end{equation*}
    where 
    \begin{equation*}
        \mathbf{q}^{[N]}_{ij}=\begin{pmatrix}
            \partial_{c_{ij}}\mathbf{q}^{[N]} \\
            (\partial_{c_{ij}}\mathbf{q}^{[N]})^{*}
        \end{pmatrix}.
    \end{equation*}
    Here $\mathbf{q}^{[N]}_{ij}$ arises from differentiating the ODE \eqref{ODE-Nsoliton} satisfied by the $N$-soliton solution, which yields
    $\mathcal{L}\mathbf{q}^{[N]}_{ij}=0$. 
    Thus, it suffices to analyze the point spectrum.

    Without loss of generality, we consider the case $\Gamma = \emptyset$ in \eqref{Gamma-set}.
    From \eqref{ort-ess}, the essential spectrum part $\mathbb{E}_{\mathrm{ess}}^{\mathrm{X}}$ does 
    not contribute to the negative direction of $\mathcal{L}$. As an example, 
    \begin{align*}
        &\left(\mathcal{L}\int_{\mathbb{R}} w_{i}(\lambda)\mathbf{S}_{i}(\lambda;x)\mathrm{d}\lambda,\int_{\mathbb{R}} w_{j}(\lambda')\mathbf{S}_{j}(\lambda';x)\mathrm{d}\lambda'\right)
        \\=&
        \left(-2^{2N}\ii|\mathcal{P}(\lambda)|^{2}\mathcal{J}\int_{\mathbb{R}} w_{i}(\lambda)\mathbf{S}_{i}(\lambda;x)\mathrm{d}\lambda,\int_{\mathbb{R}} w_{j}(\lambda')\mathbf{S}_{j}(\lambda';x)\mathrm{d}\lambda'\right)
        \\=&
        \mathrm{Re}\int_{\mathbb{R}^{2}}  -2^{2N}\ii|\mathcal{P}(\lambda)|^{2} w_{i}^{*}(\lambda) w_{j}(\lambda') \int_{\mathbb{R}}\mathbf{S}_{i}^{\dagger}(\lambda;x)\mathcal{J}\mathbf{S}_{j}(\lambda';x)\mathrm{d}x\mathrm{d}\lambda\mathrm{d}\lambda'
        \\=&
        2^{2N}\pi \delta_{ij}\int_{\mathbb{R}}|\mathcal{P}(\lambda)|^{6}|w_{i}(\lambda)|^{2}\mathrm{d}\lambda \geq 0. 
    \end{align*}
    Hence, it suffices to examine the quadratic form $(\mathcal{L}\cdot, \cdot)$ on the subspace $\mathbb{E}_{\mathrm{point}}^{\mathrm{X}}$.
    Define the negative cone by
    \begin{equation*}
        \mathcal{N}=\{\mathbf{f}\in \mathrm{X}:(\mathcal{L}\mathbf{f},\mathbf{f})<0\}.
    \end{equation*}
    The number of negative eigenvalues of $\mathcal{L}$ is then equal to $\dim(\mathcal{N})$, which is the dimension of the maximal subspace contained in $\mathcal{N} $\cite{kapitula_counting_2004,kapitula_spectral_2013}.
    
    Since $(\mathcal{L}\mathbf{f},\mathbf{g})=(0,\mathbf{g})=0$ whenever
    $\mathbf{f}\in\mathrm{span}\{\mathcal{C}\mathsf{E}_{point}^{+}\cup \mathcal{C}\ii\mathsf{E}_{point}^{+}\}$, 
    it suffices to consider the quadratic form on the space $\mathrm{span}\{\mathcal{C}\hat{\mathsf{E}}_{point}^{+}\cup \mathcal{C}\ii\hat{\mathsf{E}}_{point}^{+}\}$. 
    By the symmetry property \eqref{symmetric-squared-eigenfunctions} of the squared eigenfunctions, we have
    \begin{align*}
        \mathcal{C}\mathbf{S}_{1,\lambda}(\lambda_{k})=&\mathbf{S}_{1,\lambda}(\lambda_{k})-\mathbf{S}_{-1,\lambda}(\lambda_{k}^{*}), \\
        \mathcal{C}\ii\mathbf{S}_{1,\lambda}(\lambda_{k})=&\ii(\mathbf{S}_{1,\lambda}(\lambda_{k})+\mathbf{S}_{-1,\lambda}(\lambda_{k}^{*})).
    \end{align*}
    Hence the matrix representation of the quadratic form is block diagonal:
    \begin{equation}\label{X-negative-matrix}
        (\mathcal{L}\mathbf{f},\mathbf{g})=\mathrm{diag}(\mathbf{A}_{1},\mathbf{A}_{2},\cdots,\mathbf{A}_{N})
    \end{equation}
    where
\begin{equation*}
    \mathbf{f},\mathbf{g}\in \{
        \mathcal{C}\mathbf{S}_{1,\lambda}(\lambda_{k}), \mathcal{C}\ii\mathbf{S}_{1,\lambda}(\lambda_{k}):k=1,2,\cdots, N
    \}
\end{equation*}
    and
    \begin{equation*}
        \mathbf{A}_{k}=2\begin{pmatrix}
            -\mathrm{Re}c_{1k}^{2}A_{k} & \mathrm{Im}c_{1k}^{2}A_{k} \\
            \mathrm{Im}c_{1k}^{2}A_{k} & \mathrm{Re}c_{1k}^{2}A_{k}
        \end{pmatrix}. 
    \end{equation*}
    The matrix \eqref{X-negative-matrix} admits $N$ positive eigenvalues $|c_{1k}^{2}A_{k}|,k=1,2,\cdots,N$ and 
    $N$ negative eigenvalues $-|c_{1k}^{2}A_{k}|,k=1,2,\cdots,N$, respectively. 
\end{proof}
\begin{rem}
    Here we present the example of the $1$-soliton solution.
    Due to the Galilean transformation, we set $a_{1}=0$. Then 
    the corresponding operator is the Schrödinger operator
    \begin{equation*}
        \mathcal{L}(\mathbf{q}^{[1]})=-\partial_{xx}+
        4b_{1}^{2}-2|q^{[1]}|^{2}-2(q^{[1]})^{2}{\rm e}^{2\ii \theta_{1}}\mathbf{v}_{\alpha}^{\theta_{1}-\theta_{2}}\otimes 
        \mathbf{v}_{\alpha}^{\theta_{1}-\theta_{2}}{\cdot}^{*}
        -2|q^{[1]}|^{2}\mathbf{v}_{\alpha}^{\theta_{1}-\theta_{2}}\otimes 
        \mathbf{v}_{\alpha}^{\theta_{2}-\theta_{1}},
    \end{equation*}
    where $q^{[1]}$ is given in \eqref{NLS-1-soliton} with $\theta=a_{1}=0$. By Theorem \ref{thm-spectrum-L-tildeL}, 
$\mathcal{L}(\mathbf{q}^{[1]})$ admits one negative eigenvalue.
\end{rem}

After determining the number of negative eigenvalues of the operator $\mathcal{L}$, the nonlinear stability of the $N$-soliton 
can be established by exploiting the coercivity of $\mathcal{L}$ under suitable conditions in 
a neighborhood of the manifold consisting of $N$-soliton solutions. 

\subsection{The reduced Hamiltonian}
The negative direction of the operator $\mathcal{L}$ is associated with the matrix
\begin{equation}\label{matrix-H}
    \mathbf{H}= (H_{\sigma\tau})
\end{equation}
with entries given by the coefficients of the polynomial \eqref{P-hatP} and the conserved quantities:
\begin{equation}\label{mat-H-element}
 H_{\sigma\tau}=\partial_{\sigma \tau}\mathcal{I}-\sum_{n=0}^{2N}\partial_{\sigma \tau}(\mu_{n})\mathcal{H}_{n}
\end{equation}
for $\sigma,\tau\in\{a_{k},b_{k}:k=1,2,\cdots,N\}$. The following lemma holds:
\begin{lem}
    Let $\mathbf{H}$ be the matrix defined in \eqref{matrix-H} and \eqref{mat-H-element}. 
    If the spectral parameters are pairwise distinct, then $\mathbf{H}$ is nondegenerate. 
    Moreover, $\mathbf{H}/2^{2N+2}$ has $N$ positive eigenvalues $b_{k}|J_{k}|$ and $N$ negative eigenvalues $-b_{k}|J_{k}|$ for $k=1,2,\dots, N$, 
    where
    \begin{equation*}
        J_{k}=\left.\frac{\mathcal{P}(\lambda)\hat{\mathcal{P}}(\lambda)}{(\lambda-\lambda_{k})(\lambda-\lambda_{k}^{*})}\right|_{\lambda=\lambda_{k}}.
    \end{equation*}
    \end{lem}    
\begin{proof}
 Since 
\begin{align*}
    \partial_{\sigma\tau}\mathcal{I}=&\partial_{\sigma}\left(
        \sum_{n=0}^{2N}\partial_{\tau}(\mu_{n})\mathcal{H}_{n}+\left(\sum_{n=0}^{2N}\mu_{n}\frac{\delta \mathcal{H}_{n}}{\delta \mathbf{q}},\partial_{\tau}\mathbf{q}\right)
 \right)\\
 =&\partial_{\sigma}\left(
        \sum_{n=0}^{2N}\partial_{\tau}(\mu_{n})H_{n}
 \right)\\
 =&\sum_{n=0}^{2N}\partial_{\sigma\tau}(\mu_{n})\mathcal{H}_{n}+\sum_{n=0}^{2N}\partial_{\tau}(\mu_{n})\partial_{\sigma}(\mathcal{H}_{n}),
\end{align*}
the element in \eqref{mat-H-element} can be represented as
\begin{equation*}
 H_{\sigma\tau}=\sum_{n=0}^{2N}\partial_{\tau}(\mu_{n})\partial_{\sigma}(\mathcal{H}_{n}). 
\end{equation*}
Hence, we need to consider the derivative of \eqref{P-hatP}. 
Since
\begin{align*}
	\mathcal{P}_{a_{k}}(\lambda)={\rm i}\mathcal{P}_{b_{k}}(\lambda)=-\frac{\mathcal{P}(\lambda)}{\lambda-\lambda_{k}^{*}}, \\
	\hat{\mathcal{P}}_{a_{k}}(\lambda)=-{\rm i}\hat{\mathcal{P}}_{b_{k}}(\lambda)=-\frac{\hat{\mathcal{P}}(\lambda)}{\lambda-\lambda_{k}},
\end{align*}
we obtain 
\begin{align*}
		\partial_{a_{k}}(\mathcal{P}(\lambda)\hat{\mathcal{P}}(\lambda))=-\mathcal{P}(\lambda)\hat{\mathcal{P}}(\lambda)(\frac{1}{\lambda-\lambda_{k}^{*}}+\frac{1}{\lambda-\lambda_{k}}), \\
		\partial_{b_{k}}(\mathcal{P}(\lambda)\hat{\mathcal{P}}(\lambda))={\rm i}\mathcal{P}(\lambda)\hat{\mathcal{P}}(\lambda)(\frac{1}{\lambda-\lambda_{k}^{*}}-\frac{1}{\lambda-\lambda_{k}}),
	\end{align*}
hence 
\begin{align*}
	\partial_{a_{k}}(\mathcal{P}(\lambda)\hat{\mathcal{P}}(\lambda))=
	&\frac{-2(\lambda-a_{k})\mathcal{P}(\lambda)\hat{\mathcal{P}}(\lambda)}{(\lambda-\lambda_{k})(\lambda-\lambda_{k}^{*})}, \\
	\partial_{b_{k}}(\mathcal{P}(\lambda)\hat{\mathcal{P}}(\lambda))=
	&\frac{2b_{k}\mathcal{P}(\lambda)\hat{\mathcal{P}}(\lambda)}{(\lambda-\lambda_{k})(\lambda-\lambda_{k}^{*})}.
\end{align*}
In addition, we have 
\begin{align*}
	\partial_{a_{k}}(\mathcal{H}_{n})=2^{n+1}\mathrm{Im}(\lambda_{k}^{n}),\\
	\partial_{b_{k}}(\mathcal{H}_{n})=2^{n+1}\mathrm{Re}(\lambda_{k}^{n}). 
\end{align*}
It leads to
\begin{align*}
	H_{a_{k}a_{l}}=&\sum_{n=0}^{2N}\partial_{a_{k}}(\mu_{n})2^{n+1}\mathrm{Im}(\lambda_{l}^{n})
	=\mathrm{Im}(2^{2N+1}\partial_{a_{k}}(\mathcal{P}\hat{\mathcal{P}})(\lambda_{l}))
	=-2^{2N+2}b_{k}\delta_{kl}\mathrm{Re}(J_{k}),
\end{align*}
and 
\begin{align*}
	H_{b_{k}b_{l}}=&\sum_{n=0}^{2N}\partial_{b_{k}}(\mu_{n})2^{n+1}\mathrm{Re}(\lambda_{l}^{n})
	=\mathrm{Re}(2^{2N+1}\partial_{b_{k}}(\mathcal{P}\hat{\mathcal{P}})(\lambda_{l}))
	=2^{2N+2}b_{k}\delta_{kl}\mathrm{Re}(J_{k}).
\end{align*}
Moreover, 
\begin{align*}
	H_{a_{k}b_{l}}=&\sum_{n=0}^{2N}\partial_{a_{k}}(\mu_{n})2^{n+1}\mathrm{Re}(\lambda_{l}^{n})
	=\mathrm{Re}(2^{2N+1}\partial_{a_{k}}(\mathcal{P}\hat{\mathcal{P}})(\lambda_{l}))
	=2^{2N+2}b_{k}\delta_{kl}\mathrm{Im}(J_{k}).
\end{align*}
Denote 
\begin{equation*}
    \mathbf{J}=\mathrm{diag}(b_{1}J_{1},b_{2}J_{2},\cdots,b_{N}J_{N}), 
\end{equation*}
then the matrix $\mathbf{H}$ admits the representation
\begin{equation*}
    \mathbf{H}=2^{2N+2}\begin{pmatrix}
 -\mathrm{Re}(\mathbf{J}) & \mathrm{Im}(\mathbf{J}) \\
        \mathrm{Im}(\mathbf{J}) & \mathrm{Re}(\mathbf{J})
    \end{pmatrix}. 
\end{equation*}
The characteristic polynomial of $\mathbf{H}$ becomes 
\begin{align*}
    \det(\eta \mathbb{I}_{2N}-\mathbf{H})
 =2^{2N(2N+2)}\prod_{k=1}^{N}\left(\lambda^{2}-\mathrm{Re}(b_{k}J_{k})^{2}-\mathrm{Im}(b_{k}J_{k})^{2}\right)
\end{align*}
and the roots are $\pm b_{k}|J_{k}|$ for $k=1,2,\cdots, N$. This completes the proof.
\end{proof}
Denote by $\mathrm{n}(\cdot)$, $\mathrm{z}(\cdot)$, and $\mathrm{p}(\cdot)$ the numbers of negative, 
zero, and positive eigenvalues of $\cdot$, respectively, where $\cdot$ denotes a matrix or an operator.
The number of positive eigenvalues for $\mathbf{H}$ is given by
\begin{equation}\label{pos-H}
    \mathrm{p}(\mathbf{H})=N. 
\end{equation}
Define the auxiliary quantities
\begin{align}\label{con-aux}
    \mathcal{Q}_{a_{k}}=\sum_{n=0}^{2N} (\partial_{a_{k}}\mu_{n}) \mathcal{H}_{n}, \quad 
    \mathcal{Q}_{b_{k}}=\sum_{n=0}^{2N} (\partial_{b_{k}}\mu_{n}) \mathcal{H}_{n}
\end{align}
which are independent of time. Let $\mathcal{P}$ be the projection of $\mathrm{X}$ onto 
\begin{equation*}
    \mathrm{X}_{1}=\mathrm{span}\left\{
        \frac{\delta \mathcal{Q}_{a_{k}}}{\delta \mathbf{q}},\frac{\delta \mathcal{Q}_{b_{k}}}{\delta \mathbf{q}}:
 k=1,2,\cdots, N
 \right\}^{\perp},
\end{equation*}
then we have the following lemma 
\begin{lem}[\cite{grillakis_stability_1990,ling2024stability}]\label{spectrum-reduce}
    For $\sigma, \tau \in \{a_{1},b_{1},a_{2},b_{2},\dots,a_{N},b_{N}\}$,
    the element \eqref{mat-H-element} can be expressed as 
    \begin{equation}
 H_{\sigma \tau}=-(\mathcal{L}\partial_{\sigma}\mathbf{q},\partial_{\tau}\mathbf{q}).
    \end{equation}
    The orthogonal complement of $\mathrm{X}_{1}$ is isomorphic to
    \begin{equation*}
        \mathrm{Y}_{1}=\mathrm{span}\left\{
            \partial_{a_{k}}\mathbf{q},\partial_{b_{k}}\mathbf{q}:k=1,2,\cdots,N
 \right\}
    \end{equation*}
 and the operator $\mathcal{L}:\mathrm{Y}_{1}\to\mathrm{X}_{1}^{\perp}$ is invertible, with
    \begin{equation}
        \mathcal{L}\partial_{\sigma}\mathbf{q}=-\frac{\delta \mathcal{Q}_{\sigma}}{\delta \mathbf{q}}. 
    \end{equation}
    Moreover, for $\mathbf{f} \in \mathrm{X}_{1}$ and $\mathbf{g} \in \mathrm{Y}_{1}$,
    the inner product satisfies
    \begin{equation}
 (\mathbf{f},\mathcal{L}\mathbf{g})=0.
    \end{equation}
    The following identities hold:
    \begin{align*}
        \mathrm{n}(\mathcal{L}\mathcal{P})=&\mathrm{n}(\mathcal{L})-\mathrm{p}(\mathbf{H}), \\
        \mathrm{z}(\mathcal{L}\mathcal{P})=&\mathrm{z}(\mathcal{L}).
    \end{align*}
\end{lem}
An immediate corollary of Lemma \ref{spectrum-reduce} is that the operator $\mathcal{L}$ is coercive in
\begin{equation*}
    \mathcal{R}'(\mathbf{q})=\mathrm{Ker}(\mathcal{L})^{\perp}\cap \mathrm{X}_{1}
\end{equation*}
with respect to the $L^{2}$ norm, as follows from \eqref{pos-H} and part (a) of Theorem \ref{thm-spectrum-L-tildeL}.
Since the functions in $\mathrm{X}_{1}$ are the variation of the conserved 
quantities \eqref{con-aux}, we define
\begin{equation*}
    \mathcal{R}(\mathbf{q})=\mathrm{Ker}(\mathcal{L})^{\perp}\cap \mathrm{span}\left\{\mathbf{z}:
        \mathcal{Q}_{\sigma}(\mathbf{q})=\mathcal{Q}_{\sigma}(\mathbf{q}+\mathbf{z}),\sigma=a_{1},b_{1},
 a_{2},b_{2},\cdots,a_{N},b_{N}
 \right\}
\end{equation*}
which connects $\mathcal{R}'(\mathbf{q})$ to the spectral parameters.
Moreover, in $\mathcal{R}'(\mathbf{q})$, the operator $\mathcal{L}$ possesses $H^{N}$-coercivity rather than $L^{2}$-coercivity.
\begin{lem}\label{core-R}
    Let $\mathbf{z}\in H^{N}(\mathbb{R},\mathrm{X})\cap \mathcal{R}'(\mathbf{q})$. 
    Then the operator $\mathcal{L}$ is $H^{N}$-coercive in the sense that
        \begin{equation}\label{corecivity-R'}
     (\mathcal{L} \mathbf{z},\mathbf{z})\geq C \|\mathbf{z}\|_{H^{N}}^{2}
        \end{equation}
    for some positive constant $C$.  
    
    Moreover, if $\mathbf{z}\in H^{N}(\mathbb{R},\mathrm{X})\cap \mathcal{R}(\mathbf{q})$ 
    and $\|\mathbf{z}\|_{H^{N}}$ is sufficiently small, then
        \begin{equation}\label{corecivity-R}
     (\mathcal{L} \mathbf{z},\mathbf{z})\geq C_{1} \|\mathbf{z}\|_{H^{N}}^{2}-C_{2}\|\mathbf{z}\|_{H^{N}}^{3}
        \end{equation}
    for some positive constants $C_{1}$ and $C_{2}$.
    \end{lem}
    
\begin{proof}
 Now we prove the inequality \eqref{corecivity-R'} firstly and $\mathbf{z}\in H^{N}(\mathbb{R},\mathrm{X})\cap \mathcal{R}'(\mathbf{q})$. 
 By \eqref{pos-H}, part (a) of Theorem \ref{thm-spectrum-L-tildeL} and Lemma \ref{spectrum-reduce}, the operator have no zero eigenvalue and 
 negative eigenvalue in $H^{N}(\mathbb{R},\mathrm{X})\cap \mathcal{R}'(\mathbf{q})$, the $L^{2}$-corecivity 
 hold. 

 We first prove inequality \eqref{corecivity-R'} for $\mathbf{z}\in H^{N}(\mathbb{R},\mathrm{X})\cap \mathcal{R}'(\mathbf{q})$.
By \eqref{pos-H}, part (a) of Theorem \ref{thm-spectrum-L-tildeL}, and Lemma \ref{spectrum-reduce}, 
the operator has no zero or negative eigenvalues in $H^{N}(\mathbb{R},\mathrm{X})\cap \mathcal{R}'(\mathbf{q})$.
Hence the $L^{2}$-corecivity holds
    \begin{equation}\label{corecivity-L2}
 (\mathcal{L} \mathbf{z},\mathbf{z})\geq C \|\mathbf{z}\|_{L^{2}}^{2}. 
    \end{equation}
    Suppose that $\mathcal{L}$ is not $H^{N}$-coercive.
    Since $\mathcal{L}$ is positive by \eqref{corecivity-L2}, there exists a bounded sequence $\{\mathbf{z}_{n}\}_{n=1}^{\infty}\subset H^{N}$ such that 
    $(\mathcal{L}\mathbf{z}_{n},\mathbf{z}_{n})\to 0$ as $n\to +\infty$. 
    Without loss of generality, assume
    $\|\mathbf{z}_{n}\|_{H^{N}}=1$ for all $n$. Then the $L^{2},\dot{H}^{1},\cdots, \dot{H}^{N-1}$ norms for $\mathbf{z}_{n}$ 
 tend to zero by $L^{2}$-corecivity \eqref{corecivity-L2} and the induction argument to the inequality
    \begin{equation*}
        \|\partial_{x}^{j} \mathbf{z}_{n}\|_{L^{2}}=(-1)^{j}(\partial_{x}^{2j}\mathbf{z}_{n},\mathbf{z}_{n})\leq 
        \|\partial_{x}^{2j}\mathbf{z}_{n}\|_{L^{2}}^{1/2}\|\mathbf{z}_{n}\|_{L^{2}}^{1/2}\leq 
        \|\mathbf{z}_{n}\|_{L^{2}}^{1/2}
    \end{equation*}
 if $j\leq N/2$ and the inequality 
    \begin{equation*}
        \|\partial_{x}^{j} \mathbf{z}_{n}\|_{L^{2}}=(-1)^{N-j}(\partial_{x}^{N}\mathbf{z}_{n},\partial_{x}^{2j-N}\mathbf{z}_{n})\leq 
        \|\partial_{x}^{N}\mathbf{z}_{n}\|_{L^{2}}^{1/2}\|\partial_{x}^{2j-N}\mathbf{z}_{n}\|_{L^{2}}^{1/2}\leq 
        \|\partial_{x}^{2j-N}\mathbf{z}_{n}\|_{L^{2}}^{1/2}
    \end{equation*}
 if $N/2<j\leq N$. Since the $H^{N}$ norm for $\mathbf{z}_{n}$ is $1$, it follows that the limit for the $\dot{H}^{N}$ norm is $1$
    \begin{equation}\label{lim-zn}
        \lim_{n\to \infty}\|\partial_{x}^{N} \mathbf{z}_{n}\|_{L^{2}}=1. 
    \end{equation}
 In addition, we can rewrite
    \begin{equation}\label{Lznzn}
 (\mathcal{L}\mathbf{z}_{n},\mathbf{z}_{n})=\|\partial_{x}^{N} \mathbf{z}_{n}\|_{L^{2}}+
        \sum_{i=0}^{N}\sum_{j=0}^{i-1}(\partial_{x}^{i}\mathbf{z}_{n}, \mathbf{f}_{ij} \partial_{x}^{j}\mathbf{z}_{n})
    \end{equation}
 where $\mathbf{f}_{ij}$ is the polynomial with respect to potential $\mathbf{q}$ and its derivatives and 
 hence $\mathbf{f}_{ij}\in \mathcal{S}(\mathbb{R},\mathrm{X})$. 
 Then the $\dot{H}^{N}$ norm for $\mathbf{z}_{n}$
 vanishes as $n\to \infty$, which contradicts \eqref{lim-zn} since 
    \begin{equation*}
        \|\partial_{x}^{N} \mathbf{z}_{n}\|_{L^{2}}\leq (\mathcal{L}\mathbf{z}_{n},\mathbf{z}_{n})+
        \sum_{0\leq j < i\leq N}\|\mathbf{f}_{ij}\|_{L^{\infty}}
        \|\partial_{x}^{i}\mathbf{z}_{n}\|_{L^{2}}^{1/2}\|\partial_{x}^{j}\mathbf{z}_{n}\|_{L^{2}}^{1/2} \to 0
    \end{equation*}
 by \eqref{Lznzn}. 
 
 Now we turn to proving \eqref{corecivity-R} and assume that 
 $\mathbf{z}\in H^{N}(\mathbb{R},\mathrm{X})\cap \mathcal{R}(\mathbf{q})$ with sufficiently small $\|\mathbf{z}\|_{H^{N}}$. 
 First, we decompose $\mathbf{z}$ along $\mathcal{R}'(\mathbf{q})\oplus (\mathrm{Ker}(\mathcal{L})+ \mathrm{X}_{1}^{\perp})$:
\begin{equation}\label{decom-z}
    \mathbf{z}=\mathbf{z}_{1}
    +\sum_{i=1}^{2}\sum_{j=1}^{N}\alpha_{ij} \partial_{c_{ij}}\mathbf{q}
    +\sum_{\sigma\in\{a_{i},b_{i}:1\leq i\leq N\}}\beta_{\sigma}\frac{\delta \mathcal{Q}_{\sigma}}{\delta \mathbf{q}}.
\end{equation}
Expanding $\mathcal{Q}_{\sigma}$ around $\mathbf{q}$ with a perturbation $\mathbf{z}$, we have
\begin{equation*}
    \mathcal{Q}_{\sigma}(\mathbf{q}+\mathbf{z})=\mathcal{Q}_{\sigma}(\mathbf{q})+
    \left(\frac{\delta \mathcal{Q}_{\sigma}}{\delta \mathbf{q}},\mathbf{z}\right)
    +\mathcal{O}(\|\mathbf{z}\|_{H^{N}}^{2}),
\end{equation*}
which yields $2N$ equations:
\begin{equation}\label{coeff-lin-equ-1}
    \sum_{i=1}^{2}\sum_{j=1}^{N}\alpha_{ij}\left(
        \frac{\delta \mathcal{Q}_{\tau}}{\delta \mathbf{q}},\partial_{c_{ij}}\mathbf{q}
    \right)
    +\sum_{\sigma}\beta_{\sigma}
    \left(
        \frac{\delta \mathcal{Q}_{\tau}}{\delta \mathbf{q}},\frac{\delta \mathcal{Q}_{\sigma}}{\delta \mathbf{q}}
    \right)
    =\mathcal{O}(\|\mathbf{z}\|_{H^{N}}^{2}),
    \quad \tau\in \{a_{i},b_{i}:1\leq i\leq N\},
\end{equation}
since $\mathbf{z}\in \mathcal{R}(\mathbf{q})$.
In addition, since $\mathcal{R}(\mathbf{q})$ and 
$\mathcal{R}'(\mathbf{q})$ are orthogonal to $\mathrm{Ker}(\mathcal{L})$, we obtain $2N$ equations:
\begin{equation}\label{coeff-lin-equ-2}
    \sum_{i=1}^{2}\sum_{j=1}^{N}\alpha_{ij}\left(
        \partial_{c_{kl}}\mathbf{q},\partial_{c_{ij}}\mathbf{q}
    \right)
    +\sum_{\sigma}\beta_{\sigma}
    \left(
        \partial_{c_{kl}}\mathbf{q},\frac{\delta \mathcal{Q}_{\tau}}{\delta \mathbf{q}}
    \right)
    =0,\quad k=1,2,\ l=1,2,\dots,N.
\end{equation}
Solving \eqref{coeff-lin-equ-1} and \eqref{coeff-lin-equ-2} with respect to $\alpha_{\sigma}$ and $\beta_{ij}$, 
the coefficients are given by
\begin{equation}\label{ord-coeff}
    \alpha_{ij}=\mathcal{O}(\|\mathbf{z}\|_{H^{N}}^{2}),\quad 
    \beta_{\sigma}=\mathcal{O}(\|\mathbf{z}\|_{H^{N}}^{2}),
\end{equation}
since the coefficient matrix is the Gram matrix in $\mathcal{R}'(\mathbf{q})$ and is of order $\mathcal{O}(1)$. 
Combining \eqref{decom-z} and \eqref{ord-coeff}, we complete the proof of \eqref{corecivity-R}.
\end{proof}

\subsection{The proof for nonlinear stability for CNLS equations}
In this subsection, we complete the proof of nonlinear stability for the CNLS equations. 
It remains to introduce a modulation argument to connect the kernel of $\mathcal{L}$ with the perturbation 
of the $N$-soliton solutions. For clarity of notation, we define the neighborhood of the $N$-soliton solutions by
\begin{equation*}
    B_{\delta}=\{\mathbf{f}:\|\mathbf{f}-\mathbf{q}^{[N]}(x,t;c_{ij})\|_{H^{N}}\leq \delta\},
\end{equation*}
where $\delta>0$ is sufficiently small. 
\begin{lem}\label{lem-modulation}
    For the $N$-soliton solution $\mathbf{q}^{[N]}(x,t;c_{ij})$, if $\mathbf{u}(x)\in B_{\delta}$,   
    then there exist parameters
    \begin{equation*}
        \tilde{c}_{ij} \in \mathbb{C},
    \end{equation*}
    such that the perturbation $\mathbf{w}(x)=\mathbf{u}(x)-\mathbf{q}^{[N]}(x,t;\tilde{c}_{ij})$
    satisfies
    \begin{equation}\label{orth-ker}
        \mathbf{w}(x)\in \mathrm{Ker}(\mathcal{L})^{\perp}. 
    \end{equation}
    Moreover, if $\mathbf{u}(x,t)\in B_{\delta}$ is an $H^{N}$ solution of the CNLS equations, 
    then there exist time-dependent parameters $c_{ij}(t)$ such that the perturbation 
    $\mathbf{w}(x,t)=\mathbf{u}(x,t)-\mathbf{q}^{[N]}(x,t;c_{ij}(t))$ 
    lies in $\mathrm{Ker}(\mathcal{L})^{\perp}$, and the time derivatives of $c_{ij}(t)$ satisfy
    \begin{equation}\label{der-con}
        \sum_{i,j} |\partial_{t}c_{ij}(t)|\leq C \|\mathbf{w}(t)\|_{L^{2}},
    \end{equation}
    for some positive constant $C$. 
\end{lem}

\begin{proof}
    The equation \eqref{orth-ker} follows immediately from the implicit function theorem. 
    Then, for any $t$ with $\mathbf{u}(x,t)\in B_{\delta}$, there exist $c_{ij}(t)$ such that 
    $\mathbf{w}(x,t)\in\mathrm{Ker}(\mathcal{L})^{\perp}$. 
    Differentiating
    \begin{equation*}
        \left(\mathbf{w}(x,t),\partial_{c_{kl}}\mathbf{q}^{[N]}\right)=0
    \end{equation*}
    with respect to $t$, we obtain 
    \begin{equation}\label{der-cij-1}
        \left(\partial_{t} \mathbf{w}-\sum_{i,j}\partial_{t}c_{ij}\,\partial_{c_{ij}}\mathbf{q}^{[N]},\partial_{c_{kl}}\mathbf{q}^{[N]}\right)
        +\left(\mathbf{w},\partial_{t}\partial_{c_{kl}}\mathbf{q}^{[N]}\right)
        +\left(\mathbf{w},\sum_{i,j}\partial_{t}c_{ij}\,\partial_{c_{ij}}\partial_{c_{kl}}\mathbf{q}^{[N]}\right)=0.
    \end{equation}
    Since $\mathbf{u}(x,t)$ and $\mathbf{q}^{[N]}(x,t)$ are both solutions of the CNLS equations \eqref{CNLS}, 
    the time derivative of the perturbation $\mathbf{w}$ can be expressed in terms of spatial derivatives of 
    $\mathbf{u}$ and $\mathbf{q}^{[N]}$:
    \begin{align*}
        \left| \left(\partial_{t}\mathbf{w},\partial_{c_{kl}}\mathbf{q}^{[N]}\right) \right|
        \leq &\ \left| \left(
            \partial_{x}^{2}\mathbf{w},\partial_{c_{kl}}\mathbf{q}^{[N]}
        \right) \right|+
        \left| \left(
            |\mathbf{u}|^{2} \mathbf{w}
            +\mathbf{q}^{[N]}\big(\mathbf{u}^{\dagger}\mathbf{w}+\mathbf{w}^{\dagger}\mathbf{q}^{[N]}\big),
            \partial_{c_{kl}}\mathbf{q}^{[N]}
        \right) \right| \\
        \leq &\ \big(\|\partial_{x}^{2}\partial_{c_{kl}}\mathbf{q}^{[N]}\|_{L^{2}}
        +\|\partial_{c_{kl}}\mathbf{q}^{[N]}\|_{L^{\infty}}\|\mathbf{u}\|_{H^{1}}^{2}
        +\|\partial_{c_{kl}}\mathbf{q}^{[N]}\|_{L^{\infty}}\|\mathbf{q}^{[N]}\|_{H^{1}}^{2}\big)
        \|\mathbf{w}\|_{L^{2}}.
    \end{align*}
    Moreover, we have
    \begin{equation*}
        \left|\left(\mathbf{w},\partial_{t}\partial_{c_{kl}}\mathbf{q}^{[N]}\right)\right|
        \leq C\|\partial_{c_{kl}}\mathbf{q}^{[N]}\|_{H^{2}}
        \|\mathbf{w}\|_{L^{2}}. 
    \end{equation*}
    Hence, the time derivatives of the scattering parameters satisfy the linear system
    \begin{equation*}
        \sum_{i,j}\partial_{t}c_{ij}
        \big((\partial_{c_{ij}}\mathbf{q}^{[N]},\partial_{c_{kl}}\mathbf{q}^{[N]})
        +(\mathbf{w},\partial_{c_{kl}}\mathbf{q}^{[N]})\big)
        =\mathcal{O}(\|\mathbf{w}\|_{L^{2}}),
    \end{equation*}
    by \eqref{der-cij-1}. The coefficient matrix of this system is non-degenerate since the scattering 
    parameters $\mathbf{c}(t)$ remain close to $\mathbf{c}(0)$. 
    This completes the proof of \eqref{der-con}.
\end{proof}

To prove the nonlinear stability of $N$-soliton solutions for the CNLS equations, we proceed by contradiction. 
\begin{proof}[Proof of stability for $N$-soliton solutions in Theorem \ref{thm-stability}]
Let $\mathbf{u}_{n}(\cdot,t) \in H^{N}(\mathbb{R},\mathbb{C}^{2})$ for any time $t$ and any natural number $n$, 
with initial condition $\mathbf{u}_{n}(\cdot,0)$.  
Assume that there exists $\epsilon_{0}$ such that  
\begin{equation*}
    \|\mathbf{u}_{n}(\cdot,0)-\mathbf{q}^{[N]}(\cdot,0;\mathbf{\Lambda},\mathbf{c}(0))\|_{H^{N}}\leq \frac{1}{n}
\end{equation*}
but  
\begin{equation*}
    \|\mathbf{u}_{n}(\cdot,t_{n})-\mathbf{q}^{[N]}(\cdot,0;\mathbf{\Lambda},\mathbf{c}(t_{n}))\|_{H^{N}}=
    \epsilon_{0}
\end{equation*}
for some sequence $\{t_{n}\}$ and a $C^{1}$ function $\mathbf{c}(t)$.  
Since the Lyapunov functional is continuous, it follows that  
\begin{equation*}
    |\mathcal{I}(\mathbf{u}_{n}(x,t_{n}))-\mathcal{I}(\mathbf{q}^{[N]})|
    =|\mathcal{I}(\mathbf{u}_{n}(x,0))-\mathcal{I}(\mathbf{q}^{[N]})|
    \leq \frac{C}{n}
\end{equation*}
for some constant $C$.  

In addition, there exists a sequence $\mathbf{v}_{n}(x)$ such that  
\begin{equation*}
    \|\mathbf{v}_{n}(\cdot)-\mathbf{u}_{n}(\cdot,t_{n})\|_{H^{N}} \to 0 
\end{equation*}
as $n \to \infty$, with $\mathcal{Q}_{\sigma}(\mathbf{v}_{n})=\mathcal{Q}_{\sigma}(\mathbf{q}^{[N]})$.  
By Lemma \ref{lem-modulation}, for sufficiently large $n$ (so that $\mathbf{u}_{n} \in B_{\delta}$),  
there exists $\mathbf{c}(t)$ such that the perturbation  
\begin{equation*}
    \mathbf{z}_{n}=\mathbf{v}_{n}-\mathbf{q}^{[N]}(x,t_{n};\mathbf{\Lambda},\mathbf{c}(t_{n}))
\end{equation*}
satisfies  
\begin{equation*}
    \mathbf{z}_{n}\in \mathcal{R}(\mathbf{q}^{[N]}). 
\end{equation*}

Hence,  
\begin{equation*}
    |\mathcal{I}(\mathbf{v}_{n})-\mathcal{I}(\mathbf{q}^{[N]})|
    \leq C\|\mathbf{v}_{n}-\mathbf{q}^{[N]}\|_{H^{N}}
    \leq C\big(\|\mathbf{v}_{n}-\mathbf{u}_{n}(\cdot,t_{n})\|_{H^{N}}
              +\|\mathbf{u}_{n}(\cdot,t_{n})-\mathbf{q}^{[N]}\|_{H^{N}}\big)
    \to 0,
\end{equation*}
which contradicts  
\begin{align*}
    |\mathcal{I}(\mathbf{v}_{n})-\mathcal{I}(\mathbf{q}^{[N]})|
    &\geq (\mathcal{L}\mathbf{z}_{n},\mathbf{z}_{n})-C\|\mathbf{z}_{n}\|_{H^{N}}^{3}\\
    &\geq C_{1}\|\mathbf{z}_{n}\|_{H^{N}}^{2}-C\|\mathbf{z}_{n}\|_{H^{N}}^{3}\\
    &\geq C_{1}\epsilon_{0}^{2}
           -2C_{1}\epsilon_{0}\|\mathbf{v}_{n}-\mathbf{u}_{n}(\cdot,t_{n})\|_{H^{N}}
           +C_{1}\|\mathbf{v}_{n}-\mathbf{u}_{n}(\cdot,t_{n})\|_{H^{N}}^{2}
           -C\|\mathbf{z}_{n}\|_{H^{N}}^{3}
\end{align*}
for sufficiently large $n$ and small $\|\mathbf{z}_{n}\|_{H^{N}}$.  
The estimate for the derivative of the scattering parameters is given by \eqref{der-con}.
\end{proof}

\section{Spectral analysis and nonlinear stability of CmKdV solitons}\label{sec-stability-CmKdV}
For the CmKdV equations, the nonlinear stability analysis differs from that for the CNLS equations. 
Since the CmKdV equations are real-valued (all coefficients are real) and the potentials are real-valued functions, 
additional symmetries appear in the Lax pair. In addition to the symmetries in \eqref{sym-U-V}, 
the reality of the potential $\mathbf{Q}$ implies that the pair $(\mathbf{U},\mathbf{V})$ also satisfies
\begin{equation}\label{sym-U-V-more}
    \mathbf{U}(\lambda)=\mathbf{U}^{*}(-\lambda^{*}),\quad 
    \mathbf{V}(\lambda)=\mathbf{V}^{*}(-\lambda^{*}). 
\end{equation}
Consequently, the point spectrum of the Lax operator $\mathcal{L}_{s}$ is symmetric with respect to 
both the real and imaginary axes. 
Further analysis can be found in \cite{ling2023stability}. 

In the following, we work in the space $L^{2}(\mathbb{R},\mathbb{R}^{2})$ rather than $\mathrm{X}$, 
since the potential in the CmKdV equations is real-valued. For clarity of notation, we add a superscript to 
the FMS, the Darboux matrix, and the potentials. 

\subsection{Darboux transformation for the CmKdV equations}\label{DT-mkdv}
There are two cases for the Darboux transformation:  
we can either add a pair of spectral parameters $\lambda_{k},-\lambda_{k}^{*}$ for $\lambda_{k}\in \mathbb{C}^{+}$,  
or a single spectral parameter $\lambda_{k}\in \mathbb{C}^{+}\cap \ii \mathbb{R}$.  
This is due to the fact that the point spectrum of the Lax operator for the CmKdV equations is symmetric 
with respect to both the real and imaginary axes.  
By symmetry \eqref{sym-U-V-more}, if $\phi(\lambda)$ is an eigenfunction of the Lax operator with eigenvalue $\lambda$, 
then $\phi(\lambda)^{*}$ is an eigenfunction with eigenvalue $-\lambda^{*}$. 

Let $N_{1},N_{2}$ be positive integers with $N_{1}+N_{2}=N$, and set $\tilde{N}=N+N_{1}$.  
The Darboux transformation can then be stated as follows:

\begin{prop}\label{DT-N-cmkdv}
 For the Lax pair \eqref{lax-U}--\eqref{lax-V} with $\mathbf{V}=\mathbf{V}_{\mathrm{CmKdV}}$, 
 let $\lambda_{k}\in\mathbb{C}^{++}=\{\lambda\in\mathbb{C}^{+}:\mathrm{Re}\lambda>0\}$ for $k=1,2,\dots,N_{1}$,  
 and $\lambda_{k}\in \mathbb{C}^{+}\cap \ii \mathbb{R}$ for $k=N_{1}+1,\dots,N$,  
with all $\lambda_{k}$ distinct.  
Let $|\mathbf{y}_{k}\rangle$ denote the eigenfunction of the Lax operator with eigenvalue $\lambda_{k}$.  
Then the $N$-fold Darboux transformation is given by
    \begin{equation*}
        \tilde{\mathbf{D}}_{r}^{[N]}(\lambda;x,t)
        =\mathbb{I}_{3}-\sum_{k=1}^{\tilde{N}}
        \frac{\lambda_{k}-\lambda_{k}^{*}}{\lambda-\lambda_{k}^{*}}
        |\mathbf{x}_{k}\rangle \langle \mathbf{y}_{k}|,
    \end{equation*}
 with $\lambda_{k}=-\lambda_{k-N}^{*}$ and $|\mathbf{y}_{k}\rangle=|\mathbf{y}_{k-N}\rangle^{*}$ for 
    $k=N+1,N+2,\dots,\tilde{N}$. 

 The new fundamental solution is
    \begin{equation*}
        \tilde{\mathbf{\Phi}}^{[N]}_{r}(\lambda;x,t)
        =\tilde{\mathbf{D}}_{r}^{[N]}(\lambda;x,t)\tilde{\mathbf{\Phi}}^{[0]}(\lambda;x,t),
    \end{equation*}
 which satisfies the Lax pair \eqref{lax-U}--\eqref{lax-V} with the corresponding potential 
    \begin{equation*}
        \tilde{\mathbf{Q}}^{[N]}=\tilde{\mathbf{Q}}^{[0]}
        +2\ii\sigma_{3}\sum_{k=1}^{\tilde{N}}
        (\lambda_{k}-\lambda_{k}^{*})
        (|\mathbf{x}_k\rangle \langle \mathbf{y}_k|)^{\perp}, 
    \end{equation*}
 where $|\mathbf{x}_k\rangle$ and $|\mathbf{y}_k\rangle$ are three-component vectors,  
 $\langle \mathbf{x}_k|=(|\mathbf{x}_k\rangle)^{\dag}$, and 
 $\langle \mathbf{y}_k|=(|\mathbf{y}_k\rangle)^{\dag}$.  
 They are related by 
        \begin{equation*}
 (|\mathbf{y}_1\rangle, |\mathbf{y}_2\rangle,\dots, |\mathbf{y}_{N+N_{1}}\rangle)
 =(|\mathbf{x}_1\rangle, |\mathbf{x}_2\rangle,\dots, |\mathbf{x}_{N+N_{1}}\rangle)\mathbf{M}
        \end{equation*}
        where 
        \begin{equation*}
            \mathbf{M}=\left(
        \frac{\lambda_{k}-\lambda_{k}^{*}}{\lambda_{l}-\lambda_{k}^{*}}\langle \mathbf{y}_k|\mathbf{y}_l\rangle
 \right)_{1\leq k,l\leq N+N_{1}}.
        \end{equation*}
\end{prop}

For the zero potential, the FMS is 
\begin{equation*}
    \tilde{\mathbf{\Phi}}^{[0]}={\rm e}^{\ii \lambda(x+4\lambda^{2}t)\sigma_{3}},
\end{equation*}
and the vectors $|\mathbf{y}_k\rangle$ are given by 
\begin{equation*}
 |\mathbf{y}_k\rangle=\tilde{\mathbf{\Phi}}^{[0]}(\lambda_{k};x,t)c^{[k]}
 ={\rm e}^{{\rm i}\lambda_{k}(x+4\lambda_{k}^{2}t)\sigma_{3}}\begin{pmatrix}
        1 \\ \mathbf{c}_{k}
    \end{pmatrix},\quad k=1,2,\cdots, N,
\end{equation*}
in Proposition \ref{DT-N-cmkdv}, where 
\begin{equation*}
    \mathbf{c}_{k}=(c_{1k},c_{2k})^{T}\in\mathbb{C}^{2}\setminus \{(0,0)\}
\end{equation*}
for $k=1,2,\dots,N_{1}$, and 
\begin{equation*}
    \mathbf{c}_{k}=(c_{1k},c_{2k})^{T}\in\mathbb{R}^{2}\setminus \{(0,0)\}
\end{equation*}
for $k=N_{1}+1,N_{1}+2,\cdots,N$.  
To remove the singularity of the point spectrum with respect to the spectral parameter $\lambda$, 
we consider the Darboux transformation
\begin{equation*}
    \tilde{\mathbf{D}}^{[N]}(\lambda;x,t)=\mathcal{P}(\lambda)\tilde{\mathbf{D}}_{r}^{[N]}(\lambda;x,t),
\end{equation*}
where 
\begin{equation}
    \mathcal{P}(\lambda)=\prod_{k=1}^{N_{1}}(\lambda-\lambda_{k}^{*})(\lambda+\lambda_{k})
    \prod_{k=N_{1}+1}^{N}(\lambda-\lambda_{k}),
\end{equation}
by abuse of notation.  

\subsection{Variational characterization and squared eigenfunctions}\label{squ-eig-cmkdv}

Unlike the CNLS equations, due to the symmetry \eqref{sym-U-V-more}, the function $a(\lambda)$ satisfies
\begin{equation*}
 a(\lambda)=a(-\lambda^{*})^{*},
\end{equation*}
hence all momentum-type conserved quantities vanish, i.e., 
\begin{equation*}
    \mathcal{H}_{2n-1}=0,\quad n\geq 1. 
\end{equation*}
We define 
\begin{equation}\label{conserved-quantities-cmkdv}
    \tilde{\mathcal{H}}_{n}(\mathbf{q})=\mathcal{H}_{2n}(\mathbf{q},\mathbf{q}),
\end{equation}
where $\mathcal{H}_{2n}=\mathcal{H}_{2n}(\mathbf{q},\mathbf{q}^{*})$ is given in 
\eqref{conserved-quantities}. 
Here, in contrast to \eqref{conserved-quantities} where $\mathbf{q}$ is complex-valued, 
we restrict to the real case, so that $\mathbf{q}^{*}=\mathbf{q}$, and the functional reduces to 
\(\mathcal{H}_{2n}(\mathbf{q},\mathbf{q})\).
The $N$-soliton solutions of the CmKdV equations satisfy a semi-linear ODE of order $\tilde{N}$ with real coefficients, which represents the critical points of a Lyapunov functional similar to Lemma \ref{lem-critical-point}, with 
spectral parameters $\lambda_{k},-\lambda_{k}^{*}$ for $k=1,2,\cdots, N_{1}$ and $\lambda_{k}$ for $k=N_{1}+1,N_{1}+2,\cdots,N$.  
The Lyapunov functional is given by 
\begin{equation}\label{Lya-cmkdv}
    \tilde{\mathcal{I}}=\sum_{n=0}^{\tilde{N}}\tilde{\mu}_{n}\tilde{\mathcal{H}}_{n},
\end{equation}
where the coefficients $\tilde{\mu}_{n}$ are determined by  
\begin{equation*}
    \mathcal{P}(\lambda)\hat{\mathcal{P}}(\lambda)=\sum_{n=0}^{2\tilde{N}}2^{2n-2\tilde{N}}\tilde{\mu}_{n}\lambda^{2n},
\end{equation*}
which is a polynomial of order $\tilde{N}$ in $\lambda^{2}$, and $\hat{\mathcal{P}}(\lambda)=\mathcal{P}^{*}(\lambda^*)$.  
The conserved quantities can be expressed in terms of the spectral parameters using \eqref{H-spectral}:
\begin{equation}\label{H-spectral-cmkdv}
    \tilde{\mathcal{H}}_{n}=\frac{2^{2n+1}}{2n+1}\left(
        2\sum_{k=1}^{N_{1}}\mathrm{Im} \lambda_{k}^{2n+1}+\sum_{k=N_{1}+1}^{N}\mathrm{Im}\lambda_{k}^{2n+1}
 \right),
\end{equation}
and the variation is given by 
\begin{equation}
		\frac{\delta \tilde{\mathcal{H}}_{n}}{\delta\mathbf{q}}=
		-2^{2n}\mathrm{i}\left(
			2\sum_{k=1}^{N_{1}}\left(\lambda_{k}^{2n}\frac{\delta \lambda_{k}}{\delta\mathbf{q}}-
			(\lambda_{k}^{*})^{2n}\frac{\delta \lambda_{k}^{*}}{\delta\mathbf{q}}\right) 
			+\sum_{k=N_{1}+1}^{N}\lambda_{k}^{2n}\left(
			\frac{\delta \lambda_{k}}{\delta\mathbf{q}}-\frac{\delta \lambda_{k}^{*}}{\delta\mathbf{q}}
			\right)
		\right).
\end{equation}
The following lemma holds:
\begin{lem}
    The Lyapunov functional $\tilde{\mathcal{I}}(\mathbf{q})$ in \eqref{Lya-cmkdv} is time-independent.  
    The $[N_{1},N_{2}]$-soliton solutions are the critical points of $\tilde{\mathcal{I}}(\mathbf{q})$, i.e., each $[N_{1},N_{2}]$-soliton solution satisfies
       \begin{equation}\label{ODE-Nsoliton-cmkdv}
           \frac{\delta \tilde{\mathcal{I}}}{\delta \mathbf{q}}(\mathbf{q})=0,
       \end{equation}
    which is a semi-linear $\tilde{N}$-th order ODE.  
    All solutions of \eqref{ODE-Nsoliton-cmkdv} with the boundary condition $\mathbf{q}\to 0$ as $|x|\to\infty$ are $N$-soliton solutions.  
\end{lem}

\begin{proof}
The proof is similar to that of Lemma \eqref{lem-critical-point}, so we omit it.
\end{proof}

In particular, in this case, the second variation of $\tilde{\mathcal{I}}$ is given by
\begin{equation}\label{op-L-cmkdv}
    \tilde{\mathcal{L}}=\sum_{n=0}^{\tilde{N}}\tilde{\mu}_{n}\frac{\delta^{2}\tilde{\mathcal{H}}_{n}}{\delta \mathbf{q}^{2}}. 
\end{equation}
Since the potential is real-valued, it is unnecessary to consider $\tilde{\mathcal{L}}$ in the space $\mathrm{X}$
(unlike the operator $\mathcal{L}$ in \eqref{op-L}, the operator $\tilde{\mathcal{L}}$ does not contain the term $\cdot^{*}$). 
In this case, the operator $\mathcal{L}$ in \eqref{matrix-L} with spectral parameters 
$\lambda_{k}, -\lambda_{k}^{*}$ for $k=1,2,\ldots,N_{1}$ and $\lambda_{k}$ for $k=N_{1}+1,N_{1}+2,\ldots,N$ admits the representation
\begin{equation*}
    \mathcal{L}=\begin{pmatrix}
        \mathcal{L}_{1} & \mathcal{L}_{2} \\
        \mathcal{L}_{2} & \mathcal{L}_{1}
    \end{pmatrix},
\end{equation*}
for real-valued $\mathbf{q}$, and we have
\begin{equation*}
    \tilde{\mathcal{L}}=\mathcal{L}_{1}+\mathcal{L}_{2},
\end{equation*}
by \eqref{conserved-quantities-cmkdv} and the definitions of $\mathcal{L}$ and $\tilde{\mathcal{L}}$ in \eqref{op-L-cmkdv}.

We also consider the squared eigenfunctions at $t=0$. 
By a slight abuse of notation, we still denote 
\begin{equation}\label{eigenfunction-t0-cmkdv}
    \mathbf{P}_{\pm i}(\lambda;x)=p_{\pm i}(\tilde{\mathbf{\Phi}}^{[N]}|_{t=0}),\quad 
    \mathbf{S}_{\pm i}(\lambda;x)=s_{\pm i}(\tilde{\mathbf{\Phi}}^{[N]}|_{t=0}). 
\end{equation}
When $t=0$, the FMS of the CmKdV Lax pair coincides with that of the CNLS Lax pair,
\begin{equation*}
    \tilde{\mathbf{\Phi}}^{[0]}|_{t=0}=\mathbf{\Phi}^{[0]}|_{t=0}, 
\end{equation*}
and therefore the squared eigenfunctions and squared eigenfunction matrices can be obtained from \eqref{eigenfunction-t0} with 
spectral parameters $\lambda_{k}, -\lambda_{k}^{*}$ for $k=1,2,\ldots,N_{1}$ and $\lambda_{k}$ for $k=N_{1}+1,N_{1}+2,\ldots,N$. 
The squared eigenfunctions and squared eigenfunction matrices in \eqref{eigenfunction-t0-cmkdv} can thus be expressed using 
\eqref{eigenfunction-t0} as
\begin{equation}\label{eigenfunction-t0-cmkdv-1}
    \mathbf{P}_{\pm i}(\lambda;x)=p_{\pm i}(\mathbf{\Phi}^{[\tilde{N}]}|_{t=0}),\quad 
    \mathbf{S}_{\pm i}(\lambda;x)=s_{\pm i}(\mathbf{\Phi}^{[\tilde{N}]}|_{t=0}),
\end{equation}
where the FMS on the right-hand side is given by \eqref{FMS-Nsoliton}
with spectral parameters $\lambda_{k}, -\lambda_{k}^{*}$ for $k=1,2,\ldots,N_{1}$ and $\lambda_{k}$ for $k=N_{1}+1,N_{1}+2,\ldots,N$. 

Under the symmetry \eqref{sym-U-V-more}, 
\begin{equation*}
    \tilde{\mathbf{\Phi}}^{[N]}(\lambda;x,t)=(\tilde{\mathbf{\Phi}}^{[N]})^{*}(-\lambda^{*};x,t),
\end{equation*}
and hence 
\begin{equation}\label{symmetric-squared-eigenfunctions-more}
    \mathbf{P}_{\pm i}(\lambda;x)=\mathbf{P}_{\pm i}^{*}(-\lambda^{*};x),\quad 
    \mathbf{S}_{\pm i}(\lambda;x)=\mathbf{S}_{\pm i}^{*}(-\lambda^{*};x),
\end{equation}
which implies that the squared eigenfunctions at $-\lambda^{*}$ are obtained by taking the complex conjugate of those at $\lambda$. 

For the CmKdV equations, consider the operator 
\begin{equation}
    \mathcal{C}': L^{2}(\mathbb{R},\mathbb{C}^{4})\to L^{2}(\mathbb{R},\mathbb{C}^{2}):\quad 
    \begin{pmatrix}
        \mathbf{g}\\
        \mathbf{h}
    \end{pmatrix}\mapsto \mathbf{g}+\mathbf{h}.
\end{equation}
The squared eigenfunctions are given by 
\begin{equation*}
    \tilde{\mathbf{S}}_{\pm i}(\lambda)=\mathcal{C}' \mathbf{S}_{\pm i}(\lambda).
\end{equation*}
By part (a) of Theorem \ref{thm-spectrum-L-tildeL}, the operator $\mathcal{L}$ has $2\tilde{N}$ negative eigenvalues. 
The number of negative eigenvalues of the operator $\tilde{\mathcal{L}}$ can also be obtained using the Krein signature. 
For the essential spectrum, define 
\begin{align}
    \mathsf{E}'_{ess}=&\{ \tilde{\mathbf{S}}_{\pm i}(\lambda;x): i=1,2, \ \lambda\in\sigma_{ess}(\mathcal{L}_{s}) \}.\label{E-ess-cmkdv}
\end{align}
For the point spectrum, define
\begin{equation}\label{E-point-cmkdv}
		\begin{split}
			\mathsf{E}_{point}'=&\{\tilde{\mathbf{S}}_{1}(\lambda_{k};x),\tilde{\mathbf{S}}_{-2}(\lambda_{k};x),\tilde{\mathbf{S}}_{2}(\lambda_{k}^{*};x),
			\tilde{\mathbf{S}}_{-1}(\lambda_{k}^{*};x):k\notin\Gamma\}\cup
			\\
			&\{\tilde{\mathbf{S}}_{2}(\lambda_{k};x),\tilde{\mathbf{S}}_{-1}(\lambda_{k};x),\tilde{\mathbf{S}}_{1}(\lambda_{k}^{*};x),
			\tilde{\mathbf{S}}_{-2}(\lambda_{k}^{*};x):k\in\Gamma\}, \\
		\end{split}
	\end{equation}
	and
	\begin{equation}\label{hat-E-point-cmkdv} 
		\begin{split}
			\hat{\mathsf{E}}_{point}'=\{\tilde{\mathbf{S}}_{1,\lambda}(\lambda_{k};x),\tilde{\mathbf{S}}_{-1,\lambda}(\lambda_{k}^{*};x): k\notin\Gamma\}\cup\{\tilde{\mathbf{S}}_{2,\lambda}(\lambda_{k};x),\tilde{\mathbf{S}}_{-2,\lambda}(\lambda_{k}^{*};x): k\in\Gamma\}.
		\end{split}
\end{equation}

The key point in the proof of nonlinear stability is to determine the negative directions and the kernel of the operator $\tilde{\mathcal{L}}$.  
The completeness of the squared eigenfunctions and their orthogonality will be used in this section with slight modifications for the case of the CmKdV equations.  

Based on Theorem \ref{thm-orthogonal} in Section \ref{sec-stability-CNLS}, the orthogonality relations for the squared eigenfunctions have been established.  
The orthogonality in the sets \eqref{E-ess-cmkdv}, \eqref{E-point-cmkdv} and \eqref{hat-E-point-cmkdv} follows from that of \eqref{eigenfunction-t0-cmkdv-1} on the Lax spectrum.  
By the completeness of the squared eigenfunctions we determine the number of negative eigenvalues of $\tilde{\mathcal{L}}$.  

Combining the kernel and the negative directions of $\tilde{\mathcal{L}}$ yields the nonlinear stability of the $N$-soliton solutions by the standard arguments in Section \ref{sec-stability-CNLS}.  
The positive Krein signature of the matrix $\mathbf{H}$ in \eqref{matrix-H} with spectral parameters $\lambda_{k}, -\lambda_{k}^{*}$ for $k=1,2,\dots, N_{1}$ and $\lambda_{k}$ for $k=N_{1}+1, N_{1}+2, \dots, N$ equals the number of negative eigenvalues of $\tilde{\mathcal{L}}$.  
The remaining proof is analogous to that in Section \ref{sec-stability-CNLS} and is omitted.

\subsection{The spectrum of $\tilde{\mathcal{L}}$}
For the essential spectrum part we define the set $\mathsf{E}_{ess}'^{+}$, analogous to $\mathsf{E}_{ess}^{+}$, as
\begin{equation*}
    \mathsf{E}_{ess}'^{+}=\{ \mathrm{Re}\tilde{\mathbf{S}}_{i}(\lambda;x),\mathrm{Im}\tilde{\mathbf{S}}_{i}(\lambda;x):i=1,2, \ \lambda\in\sigma_{ess}(\mathcal{L}_{s}) \}. 
\end{equation*}
For the point spectrum part there are slight differences:
\begin{equation*}
    \begin{split}
        &\mathsf{E}_{point}'^{+}=
        \\&\{\mathrm{Re}\tilde{\mathbf{S}}_{1}(\lambda_{k};x),\mathrm{Im}\tilde{\mathbf{S}}_{1}(\lambda_{k};x), 
    \mathrm{Re}\tilde{\mathbf{S}}_{2}(\lambda_{k}^{*};x),\mathrm{Im}\tilde{\mathbf{S}}_{2}(\lambda_{k}^{*};x):k=1,2,\dots,N_{1},k\notin\Gamma\}
 \\ &\cup\{\tilde{\mathbf{S}}_{1}(\lambda_{k};x), \tilde{\mathbf{S}}_{2}(\lambda_{k}^{*};x):k=N_{1}+1,\dots,N,k\notin\Gamma\}
 \\ &\cup\{\mathrm{Re}\tilde{\mathbf{S}}_{2}(\lambda_{k};x),\mathrm{Im}\tilde{\mathbf{S}}_{2}(\lambda_{k};x),
    \mathrm{Re}\tilde{\mathbf{S}}_{1}(\lambda_{k}^{*};x),\mathrm{Im}\tilde{\mathbf{S}}_{1}(\lambda_{k}^{*};x):k=1,2,\dots,N_{1},k\in\Gamma\}
 \\ &\cup \{\tilde{\mathbf{S}}_{2}(\lambda_{k};x), \tilde{\mathbf{S}}_{1}(\lambda_{k}^{*};x):k=N_{1}+1,\dots,N,k\in\Gamma\}.
    \end{split}
\end{equation*}
Similarly,
\begin{equation*}
    \begin{split}
        \hat{\mathsf{E}}_{point}'^{+}=&\{\mathrm{Re}\tilde{\mathbf{S}}_{1,\lambda}(\lambda_{k};x),\mathrm{Im}\tilde{\mathbf{S}}_{1,\lambda}(\lambda_{k};x):k=1,2,\dots,N_{1},k\notin \Gamma\}
 \\ &\cup\{\tilde{\mathbf{S}}_{1,\lambda}(\lambda_{k};x):k=N_{1}+1,\dots,N,k\notin \Gamma\}
 \\ &\cup\{\mathrm{Re}\tilde{\mathbf{S}}_{2,\lambda}(\lambda_{k};x),\mathrm{Im}\tilde{\mathbf{S}}_{2,\lambda}(\lambda_{k};x):k=1,2,\dots,N_{1},k\in \Gamma\}
 \\ &\cup\{\tilde{\mathbf{S}}_{2,\lambda}(\lambda_{k};x):k=N_{1}+1,\dots,N,k\in \Gamma\}. 
    \end{split}
\end{equation*}

For $k=N_{1}+1,\dots,N$ the squared eigenfunctions $\tilde{\mathbf{S}}_{i}(\lambda_{k}),\tilde{\mathbf{S}}_{i}(\lambda_{k}^{*}),
\tilde{\mathbf{S}}_{i,\lambda}(\lambda_{k}),\tilde{\mathbf{S}}_{i,\lambda}(\lambda_{k}^{*})$ are real by the symmetry 
\eqref{symmetric-squared-eigenfunctions-more} and the property $\lambda_{k}\in \mathbb{C}^{+}\cap\ii \mathbb{R}$.  
The completeness of squared eigenfunctions for $N$-soliton solutions of the CmKdV equations is as follows:

\begin{lem}\label{L2-basic-cmkdv}
The space $L^{2}(\mathbb{R},\mathbb{C}^{2})$ decomposes as
    \begin{equation}\label{complete-eigen-cmkdv}
 L^{2}(\mathbb{R},\mathbb{C}^{2})=\mathbb{E}_{ess}'+\mathbb{E}_{point}',
    \end{equation}
where the essential spectrum part is  
    \begin{equation*}
        \mathbb{E}_{ess}'=\mathrm{span}\left\{
            \int_{\mathbb{R}}\omega(\lambda)\tilde{\mathbf{S}}(\lambda;x)\,\mathrm{d}\lambda:\tilde{\mathbf{S}}(\lambda;x)\in
            \mathsf{E}_{ess}',\ \omega(\lambda)\in L^{2}(\mathbb{R},\mathbb{C})
 \right\}
    \end{equation*}
and the point spectrum part is  
    \begin{equation*}
        \mathbb{E}_{point}'=\mathrm{span}\left\{
            \tilde{\mathbf{S}}:\tilde{\mathbf{S}}\in \mathsf{E}_{point}'\cup \hat{\mathsf{E}}_{point}'
 \right\}. 
    \end{equation*}
Moreover, the space $L^2(\mathbb{R},\mathbb{R}^{2})$ decomposes as  
    \begin{equation}\label{complete-eigen-2-cmkdv}
 L^2(\mathbb{R},\mathbb{R}^{2})=\mathbb{E}_{ess}^{\mathrm{X'}} +\mathbb{E}_{point}^{\mathrm{X'}},
    \end{equation}
where  
    \begin{align*}
        \mathbb{E}_{ess}^{\mathrm{X'}}=&\mathrm{span}\left\{
            \int_{\mathbb{R}}\omega(\lambda)\tilde{\mathbf{S}}(\lambda;x)\,\mathrm{d}\lambda:\tilde{\mathbf{S}}(\lambda;x)\in
            \mathsf{E}_{ess}'^{+},\ \omega(\lambda)\in L^{2}(\mathbb{R},\mathbb{R})
 \right\},\\
        \mathbb{E}_{point}^{\mathrm{X'}}=&\mathrm{span}\left\{
            \tilde{\mathbf{S}}:\tilde{\mathbf{S}}\in \mathsf{E}_{point}'^{+}\cup \hat{\mathsf{E}}_{point}'^{+}
 \right\}. 
    \end{align*}
\end{lem}

\begin{proof}
The decomposition \eqref{complete-eigen-cmkdv} follows immediately by taking $\mathcal{C}'$ on both 
sides of \eqref{complete-eigen-1} in Lemma \ref{L2-X-basic}. The basis can be reduced to 
$\mathsf{E}_{ess}'$, $\mathsf{E}_{point}'$, and $\hat{\mathsf{E}}_{point}'$ by the symmetries 
\eqref{symmetric-squared-eigenfunctions} and \eqref{symmetric-squared-eigenfunctions-more}:
\begin{equation}\label{symm-less}
    \begin{split}
        &\tilde{\mathbf{S}}_{\pm i}(-\lambda_{k}^{*})=-\tilde{\mathbf{S}}_{\mp i}(\lambda_{k}^{*}),\quad
    \tilde{\mathbf{S}}_{\pm i}(-\lambda_{k})=-\tilde{\mathbf{S}}_{\mp i}(\lambda_{k}),\quad \\
    &\tilde{\mathbf{S}}_{\pm i,\lambda}(-\lambda_{k}^{*})=-\tilde{\mathbf{S}}_{\mp i,\lambda}(\lambda_{k}^{*}),\quad
    \tilde{\mathbf{S}}_{\pm i,\lambda}(-\lambda_{k})=-\tilde{\mathbf{S}}_{\mp i,\lambda}(\lambda_{k})
    \end{split}
\end{equation}
since  
\begin{equation*}
    \mathcal{C}'\Sigma \mathbf{S}=\mathcal{C}'\mathbf{S}
\end{equation*}
for any function $\mathbf{S}$.  

It is obvious that $\mathcal{C}'(\mathbb{E}_{ess}+\mathbb{E}_{point})\subset L^{2}(\mathbb{R},\mathbb{C}^{2})$.  
Conversely, $L^{2}(\mathbb{R},\mathbb{C}^{2})\subset \mathcal{C}'(\mathbb{E}_{ess}+\mathbb{E}_{point})$,  
since for any $\mathbf{f}\in L^{2}(\mathbb{R},\mathbb{C}^{2})$, the vector $(\mathbf{f}^{T},\mathbf{f}^{T})^{T}$  
lies in $\mathbb{E}_{ess}+\mathbb{E}_{point}$, hence $\mathbf{f}\in \mathcal{C}'(\mathbb{E}_{ess}+\mathbb{E}_{point})$.  

For the decomposition \eqref{complete-eigen-2-cmkdv}, we also take $\mathcal{C}'$ on both sides of \eqref{complete-eigen-2}.  
Using the symmetry \eqref{symm-less} and the identity  
\begin{equation}\label{connect-C-Re}
    \mathcal{C}'\mathcal{C}(\mathbf{S}^{*})=\mathcal{C}'\mathcal{C}\mathbf{S}=2\mathrm{Re}\,\mathcal{C}'\mathbf{S},
\end{equation}
the basis reduces to $\mathsf{E}_{ess}'^{+}$, $\mathsf{E}_{point}'^{+}$, and $\hat{\mathsf{E}}_{point}'^{+}$.  
This completes the proof.
\end{proof}

Now we consider the operator $\tilde{\mathcal{L}}$.  
It remains to obtain the negative Krein symbol for $\tilde{\mathcal{L}}$.  
Since  
\begin{equation}\label{connect-L-hat-L}
 (\tilde{\mathcal{L}}\mathcal{C}'\mathbf{S},\mathcal{C}'\mathbf{S}')=
 (\mathcal{L}\mathbf{S},\mathbf{S}')+(\mathcal{L}\mathbf{S},\Sigma\mathbf{S}')
\end{equation}
for $\mathbf{S},\mathbf{S}'\in L^{2}(\mathbb{R},\mathbb{C}^{4})$, and  
\begin{equation}\label{sym-sigma-squared}
    \Sigma\mathbf{S}_{i}(\lambda)=-\mathbf{S}_{-i}(\lambda^{*})^{*}=-\mathbf{S}_{-i}(-\lambda)
\end{equation}
for $\mathbf{S}_{i}$ being the squared eigenfunctions in  
\eqref{symmetric-squared-eigenfunctions} and \eqref{symmetric-squared-eigenfunctions-more},  
the quadratic form $(\tilde{\mathcal{L}}\cdot,\cdot)$ along the decomposition \eqref{complete-eigen-cmkdv}  
can be obtained from $(\mathcal{L}\cdot,\cdot)$ in space $\mathrm{X}$,  
since the Lax spectrum is symmetric with respect to both the real and imaginary axes.  

Now we can complete the proof of part (b) in Theorem \ref{thm-spectrum-L-tildeL}. 
\begin{proof}[Proof of (b) in Theorem \ref{thm-spectrum-L-tildeL}]
 The proof is similar to the proof of (a) in Theorem \ref{thm-spectrum-L-tildeL}. 
 The kernel can be represented by the squared 
 eigenfunctions in $\mathsf{E}_{point}'^{+}$ or the derivative of scattering parameters
    \begin{equation*}
        \begin{split}
            \mathrm{Ker}(\mathcal{L})=&\mathrm{span}\left\{
            \partial_{\mathrm{Re}c_{ik}}\mathbf{q}^{[N]},
            \partial_{\mathrm{Im}c_{ik}}\mathbf{q}^{[N]}: i=1,2 ,\ k=1,2,\cdots,N_{1}
 \right\}\\ &\cup
 \left\{
            \partial_{c_{ik}}\mathbf{q}^{[N]}: i=1,2 ,\ k=N_{1}+1,N_{1}+2,\cdots,N
 \right\}. 
        \end{split}
    \end{equation*}
 Now we show that the number of negative 
 eigenvalues of the quadratic form 
    $(\tilde{\mathcal{L}}\mathbf{f},\mathbf{g})$ in space $\mathrm{span}\{\hat{\mathsf{E}}_{point}'^{+}\}$ is 
    $N_{1}+\lfloor (N_{2}+1)/2 \rfloor $. Without loss of generality, we consider the case $\Gamma=\emptyset$. 
 The function in $\hat{\mathsf{E}}_{point}'^{+}$ have representation
    \begin{equation*}
        \mathrm{Re}\tilde{\mathbf{S}}_{1,\lambda}(\lambda_{k})=\frac{1}{2}\mathcal{C}'\mathcal{C}\mathbf{S}_{1,\lambda}(\lambda_{k}),\quad 
        \mathrm{Im}\tilde{\mathbf{S}}_{1,\lambda}(\lambda_{k})=-\frac{1}{2}\mathcal{C}'\mathcal{C}\ii\mathbf{S}_{1,\lambda}(\lambda_{k})
    \end{equation*}
 by \eqref{connect-C-Re}. Moreover, since $\mathcal{C}\mathbf{\Sigma}=\mathbf{\Sigma}\mathcal{C}$, we have 
    \begin{equation*}
        \begin{split}
            \mathbf{\Sigma}\mathcal{C}\mathbf{S}_{1,\lambda}(\lambda_{k})
 =&\mathcal{C}\mathbf{\Sigma}\mathbf{S}_{1,\lambda}(\lambda_{k})\\
 =&-\mathcal{C}\mathbf{S}_{-1,\lambda}(-\lambda_{k})\\
 =&\mathcal{C}\mathbf{S}_{1,\lambda}(-\lambda_{k}^{*})
        \end{split}
    \end{equation*}
 and 
    \begin{equation*}
        \mathbf{\Sigma}\mathcal{C}\ii\mathbf{S}_{1,\lambda}(\lambda_{k})=-\mathcal{C}\ii\mathbf{S}_{1,\lambda}(-\lambda_{k}^{*})
    \end{equation*}
 by \eqref{sym-sigma-squared} and \eqref{sym-C-squared}. Then 
    \begin{equation*}
        \begin{split}
 (\tilde{\mathcal{L}}\mathrm{Re}\tilde{\mathbf{S}}_{1,\lambda}(\lambda_{k}),\mathrm{Re}\tilde{\mathbf{S}}_{1,\lambda}(\lambda_{k}))
 =&\frac{1}{4}(\tilde{\mathcal{L}}\mathcal{C}'\mathcal{C}\mathbf{S}_{1,\lambda}(\lambda_{k}),\mathcal{C}'\mathcal{C}\mathbf{S}_{1,\lambda}(\lambda_{k}))\\
 =&\frac{1}{4}\left(
 (\mathcal{L}\mathcal{C}\mathbf{S}_{1,\lambda}(\lambda_{k}),\mathcal{C}\mathbf{S}_{1,\lambda}(\lambda_{k}))+
 (\mathcal{L}\mathcal{C}\mathbf{S}_{1,\lambda}(\lambda_{k}),\mathbf{\Sigma}\mathcal{C}\mathbf{S}_{1,\lambda}(\lambda_{k}))
 \right)\\
 =&\frac{1}{4}\left(
 -2\mathrm{Re}A_{k}+
 (\mathcal{L}\mathcal{C}\mathbf{S}_{1,\lambda}(\lambda_{k}),\mathcal{C}\mathbf{S}_{1,\lambda}(-\lambda_{k}^{*}))
 \right)\\
 =&-\frac{1}{2}\mathrm{Re}c_{1k}^2A_{k}
        \end{split}
    \end{equation*}
 for $k=1,2,\cdots,N_{1}$. Similarly, for $k=1,2,\cdots,N_{1}$, we obtain 
    \begin{equation*}
        \begin{split}
 (\tilde{\mathcal{L}}\mathrm{Re}\tilde{\mathbf{S}}_{1,\lambda}(\lambda_{k}),\mathrm{Im}\tilde{\mathbf{S}}_{1,\lambda}(\lambda_{k}))
 =&-\frac{1}{4}(\tilde{\mathcal{L}}\mathcal{C}'\mathcal{C}\mathbf{S}_{1,\lambda}(\lambda_{k}),\mathcal{C}'\mathcal{C}\ii\mathbf{S}_{1,\lambda}(\lambda_{k}))\\
 =&-\frac{1}{4}\left(
 (\mathcal{L}\mathcal{C}\mathbf{S}_{1,\lambda}(\lambda_{k}),\mathcal{C}\ii\mathbf{S}_{1,\lambda}(\lambda_{k}))+
 (\mathcal{L}\mathcal{C}\mathbf{S}_{1,\lambda}(\lambda_{k}),\mathbf{\Sigma}\mathcal{C}\ii\mathbf{S}_{1,\lambda}(\lambda_{k}))
 \right)\\
 =&-\frac{1}{4}\left(
                2\mathrm{Im}A_{k}-
 (\mathcal{L}\mathcal{C}\mathbf{S}_{1,\lambda}(\lambda_{k}),\mathcal{C}\ii\mathbf{S}_{1,\lambda}(-\lambda_{k}^{*}))
 \right)\\
 =&-\frac{1}{2}\mathrm{Im}c_{1k}^2A_{k}
        \end{split}
    \end{equation*}
 and 
    \begin{equation*}
 (\tilde{\mathcal{L}}\mathrm{Im}\tilde{\mathbf{S}}_{1,\lambda}(\lambda_{k}),\mathrm{Im}\tilde{\mathbf{S}}_{1,\lambda}(\lambda_{k}))=\frac{1}{2}\mathrm{Re}c_{1k}^2A_{k}.
    \end{equation*}
 For $k=N_{1}+1,N_{1}+2,\cdots,N$, since $\tilde{\mathbf{S}}_{1,\lambda}(\lambda_{k})$ is real-valued, we have 
    \begin{equation*}
        \begin{split}
 (\tilde{\mathcal{L}}\tilde{\mathbf{S}}_{1,\lambda}(\lambda_{k}),\tilde{\mathbf{S}}_{1,\lambda}(\lambda_{k}))
 =&\frac{1}{4}\left(
 -2\mathrm{Re}A_{k}+
 (\mathcal{L}\mathcal{C}\mathbf{S}_{1,\lambda}(\lambda_{k}),\mathcal{C}\mathbf{S}_{1,\lambda}(-\lambda_{k}^{*}))
 \right)\\
 =&\frac{1}{4}\left(
 -2\mathrm{Re}A_{k}+
 (\mathcal{L}\mathcal{C}\mathbf{S}_{1,\lambda}(\lambda_{k}),\mathcal{C}\mathbf{S}_{1,\lambda}(\lambda_{k}))
 \right)\\
 =&-\mathrm{Re}c_{1k}^2A_{k}\\
 =&-c_{1k}^2A_{k}
        \end{split}
    \end{equation*}
 since $A_{k}$ is real number. 
 Hence, the quadratic form 
    \begin{equation}\label{X-negative-matrix-cmkdv}
 (\mathcal{L}\mathbf{f},\mathbf{g})=\mathrm{diag}(\tilde{\mathbf{A}}_{1},\tilde{\mathbf{A}}_{2},\cdots,\tilde{\mathbf{A}}_{N_{1}},
 -c_{1,N_{1}+1}^{2}A_{N_{1}+1},-c_{1,N_{1}+2}^{2}A_{N_{1}+2},\cdots,-c_{1N}^{2}A_{N_{1}+1})
    \end{equation}
 where 
    \begin{equation*}
        \tilde{\mathbf{A}}_{k}=\frac{1}{2}\begin{pmatrix}
 -\mathrm{Re}c_{1k}^{2}A_{k} & -\mathrm{Im}c_{1k}^{2}A_{k} \\
 -\mathrm{Im}c_{1k}^{2}A_{k} & \mathrm{Re}c_{1k}^{2}A_{k}
        \end{pmatrix}. 
    \end{equation*}
 Since $c_{1k}\in\mathbb{R}$ and 
    \begin{equation*}
        \begin{split}
 A_{k}=&2^{2N}\ii(\mathcal{P}_{\lambda}(\lambda_{k})\mathcal{P}(\lambda_{k}))^{3}\\
 =&2^{2N}\ii\left(\left.(\mathcal{P}(\lambda)\mathcal{P}(\lambda))_{\lambda}\right|_{\lambda=\lambda_{k}}\right)^{3}\\
 =&2^{2N+3}b_{k}^{3}\prod_{n=1}^{N_{1}}((-b_{k}^{2}-a_{n}^{2}+b_{n}^{2})^{2}+4a_{n}^{2}b_{n}^{2})^{3}
            \prod_{n=N_{1}+1,n\ne k}^{N}(b_{n}^{2}-b_{k}^{2})^{3}
        \end{split}
    \end{equation*}
 for $k=N_{1}+1,N_{1}+2,\cdots,N$, we reindex $b_{k}$ be $b_{(N_{1}+1)}>b_{(N_{1}+2)}>\cdots>b_{(N)}$, 
 then the number $-A_{(N_{1}+1)},-A_{(N_{1}+3)},\cdots$ are negative and 
    $-A_{(N_{1}+2)},-A_{(N_{1}+4)},\cdots$ are positive. 
 Hence the matrix \eqref{X-negative-matrix-cmkdv} admit $N_{1}+\lfloor (N_{2}+1)/2 \rfloor $ negative eigenvalues 
 and $N-\lfloor (N_{2}+1)/2 \rfloor $ positive eigenvalues.
\end{proof}

\subsection{The reduced Hamiltonian}
Define  
\begin{equation}
    \tilde{\mathbf{H}}=(\tilde{H}_{\sigma \tau}),
\end{equation}
where each element is given by  
\begin{equation*}
    \tilde{H}_{\sigma \tau}=\partial_{\sigma \tau}\tilde{\mathcal{I}}-\sum_{n=0}^{\tilde{N}}
    \partial_{\sigma \tau}(\tilde{\mu}_{n})\tilde{H}_{n},
\end{equation*}
for $\sigma,\tau\in\{a_{k},b_{k}:k=1,2,\ldots,N_{1}\}\cup \{b_{k}:k=N_{1}+1,\ldots,N\}$.  
The matrix $\tilde{\mathbf{H}}$ can be characterized by the following lemma.  

\begin{lem}
 The matrix $\tilde{\mathbf{H}}$ is nondegenerate when the spectral parameters are distinct.  
 Moreover, $\tilde{\mathbf{H}}$ admits $N_{1}+\lfloor (N_{2}+1)/2 \rfloor $ positive eigenvalues  
 and $N-\lfloor (N_{2}+1)/2 \rfloor $ negative eigenvalues. 
\end{lem}

\begin{proof}
 By \eqref{H-spectral-cmkdv}, we have  
 \begin{equation*}
    \begin{split}
        \partial_{a_{k}}\tilde{H}_{n} &= 2^{2n+2}\mathrm{Im}\,\lambda_{k}^{2n}, \\
        \partial_{b_{k}}\tilde{H}_{n} &= 2^{2n+2}\mathrm{Re}\,\lambda_{k}^{2n},
    \end{split}
 \end{equation*}
 for $k=1,2,\ldots,N_{1}$, and  
 \begin{equation*}
    \partial_{b_{k}}\tilde{H}_{n} = 2^{2n+1}(\ii b_{k})^{2n}
 \end{equation*}
 for $k=N_{1}+1,\ldots,N$.  
 Taking the derivative of $\mathcal{P}(\lambda)\hat{\mathcal{P}}(\lambda)$ yields  
 \begin{equation*}
    \begin{split}
        \partial_{a_{k}}(\mathcal{P}\hat{\mathcal{P}})(\lambda) &=
        -4a_{k}(\lambda^{2}-a_{k}^{2}-b_{k}^{2})\frac{\mathcal{P}(\lambda)\hat{\mathcal{P}}(\lambda)}{(\lambda^{2}-\lambda_{k}^{2})(\lambda^{2}-(\lambda_{k}^{*})^{2})}, \\
        \partial_{b_{k}}(\mathcal{P}\hat{\mathcal{P}})(\lambda) &=
        4b_{k}(\lambda^{2}+a_{k}^{2}+b_{k}^{2})\frac{\mathcal{P}(\lambda)\hat{\mathcal{P}}(\lambda)}{(\lambda^{2}-\lambda_{k}^{2})(\lambda^{2}-(\lambda_{k}^{*})^{2})},
    \end{split}
 \end{equation*}
 for $k=1,2,\ldots,N_{1}$, since  
 \begin{equation*}
    (\lambda^{2}-\lambda_{k}^{2})(\lambda^{2}-(\lambda_{k}^{*})^{2}) = (\lambda^{2}-a_{k}^{2}+b_{k}^{2})^{2} + 4a_{k}^{2}b_{k}^{2}.
 \end{equation*}
 Moreover,  
 \begin{equation*}
    \partial_{b_{k}}(\mathcal{P}\hat{\mathcal{P}})(\lambda) =
    2b_{k}\frac{\mathcal{P}(\lambda)\hat{\mathcal{P}}(\lambda)}{(\lambda^{2}-\lambda_{k}^{2})}
 \end{equation*}
 for $k=N_{1}+1,\ldots,N$. 
 Denote 
\begin{equation*}
    J_{k}=\left.\frac{\mathcal{P}(\lambda)\hat{\mathcal{P}}(\lambda)}{(\lambda^{2}-\lambda_{k}^{2})(\lambda^{2}-(\lambda_{k}^{*})^{2})}\right|_{\lambda=\lambda_{k}}
\end{equation*}
for $k=1,2,\ldots,N_{1}$ and  
\begin{equation*}
    J_{k}=\left.\frac{\mathcal{P}(\lambda)\hat{\mathcal{P}}(\lambda)}{(\lambda^{2}-\lambda_{k}^{2})}\right|_{\lambda=\lambda_{k}}
\end{equation*}
for $k=N_{1}+1,\ldots,N$. Hence the matrix 
\begin{equation*}
    \tilde{\mathbf{H}}=2^{2N+2}\begin{pmatrix}
        -\mathrm{Re}\mathbf{J}^{[N_{1}]} & \mathrm{Im} \mathbf{J}^{[N_{1}]} & \mathbf{0} \\
        \mathrm{Im} \mathbf{J}^{[N_{1}]} & \mathrm{Re}\mathbf{J}^{[N_{1}]} & \mathbf{0} \\
        \mathbf{0} & \mathbf{0} & \tilde{\mathbf{J}}^{[N_{2}]}
    \end{pmatrix},
\end{equation*}
where 
\begin{equation*}
    \mathbf{J}^{[N_{1}]}=\mathrm{diag}\left(
        8a_{1}b_{1}(a_{1}^{2}+b_{1}^{2})J_{1},\,8a_{2}b_{2}(a_{2}^{2}+b_{2}^{2})J_{2},\,\ldots,\, 8a_{N_{1}}b_{N_{1}}(a_{N_{1}}^{2}+b_{N_{1}}^{2})J_{N_{1}}
    \right),
\end{equation*}
and 
\begin{equation*}
    \tilde{\mathbf{J}}^{[N_{2}]}=\mathrm{diag}\left(
        b_{N_{1}+1}J_{N_{1}+1},\,b_{N_{1}+2}J_{N_{1}+2},\,\ldots,\,b_{N}J_{N}
    \right).
\end{equation*}
Since 
\begin{equation}
    J_{k}=\prod_{n=1}^{N_{1}}((-b_{k}^{2}-a_{n}^{2}+b_{n}^{2})^{2}+4a_{n}^{2}b_{n}^{2})
            \prod_{n=N_{1}+1,n\ne k}^{N}(b_{n}^{2}-b_{k}^{2})
\end{equation}
for $k=N_{1}+1,N_{1}+2,\ldots,N$, the matrix $\tilde{\mathbf{J}}^{[N_{2}]}$ 
has $\lfloor (N_{2}+1)/2 \rfloor$ positive eigenvalues and 
$N_{2}-\lfloor (N_{2}+1)/2 \rfloor$ negative ones. 
This concludes the proof.
\end{proof}

The remaining proof for stability results on $(N_{1},N_{2})$-soliton solutions is standard, which is 
similar to Section \ref{sec-stability-CNLS}. 
Denote $\Sigma_{sp}=\{a_{k},b_{k}:k=1,2,\cdots,N_{1}\}\cup\{b_{k}:k=N_{1}+1,N_{1}+2,\cdots,N\}$. 
Applying Lemma \ref{spectrum-reduce} to the matrix $\tilde{\mathbf{H}}$, 
we obtain
\begin{equation*}
    \mathrm{n}(\tilde{\mathcal{L}}\tilde{\mathcal{P}})=\mathrm{n}(\tilde{\mathcal{L}})-\mathrm{p}(\tilde{\mathbf{H}})=0
\end{equation*}
where $\tilde{\mathcal{P}}$ is the projection of $L^{2}(\mathbb{R},\mathbb{R}^{2})$ to the space 
\begin{equation*}
    \tilde{\mathrm{X}}_{1}=\mathrm{span}\left\{
        \frac{\delta \tilde{\mathcal{Q}}_{\sigma}}{\delta \mathbf{q}}:\sigma\in \Sigma_{sp}
 \right\}
\end{equation*}
with 
\begin{equation*}
    \tilde{\mathcal{Q}}_{\sigma}=\sum_{n=0}^{\tilde{N}}(\partial_{\sigma}\tilde{\mu}_{n})\tilde{\mathcal{H}}_{n}
\end{equation*}
for $\sigma\in\Sigma_{sp}$. 
Denote 
\begin{equation*}
    \tilde{\mathcal{R}}(\mathbf{q})=\mathrm{Ker}(\tilde{\mathcal{L}})^{\perp}\cap 
    \mathrm{span} \{\mathbf{z}:\tilde{\mathcal{Q}}_{\sigma}(\mathbf{q}+\mathbf{z})=\tilde{\mathcal{Q}}_{\sigma}(\mathbf{q}):\sigma\in \Sigma_{sp}\},
\end{equation*}
by Lemma \ref{core-R}, we obtain 
\begin{equation*}
 (\tilde{\mathcal{L}}\mathbf{z},\mathbf{z})\geq C_{1}\|\mathbf{z}\|_{H^{\tilde{N}}}^{2}-C_{2}\|\mathbf{z}\|_{H^{\tilde{N}}}^{3}
\end{equation*}
if $\|\mathbf{z}\|_{H^{\tilde{N}}}$ is small and $\mathbf{z}\in H^{\tilde{N}}\cap \tilde{\mathcal{R}}(\mathbf{q}^{[N_{1},N_{2}]})$. 
\begin{proof}[Proof of stability for $(N_{1},N_{2})$-soliton in Theorem \ref{thm-stability}]
 For small perturbation $\mathbf{z}$ in $\tilde{\mathcal{R}}(\mathbf{q}^{[N_{1},N_{2}]})$, we obtain the
 corecivity for operator $\tilde{\mathcal{L}}$. In addition, applying Lemma \ref{lem-modulation} to $(N_{1},N_{2})$-soliton 
 solutions with scattering parameters 
    \begin{equation*}
        \mathbf{c}_{i}\in \mathbb{C}^{2}\backslash\{(0,0)\},\, i=1,2,\cdots,N_{1},\quad 
        \mathbf{c}_{i}\in \mathbb{R}^{2}\backslash\{(0,0)\},\, i=N_{1}+1,N_{1}+2,\cdots,N,
    \end{equation*}
 we can find $\tilde{\mathbf{c}_{i}}$ such that the perturbation 
    \begin{equation*}
        \mathbf{u}(x,t)-\mathbf{q}^{[N_{1},N_{2}]}(x,t;\tilde{c}_{i})\in \mathrm{Ker}(\tilde{\mathcal{L}})^{\perp}
    \end{equation*}
 and $|\partial_{t}\tilde{c}_{ij}|$ can be controlled by the norm of perturbation, as in \eqref{der-con}. 
 Then the stability results can be obtained by contradiction, 
similar to the argument for $N$-soliton solutions with a different index $\tilde{N}$, 
and the details are omitted. 
\end{proof}

\subsection*{Acknowledgements}

Liming Ling is supported by the National Natural Science Foundation of China (No. 12471236), 
Guangzhou Science and Technology Plan (No. 2024A04J6245) 
and Guangdong Natural Science Foundation grant (No. 2025A1515011868).
The authors would also like to thank Shengxiong Yang for carefully verifying the calculations and for
his valuable contributions.

\subsection*{Data availability statement}
Data sharing is not applicable to this article as no datasets were generated or
analysed during the current study.

\subsection*{Conflict of interest}
On behalf of all authors, the corresponding author states that there is no conflict of
interest.

\bibliography{Ref-CNLS-stability}

\end{document}